\documentclass[usenames,dvipsnames]{article}
\usepackage{blindtext,inputenc,graphicx,overpic}
\usepackage{amsmath,amssymb,amsthm,bm,mathrsfs,mathtools,enumitem}

\usepackage[scr=boondox]{mathalpha}

\usepackage{multicol}
\usepackage{esint,braket}
\usepackage{fullpage,setspace}
\usepackage{authblk}
\usepackage{verbatim} 
\usepackage{caption}
\usepackage{subcaption}

\usepackage[
    backend=biber,
    backref=true,
    style=alphabetic,
    maxbibnames=99,
  ]{biblatex}
\DefineBibliographyStrings{english}{
  backrefpage={p.},
  backrefpages={p.}
}

\usepackage{tikz}
\usetikzlibrary{arrows.meta,decorations.markings,decorations.pathmorphing}
\usepackage{pgfplots}
\usepgfplotslibrary{polar}

\usepackage{hyperref}
\hypersetup{
    colorlinks=true,
    linkcolor=BrickRed,
    filecolor=red,      
    urlcolor=Periwinkle,
    citecolor=teal,
}
\usepackage{framed}
\definecolor{shadecolor}{rgb}{.6,.8,1}

\theoremstyle{definition}
\newtheorem{defn}{Definition}[section]
\newtheorem{theorem}{Theorem}[section]
\newtheorem{lemma}{Lemma}[section]
\newtheorem{remark}{Remark}[section]
\newtheorem{prop}{Proposition}[section]
\newtheorem{cor}{Corollary}[section]

\DeclareMathOperator{\tr}{tr}

\DeclareMathOperator{\CC}{\mathbb{C}}
\DeclareMathOperator{\RR}{\mathbb{R}}
\DeclareMathOperator{\DD}{\mathbb{D}}

\DeclareMathOperator{\ZZ}{\mathbb{Z}}

\DeclareMathOperator{\OO}{\mathcal{O}}

\DeclareMathOperator*{\Ress}{Res}
\newcommand{\Res}{\displaystyle\Ress}

\numberwithin{equation}{section}
\numberwithin{figure}{section}

\usepackage{todonotes}
\usepackage{lineno}

\usepackage{titlesec}
\titleformat{\section}{\centering\normalfont\scshape}{\thesection .}{1em}{}
\titleformat{\subsection}{\normalfont\scshape}{\thesubsection .}{1em}{}
\titleformat{\subsubsection}{\normalfont\scshape}{\thesubsubsection .}{1em}{}
\usepackage{abstract}

\title{\bf The Ising Model Coupled to 2D Gravity: Critical Partition Function}

\date{}

\author[1]{\scshape Maurice Duits\thanks{\href{mailto:duits@kth.se}{duits@kth.se}}}
\author[1]{\scshape Nathan Hayford\thanks{\href{mailto:nhayford@kth.se}{nhayford@kth.se}}}
\author[2]{\scshape Seung-Yeop Lee\thanks{\href{mailto:lees3@usf.edu}{lees3@usf.edu}}}
\affil[1]{\small\textit{Department of Mathematics, Royal Institute of Technology (KTH), Stockholm, Sweden}}
\affil[2]{\small\textit{Department of Mathematics and Statistics, University of South Florida, Tampa, FL}}

\begin{document}

\maketitle

\begin{abstract}
    We prove that the differential of the log of the partition function for the $2$-matrix model with quartic interactions converges in a certain double-scaling regime to the differential of the $\boldsymbol{\tau}$-function for the $(3,4)$ string equation. This confirms the convergence of the critical Ising model on random surfaces to the $(3,4)$ topological minimal model, which was stated in the works of Douglas and Shenker, Br\'{e}zin and Kazakov, and Gross and Migdal. Our analysis is based on a steepest-descent analysis of a Riemann-Hilbert problem associated to a family of biorthogonal polynomials. New features in the matching problem in the construction of local parametrices appear.
\end{abstract}


\tableofcontents


\section{Introduction}
In this work, we will consider the integral
    \begin{equation}\label{partition-function}
        Z_n(\tau,t,H;N) := \iint \exp \bigg(N \tr \Phi(X,Y;\tau,t,H) \bigg)dXdY,
    \end{equation}
where each integration is carried out over the space of $n\times n$ Hermitian matrices, and
    \begin{equation}
        \Phi(x,y;\tau,t,H):=\tau xy - \frac{1}{2}x^2 - \frac{te^{H}}{4}x^4- \frac{1}{2}y^2 - \frac{te^{-H}}{4}y^4.
    \end{equation}
This integral is called the \textit{partition function} for the $2$-matrix model, and is well-defined when $0<\tau<1$, $t>0$, and $H\in \RR$. For given $N>0$, the partition function is an analytic function of $t$, and as such can be analytically continued to the complex $t$ plane, with a branch point at $t=0$. We are particularly interested in $Z_n$ when $t<0$. We recall from \cite{DHL1} that the \textit{genus $0$ free energy:}
    \begin{equation}
        F_0(\tau,t,H) := \lim_{N\to\infty}\frac{1}{N^2}\log\frac{Z_N(\tau,t,H;N)}{Z_N(\tau,0,0;N)},
    \end{equation}
also admits an analytic continuation to $t<0$, up to a special critical surface (shown in Figures \ref{fig:PhaseDiagram2D}, \ref{fig:PhaseDiagram3D}), and undergoes a $3^{rd}$ order phase transition at the point
    \begin{equation}\label{multicritical-point}
        \tau = \tau_c:= \frac{1}{4}, \qquad t= t_c:=-\frac{5}{72}, \qquad H =H_c:=0,
    \end{equation}
which confirms the conjectured results of Kazakov and Boulatov \cite{Kazakov1,Kazakov2}.
Our main result concerns how, in a special scaling limit around this point, $\log Z_n$ converges to the $\boldsymbol{\tau}$-function for the \textit{$(3,4)$ string equation}, which is a higher order Painlev\'{e}-type equation. 

The relation between matrix integrals like \eqref{partition-function} and Painlev\'{e} equations (which we shall refer to as \textit{critical phenomena}) has been studied extensively for almost 40 years. The interest in this relation stemmed from the connection of these critical phenomena to the minimal models of conformal field theory coupled to topological gravity. In the mid-80s, Kazakov and Boulatov \cite{Kazakov1,Kazakov2} demonstrated that matrix models held promise in predicting what came to be known as the 
Kniznik Polyakov Zamolodchikov (KPZ) scaling relations \cite{KPZ}. Several years later, the seminal works of Douglas and Shenker \cite{DS}, Br\'{e}zin and Kazakov 
\cite{Kazakov3}, and Gross and Migdal \cite{GM1,GM2}, showed how in certain `double scaling' limits, the partition function for a given matrix model will converge 
to a $\boldsymbol{\tau}$-function for the corresponding Painlev\'{e} equation, which could then be thought of as a partition function for the corresponding `topological minimal model'. These ideas were developed more analytically in the works of Fokas, Its, and Kitaev \cite{FIK1,FIK2}, where the importance of the isomonodromy approach was elucidated in the case of the $1$-matrix model. Also introduced in their work was a \textit{Riemann-Hilbert problem} (RHP) for the associated orthogonal polynomials, which allowed for rigorous asymptotic analysis via the Deift-Zhou method \cite{Deift-Zhou}. These results caused a surge in activity in the study of random matrices, Painlev\'{e} equations, and orthogonal polynomials. The techniques introduced in \cite{Deift-Zhou} were applied soon thereafter to prove local universality results for the Hermitian ensembles \cite{DKMVZ,KMVV}.  Using these techniques, analytical analyses related to the appearance of Painlev\'{e} equations, in particular Painlev\'{e} II, were found by Bleher and Its \cite{BI}, and were shown to be universal by Claeys and Kuijlaars \cite{CK}. Results on the Painlev\'{e} I critical point of \cite{FIK1,FIK2} were finally obtained by Duits and Kuijlaars \cite{DK0}, and later on by Bleher and Dea\~{n}o \cite{BleherDeano} (see also the more recent \cite{Bertola-Tovbis, BGM}). However, a rigorous analysis of a critical phenomena in the $2$-matrix model as studied in the early physics literature was never performed, although new critical phenomenon related to Painlev\'{e} II in the $2$-matrix model were indeed studied by Duits and Geudens \cite{DG}.

This kind of analysis also attracted interest in enumerative combinatorics. It can be shown that, as $N\to\infty$, provided $\tau,t,H$ are sufficiently small,
    \begin{equation}\label{free-energy-expansion}
        \frac{1}{N^2}\log \frac{Z_N(\tau,t,H;N)}{Z_N(\tau,0,0;N)} \sim \sum_{g=0}^{\infty} \frac{F_g(\tau,t,H)}{N^{2g}},
    \end{equation}
where 
    \begin{equation}\label{Fg-expansion}
        F_g(\tau,t,H) = \sum_{n\in \mathbb{N}} \mathcal{Z}_{n,g}(\tau,H) \left(\frac{-t\tau}{4(1-\tau^2)^2}\right)^n,\qquad \text{ and }\qquad\mathcal{Z}_{n,g}(\tau,H) = \sum_{\pi\in \Pi_{g,n}} Z_{\pi}(\tau,H),
    \end{equation}
and where $\Pi_{g,n}$ is the set of all \textit{connected, labeled, $4$-regular, genus $g$ maps with $n$ vertices}, and $Z_{\pi}(\tau,H)$ denotes the partition function for the Ising model on the map $\pi$ (see the Introduction of \cite{DHL1} for the exact details of this correspondence). By tuning the parameters $t$ to be critical, and subsequently $\tau$ and $H$, one is able to study the behavior of the critical Ising model on quadrangulations (which are dual to $4$-regular maps) of genus $g$ maps with infinitely many vertices. Combinatorial studies of the critical Ising model on \textit{planar triangulations} with infinite vertices were carried out by the groups of Albenque, M\'{e}nard, and Schaeffer \cite{AMS,AM}, Bernardi and Bousquet-M\'{e}lou \cite{BBM}, and Chen and Turunen \cite{CT1,CT2,Tur}. These groups were able to calculate the critical exponents of this model, but are mostly limited to studying random $\textit{planar}$ maps, i.e. maps of genus $0$. An important aspect of our result is that it involves all genera, and so one should be able obtain results the nonplanar situation from our main theorem.

\begin{figure}[t]
    \centering
    \begin{overpic}[scale=.15]{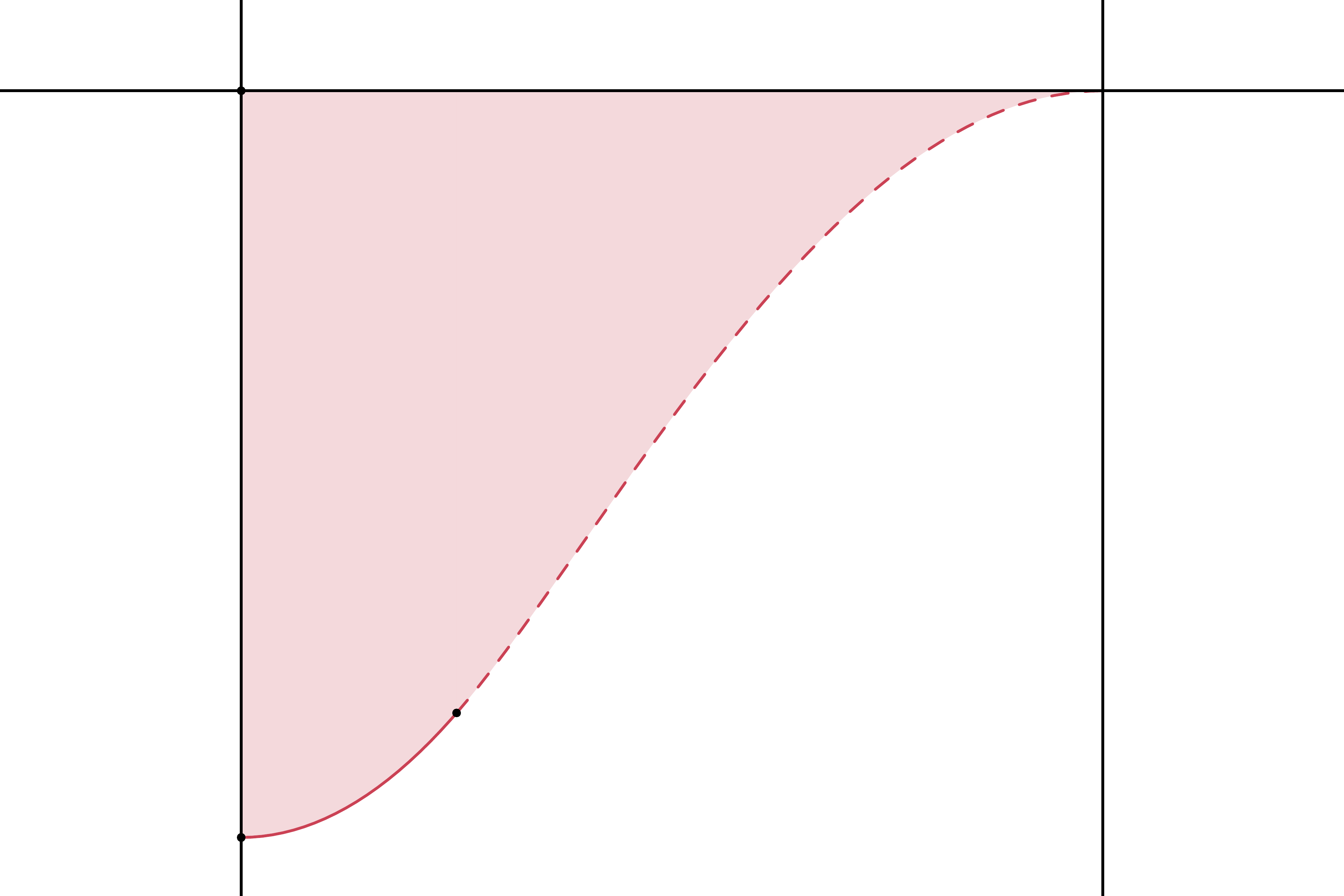}
            \put (79,68) {$\tau = 1$}
            \put (86,61) {\large $\tau$}
            \put (20,64) {\large $t$}
            \put (4,4) {$(0,-\frac{1}{12})$}
            \put (20,15) {$(\frac{1}{4},-\frac{5}{72})$}
            \put (30,6) {\large $t_{low}(\tau) = -\frac{1}{12} + \frac{2}{9}\tau^2$ }
            \put (60,40) {\large $t_{high}(\tau) = -\frac{2}{9} \sqrt{\tau}(\sqrt{\tau} - 1)^2(\sqrt{\tau} +2)$}
    \end{overpic}
    \caption{$H=0$ cross-section of the phase diagram for the genus $0$ free energy $F_0(\tau,t,H)$. The function $F_0(\tau,t,H)$ is analytic in the shaded region, and analyticity breaks down on the dashed and dotted lines. The critical point studied in this work is at $\tau=\frac{1}{4}$, $t=-\frac{5}{72}$, and $H=0$. Note that the critical point of the quartic $1$-matrix model (at $t=-\frac{1}{12}$) is also visible here.}
    \label{fig:PhaseDiagram2D}
\end{figure}

\begin{figure}[t]
    \centering
    \begin{overpic}[scale=.5]{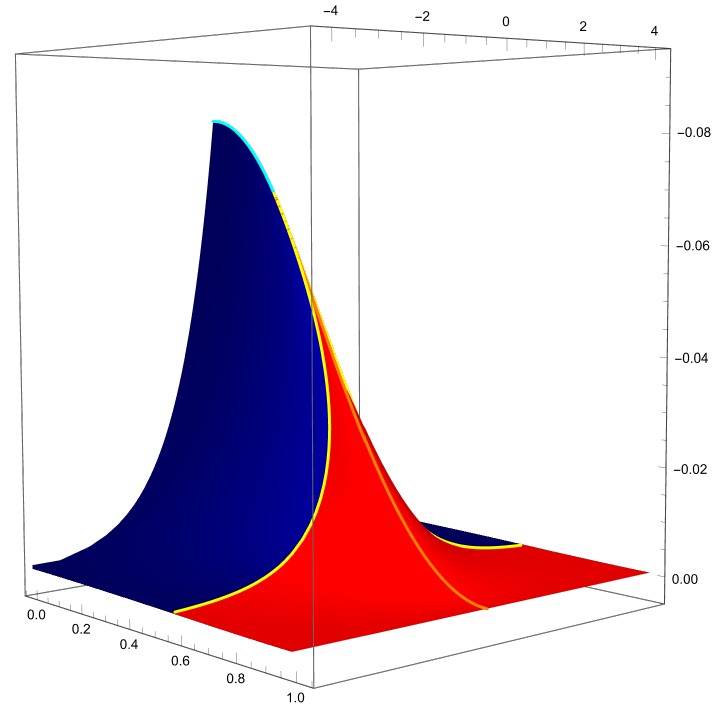}
        \put (12,3) {$\tau$}
        \put (100,50) {$t$}
        \put (65,96) {\small $H$}
        \put (35,78) {\small \textcolor{Cyan}{$\gamma_{low,0}$}}
        \put (70,12) {\small \textcolor{Orange}{$\gamma_{high,0}$}}
        \put (35,25) {\small \textcolor{Yellow}{$\gamma_{b}$}}
    \end{overpic}
    \caption{Phase portrait for the $2$-matrix model with quartic interactions in the $(\tau,t,H)$ plane. The region where $F_0(\tau,t,H)$ admits an analytic continuation lies in the half space where $\tau>0,t<0$, and is bounded by the red and blue surfaces. $F_0(\tau,t,H)$ is continuous on this surface, and in the genus zero region; on the surface, $F$ is analytic away from the low-temperature critical curve, shown here in light blue. The multicritical point $\left(\frac{1}{4},-\frac{5}{72},0\right)$ is the meeting point of the high-temperature critical curve (orange), the two phase boundaries joining the low and high temperature phases (yellow), and the low-temperature critical curve.}
    \label{fig:PhaseDiagram3D}
\end{figure}
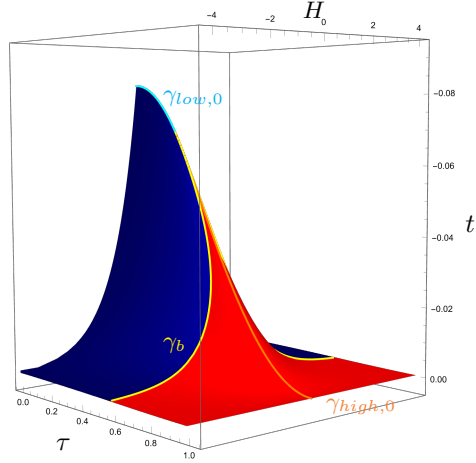

\subsection{$(3,4)$ string equation \& tau-function}

Our main result is given in terms of a special solution to the \textit{$(3,4)$ string equation}. This equation reads
    \begin{equation} \label{string-equation}
        \begin{cases}
             0 = \frac{1}{2}V'' - \frac{3}{2}UV + \frac{5}{2}\eta V + \mu,\\
             0 = \frac{1}{12} U^{(4)} -\frac{3}{4}U''U -\frac{3}{8}(U')^2+\frac{3}{2}V^2 + \frac{1}{2}U^3 - \frac{5}{12}\eta \left(3U^2 - U''\right) + \nu,
        \end{cases}
    \end{equation}
where $U = U(\eta,\mu,\nu)$, $V= V(\eta,\mu,\nu)$ are meromorphic functions of the variables $(\eta,\mu,\nu)$, and $' = \frac{\partial}{\partial \nu}$. This string equation is augmented by a collection of compatible PDEs, which govern the dependence of $U,V$ on $\eta,\mu$:
    \begin{align}
            \frac{\partial U}{\partial \mu} &= -2V', \label{U_mu}\\
            \frac{\partial V}{\partial \mu} &= \frac{1}{6}U''' - UU', \label{V_mu}\\
            \frac{\partial U}{\partial \eta} &=\frac{\partial}{\partial \nu}\left[-\frac{1}{6}UU'' + \frac{1}{8}(U')^2 + \frac{1}{4}U^3 - 
            \frac{1}{2}V^2 - \frac{5}{9}\eta\left(3U^2-U''\right) + \frac{4}{3}\nu\right],\label{U_eta}\\
            \frac{\partial V}{\partial \eta} &= \frac{\partial}{\partial \nu}\left[ \frac{1}{12}U''V - \frac{1}{4}U'V' + \frac{5}{16}U^2V - \left(\frac{5}{3}\eta + \frac{1}{4}U\right)^2V - \mu U\right]\label{V_eta}.
        \end{align}
Although these equations look quite complicated, it has been shown that:
\begin{prop}[\cite{DHL2}, Theorem 1.1] \label{AAA}
    Under the identification
        \begin{equation*}
            \nu = t_1, \qquad \qquad \mu=t_2,\qquad \qquad\eta = t_5,
        \end{equation*}
    there exist polynomials $H_1,H_2,H_5$ in the variables $(P_1,P_2,P_3,Q_1,Q_2,Q_3;t_1,t_2,t_5)$ such that the equations
        \begin{equation}
            \frac{\partial H_a}{\partial P_k} = \frac{\partial Q_k}{\partial t_a},\qquad\qquad -\frac{\partial H_a}{\partial Q_k} = \frac{\partial P_k}{\partial t_a},
        \end{equation}
    $a=1,2,5$, $k=1,2,3$ are equivalent to the equations \eqref{string-equation}, \eqref{U_mu}--\eqref{V_eta}. Moreover, for any $a,b=1,2,5$,
    \begin{align}
        0&= \{H_a,H_b\}:= \sum_{k=1}^{3}\left(\frac{\partial H_a}{\partial P_k}\frac{\partial H_b}{\partial Q_k}-\frac{\partial H_a}{\partial Q_k}\frac{\partial H_b}{\partial P_k}\right),\\
        0&=\frac{\partial H_a}{\partial t_b} - \frac{\partial H_b}{\partial t_a}.
    \end{align}
\end{prop}
In other words, the $(3,4)$ string equation can be characterized as a completely integrable (non-autonomous) Hamiltonian system. The explicit form of these Hamiltonians is given in Appendix \ref{Appendix-A}, both in the Darboux coordinates and evaluated on a solution to the string equation. An important observation is that one can define the following \textit{closed} differential
    \begin{equation}
        {\bf d}\log\boldsymbol{\tau}(\eta,\mu,\nu) := H_1 d\nu + H_2 d\mu + H_5d\eta,
    \end{equation}
whose primitive we refer to as the $\boldsymbol{\tau}$-function. This function is well-defined up to a multiplicative constant. 

Another important characterization of Equation \eqref{string-equation} is as the set of isomonodromy deformations of a certain linear differential equation with rational coefficients, or alternatively as the solution to a certain Riemann-Hilbert problem, which is given in Appendix \ref{Appendix-A}. This characterization was demonstrated in \cite{DHL2}, and in \cite{Nathan1}, a special solution to this equation was shown to exist asymptotically for a particular set of \textit{Stokes multipliers}:
    \begin{align}
        s_1 &= 0,\qquad s_2 = -1,\qquad s_3 = 0,\qquad s_4 = 0,\qquad s_5 = 1,\qquad s_6 = -1, \qquad s_7 = 0,\nonumber\\
        s_{-1} &= 0,\qquad s_{-2} = 1,\qquad s_{-3} = -1,\qquad s_{-4} = 0,\qquad s_{-5} = 0,\qquad s_{-6} = 1,\qquad s_{-7} = 0.
    \end{align}
More precisely, the solution to this model RHP with the above Stokes data (which indeed appears in the present work) was shown to be asymptotically free of singularities in a particular sector of the $(\eta,\mu,\nu)$-space in \cite{Nathan1}. In a sense, the solution studied there is analogous to the \textit{tritronqu\'{e}e} solution of Painlev\'{e} I:
    \begin{equation}
        u''(x) = 6u(x)^2 + x,
    \end{equation}
which appears in the study of the $1$-matrix model. Indeed, it was also demonstrated in \cite{Nathan1} that $\boldsymbol{\tau}(\eta,\mu,\nu)$ limits to the tau function for this solution of the Painlev\'{e} I equation in a certain regime.

The main result of this work can now be stated slightly more precisely: \textit{in a certain double-scaling regime, the partition function \eqref{partition-function} converges to $\boldsymbol{\tau}(\eta,\mu,\nu)$.} Of course, since we have only defined $\boldsymbol{\tau}(\eta,\mu,\nu)$ up to a multiplicative constant, we can only state this convergence in terms of the differential of the partition function.

\begin{remark} \textit{The (3,4) minimal model coupled to gravity.}
It is by now folklore that the Ising critical point is characterized by the $(3,4)$ minimal model of conformal field theory \cite{DFMS}. This field theory has central charge $c=\frac{1}{2}$, and contains operators $\mathbb{I}$, $\epsilon$, and $\sigma$,
which we refer to as the \textit{identity}, \textit{thermal}, and \textit{spin} operators, respectively. These operators have scaling dimensions
    \begin{equation}
        \Delta_{\mathbb{I}} = 0, \qquad\qquad \Delta_{\epsilon} = \frac{1}{2},\qquad\qquad \Delta_{\sigma} =\frac{1}{16}.
    \end{equation}
The KPZ formula \cite{KPZ} provides a formula for how the scaling dimensions change upon coupling to quantum gravity:
    \begin{equation}
        \hat{\Delta}_{\bullet} = \frac{\sqrt{1-c+24\Delta_{\bullet}} -\sqrt{1-c}}{\sqrt{25-c} -\sqrt{1-c}},\qquad\qquad \bullet \in
        \{\mathbb{I},\epsilon, \sigma \}
    \end{equation}
The ``dressed'' weights are
    \begin{equation}
        \hat{\Delta}_{\mathbb{I}} = 0, \qquad\qquad \hat{\Delta}_{\epsilon} = \frac{2}{3},\qquad\qquad \hat{\Delta}_{\sigma} = \frac{1}{6}.
    \end{equation}
The parameters $\eta,\mu,\nu$ `couple' to the operators $\epsilon, \sigma, \mathbb{I}$, respectively, in the sense that correlation functions of the fields $\epsilon,\sigma,\mathbb{I}$ may be computed by taking appropriate derivatives of the $\boldsymbol{\tau}$-function.
\end{remark}

\subsection{Main result}
Our main result is that, after suitable regularization by a polynomial $\log Z_{reg}$, the differential of the partition function \eqref{partition-function} converges to the $\boldsymbol{\tau}$-differential for the $(3,4)$ string equation. The explicit form of $\log Z_{reg}$ is not so important to us; its explicit form and a brief explanation of its origin are given in Appendix \ref{appendix-regularization}. We now define the scaling regime we will consider.

\begin{defn}\label{scaling-variables-def}
    Put\footnote{If one expands $t_{high}(\tau)$ in Figure \ref{fig:PhaseDiagram2D} around $\tau = \tau_c$, one obtains that $t_{high}(\tau) = -\frac{5}{72} + \delta t_{high}(\tau-\tau_c) + \OO((\tau-\tau_c)^4)$.}
    \begin{equation}
        \delta t_{high}(s) := \frac{1}{9}s + \frac{2}{9}s^2-\frac{8}{9}s^3.
    \end{equation}
    Let $\eta,\mu,\nu\in \RR$ be outside the singularity set of $\log \boldsymbol{\tau}(\eta,\mu,\nu)$.
    We define the variables $t,\tau,H$ to be
        \begin{align}
            t &= t(\eta,\nu;n) := t_c + \delta t_{high}\left(-\frac{c_{\eta}\eta}{n^{2/7}} \right) - \frac{1}{9}c_{\nu} \frac{\kappa \nu}{n^{6/7}},\\
            \tau &= \tau(\eta,\nu;n):= \tau_c -c_{\eta} \frac{\eta}{n^{2/7}}  + c_{\nu}\frac{(1-\kappa)\nu}{n^{6/7}},\\
            H &= H(\mu;n) := H_c + c_{\mu} \frac{\mu}{n^{5/7}},
        \end{align}
    where $c_{\eta} = \frac{2^{6/7}5^{9/7}}{48},c_{\mu} = \frac{2^{1/7}5^{5/7}}{12}, c_{\nu} = \frac{2^{4/7}5^{6/7}}32$ are some positive constants, and $\kappa$ is a real number.
\end{defn}

\begin{remark}
    The `exploration parameters' $\eta,\mu$ and $\nu$ correspond to deviations away from the critical point along the curve $\gamma_{high,0}$ (see Figures \ref{fig:PhaseDiagram3D}, \ref{fig:PhaseDiagram2D}), along the critical surface in the direction of nonzero magnetic field, and nontangentially to the critical surface, respectively. The parameter $\kappa$ in the above definition represents the freedom in choosing a non-tangential approach to the critical point; it does not play an essential role in the final result, as we will see shortly. Observe also that we could replace $\delta t_{high}\left(-\frac{c_{\eta}\eta}{n^{2/7}} \right)$ with simply $-\frac{c_{\eta}\eta}{n^{2/7}}$; the result would be that the dependence of $U,V$ on the variables $\eta,\mu,\nu$ in our main theorem would be nonlinear; cf. the statement of \cite{DG}, Theorem 2.3.
\end{remark}

We are now ready to state our main theorem.
\begin{theorem}\label{main-theorem}
    Let $Z_{reg}(\tau,t,H)$ be as in Appendix \ref{appendix-regularization}, and $\tau,t,H$ as in Definition \ref{scaling-variables-def}. Then,
        \begin{equation}
            \lim_{n\to \infty} n^2{\bf d}\log\frac{Z_{n}(\tau,t,H;n)}{Z_{reg}(\tau,t,H)} = H_1 d\nu + H_2d\mu + H_5d\eta,
        \end{equation}
    where $H_j$ are the Hamiltonians for the $(3,4)$ string equation, as in Proposition \ref{AAA} (see also Appendix \ref{STRING-APPENDIX}).
\end{theorem}

From the combinatorial standpoint, our result has something to say about the Ising model on random maps with infinitely many vertices at higher genus. As we have discussed already, using a variety of methods, the groups of Albenque, M\'{e}nard, and Schaeffer (\cite{AM} Theorem 6), Bernardi and Bousquet-M\'{e}lou (\cite{BBM} Claim 22), and Chen and Turunen (\cite{CT1} Theorem 1) are able to prove a version of the following result, which we state here in the context of the present work:\vspace{2mm}

\noindent
        As $n\to \infty$, and with the notations of Equations \eqref{free-energy-expansion}, \eqref{Fg-expansion}, there exist $C_0 = C_0(\tau,H)$, $z_0 = z_0(\tau,H)$ such that\footnote{The results of these works instead count \textit{unlabelled, rooted} maps, and can be compared to what is presented here by multiplying the right hand side of \eqref{Z0t},\eqref{Zgt} by $n$.}
            \begin{equation}\label{Z0t}
                \mathcal{Z}_{n,0}(\tau,H) = 
                n!C_0(\tau,H) [z_0(\tau,H)]^{-n} 
                \cdot \begin{cases}
                    n^{-\frac{7}{2}}[1+\OO(n^{-1/2})], & (\tau,H) \neq (\tau_c,H_c),\\
                    n^{-\frac{10}{3}} [1+\OO(n^{-1/3})], & (\tau,H) = (\tau_c,H_c).
                \end{cases}
            \end{equation}

These combinatorial results were the first rigorous derivation of the scaling exponents originally calculated in the work of Kazakov and Boulatov \cite{Kazakov1,Kazakov2}. In fact, there is a conjecture generalizing the above \cite{Kostov,Duplantier,PGZ} regarding the behavior of the Taylor coefficient $Z_{n,g}(\tau,H)$ on arbitrary genus surfaces: \vspace{2mm}

\noindent
\textbf{Conjecture.} \textit{(\cite{PGZ} Section 7).}
        As $n\to \infty$, there exist $C_g = C_g(\tau,H)$, $z_0 = z_0(\tau,H)$ such that
        \begin{equation}\label{Zgt}
                \mathcal{Z}_{n,g}(\tau,H) = 
                n!C_g(\tau,H) [z_0(\tau,H)]^{-n} \cdot 
                \begin{cases}
                     n^{-\frac{1}{2}(7-5g)}[1+\OO(n^{-1/2})], & (\tau,H) \neq (\tau_c,H_c),\\
                     n^{-\frac{1}{3}(10-7g)} [1+\OO(n^{-1/3})], & (\tau,H) = (\tau_c,H_c).
                \end{cases}
            \end{equation}

Results such as Theorem \ref{main-theorem} are expected to shed light on the nature of these combinatorial problems, especially away from the genus $0$ situation \cite{PGZ}. Moreover, it is expected that when $(\tau,H) \neq (\tau_c,H_c)$ (and is in fact proven, for instance, when $\tau=H=0$) the constants $C_g(\tau,H)$ in the above conjecture are related to the asymptotics of the tritronqu\'{e}e solution to Painlev\'{e} I (see \cite{GJ,BenderGaoRichmond}, and more concretely \cite{BGM}, Theorem  1.7). One should expect that for $(\tau,H) = (\tau_c,H_c)$, the constants $C_g(\tau_c,H_c)$ are related to the special solution of the string equation \eqref{string-equation} studied in this work.

\subsection{Associated biorthogonal polynomials \& Riemann-Hilbert problem}
The analysis of the partition function \eqref{partition-function} and its multiscaling limit is made possible through its connection to a family of
\textit{biorthogonal polynomials}. Let us briefly summarize this connection; we refer the reader to \cite{DHL1} Section 1.2 for details. We define
two families of monic polynomials $\{P_n(z)\}$, $\{Q_n(w)\}$ via the relation
    \begin{equation}\label{BOP-def}
        \int_{\Gamma_t}\int_{\Gamma_t} P_n(z)Q_m(w) \exp\left[N\Phi(z,w;\tau,t,H)\right]dzdw = h_{n}(\tau,t,H;N)\delta_{nm},
    \end{equation}
where $\Gamma_t:=\{xe^{i\theta}\mid x\in \RR, \theta = \arg (t^{-1/4})\}$, and is given the orientation which increasing with $x$. The partition function \eqref{partition-function} is related to these polynomials via the relation
    \begin{equation}\label{BOP-partition-function}
        Z_n(\tau,t,H;N) = n!\frac{C_{n,N}}{\tau^{-\frac{n(n-1)}{2}}} \prod_{k=0}^{n-1}h_k(\tau,t,H;N),
    \end{equation}
where $C_{n,N}$ is a constant depending only on $n,N$, which is not important to us. Note that the right hand side of the above expression makes sense
for complex $t$, and thus can be used as a definition for the analytic continuation of $Z_n$. We are particularly interested in the regime when $t<0$, and so from here on we fix $t$ to be negative, and take as a definition $\Gamma_t:=\Gamma$ to be the contour starting from $e^{3\pi i/4}\cdot\infty$ and ending at $e^{-\pi i/4}\cdot \infty$, which corresponds to analytic continuation of $Z_n$ through the upper half-plane. Let us observe that we can redefine $\Gamma$ in any compact subset of $\CC$ and still obtain the same family of biorthogonal polynomials (and thus the same analytic continuation). We will make use of this fact later on.

We will make use of some of the special properties of the biorthogonal polynomials $\{P_n(z)\}$, $\{Q_n(w)\}$. These polynomials satisfy finite-term recurrence relations \cite{Mehta,Kazakov1} :
    \begin{align}
        zP_n(z) = P_{n+1}(z) + R_nP_{n-1}(z) + S_nP_{n-3}(z), \label{P-recurrence}\\
        wQ_n(w) = Q_{n+1}(w) + \tilde{R}_nQ_{n-1}(w) + \tilde{S}_nQ_{n-3}(w),\label{Q-recurrence}
    \end{align}
and 
    \begin{align*}
        R_n = R_n(\tau,t,H;N), \qquad&\qquad S_n = S_n(\tau,t,H;N),\\
        \tilde{R}_n(\tau,t,H;N) = R_n(\tau,t,-H;N), \qquad&\qquad \tilde{S}_n(\tau,t,H;N) = S_n(\tau,t,-H;N)
    \end{align*}
are the \textit{recurrence coefficients}. We also frequently encounter the quantity
    \begin{equation}
        f_n = f_n(\tau,t,H;N):= \frac{h_n(\tau,t,H;N)}{h_{n-1}(\tau,t,H;N)},
    \end{equation}
which we will also refer to as a recurrence coefficient. As is well-known (cf. \cite{Mehta,Kazakov1,Kazakov2}), these recurrence coefficients satisfy 
the following non-linear polynomial equations:
        \begin{align}
            \tau S_n &= te^{-H}f_nf_{n-1}f_{n-2},\label{dString-a}\\
            \tau \tilde{S}_n &= te^{H}f_nf_{n-1}f_{n-2},\\
            \tau R_n &= f_n[1 + te^{-H}(\tilde{R}_{n+1} + \tilde{R}_{n} + \tilde{R}_{n-1})],\\
            \tau \tilde{R}_n &= f_n[1 + te^{-H}(R_{n+1} + R_{n} + R_{n-1})],\\
            \frac{n}{N} &= -\tau f_n + R_n + te^{H}[S_{n+2} + S_{n+1} + S_{n} + R_n(R_{n+1}+R_{n}+R_{n-1})].\label{dString}
        \end{align}
which we collectively refer to as the \textit{discrete string equation}. These equations can be obtained by integrating by parts the recurrence relations \eqref{P-recurrence}, \eqref{Q-recurrence} against appropriately chosen polynomials. Observe that one can eliminate $S_n$, $\tilde{S}_n$ from the above equations entirely to obtain a pair of equations of $R_n$ and $f_n$ alone (respectively, $\tilde{R}_n$ and $f_n$ alone). Part of our analysis will rely on the fact that the recurrence coefficients satisfy these equations.

We are also able to obtain the following Theorem regarding the asymptotic behavior of the recurrence coefficients in the double scaling regime:

\begin{theorem}\label{propA}
    With $\tau,t,H$ as in Definition \ref{scaling-variables-def}, the recurrence coefficients $R_n(\tau,t,H;n)$, $f_n(\tau,t,H;n)$ admit the 
    following asymptotic expansions: 
        \begin{align}
            R_n(\tau,t,H;n) &\sim \mathfrak{r}_c + \sum_{k=2}^{\infty} \mathfrak{r}_k(\eta,\mu,\nu)n^{-k/7},\\
            f_n(\tau,t,H;n) &\sim \mathfrak{f}_c + \sum_{k=2}^{\infty} \mathfrak{f}_k(\eta,\mu,\nu)n^{-k/7}
        \end{align}
    where $\mathfrak{r}_c = \frac{12}{5}$, $\mathfrak{f}_c = \frac{6}{5}$, and the above expansions hold locally uniformly in $\eta,\mu,\nu$. Furthermore, the following properties of the functions 
    $\{\mathfrak{r}_k(\eta,\mu,\nu),\mathfrak{f}_k(\eta,\mu,\nu)\}_{k=2}^{\infty}$ hold:
        \begin{enumerate}
            \item The functions $\mathfrak{r}_k(\eta,\mu,\nu)$, $\mathfrak{f}_k(\eta,\mu,\nu)$ are differential polynomials in $U,V$, $\eta,\mu,\nu$\footnote{Here, the phrase \textit{differential polynomial in a function $f$} means a polynomial in $f$ and its derivatives.},
            \item We have explicitly that
                \begin{align}
                    \mathfrak{r}_{2}(\eta,\mu,\nu) &= -\frac{c^2}{10}\left(3U(\eta,\mu,\nu)-\eta\right),\qquad \mathfrak{f}_{2}(\eta,\mu,\nu) =  -\frac{c^2}{10}\left(3U(\eta,\mu,\nu)+2\eta\right),\\
                    \mathfrak{r}_{3}(\eta,\mu,\nu) &= \frac{3c^3}{10}V(\eta,\mu,\nu),\qquad\quad\,\,\,\,\qquad \mathfrak{f}_{3}(\eta,\mu,\nu) = 0,
                \end{align}
            where $c:=5^{1/7}2^{5/7}$ is a positive constant.
        \end{enumerate}
\end{theorem}
This theorem is the analog of Theorem 1.2 of \cite{DK0} (see also Theorem 4.1 of \cite{Bleher-Kuijlaars}), and shows how the recurrence coefficients of the biorthogonal polynomials behave in the double scaling regime. Observe that the asymptotic expansion is in powers of $n^{-1/7}$ here, whereas in \cite{DK0,Bleher-Kuijlaars}, the expansion was in powers of $n^{-1/5}$.
The proof of this Theorem is non-trivial, as it relies on Deift-Zhou analysis of the Riemann-Hilbert problem for the associated biorthogonal polynomials in this critical regime, and takes up the bulk of this work. The proof of Theorem \ref{propA} is established in Section \ref{section:RHP}.

Part of the proof of our main result relies on the fact \cite{KM} that the biorthogonal polynomials defined by \eqref{BOP-def} admit a Riemann-Hilbert formulation, which we present here. Consider the following Riemann-Hilbert problem (RHP) on $\Gamma$ ($=$ a contour starting from $e^{3\pi i/4}\cdot\infty$ and ending at $e^{-\pi i/4}\cdot \infty$):
    \begin{equation}\label{BOP-RHP}
        \boldsymbol{Y}_+(z) = \boldsymbol{Y}_-(z)
        \begin{pmatrix}
            1 & \hat{w}_0(z) & \hat{w}_1(z) & \hat{w}_2(z)\\
            0 & 1 & 0 & 0\\
            0 & 0 & 1 & 0\\
            0 & 0 & 0 & 1
        \end{pmatrix}, \qquad z\in \Gamma,
    \end{equation}
where
    \begin{equation}
        \hat{w}_j(z) = \int_{\Gamma} w^j e^{N\Phi(z,w)} dw, \qquad j=0,1,2.
    \end{equation}
As $z\to \infty$, $\mathbf{Y}(z)$ is normalized as
    \begin{equation}\label{BOP-normalization}
        \boldsymbol{Y}(z) = \left[\mathbb{I}+ \OO(z^{-1})\right]z^{n\hat{\sigma}},
    \end{equation}
when $n$ is a multiple of $3$, where
    \begin{equation}\label{hat-sigma-def}
        \hat{\sigma} := \text{diag }(1,-1/3,-1/3,-1/3).
    \end{equation}
\textit{We emphasize that this restriction on $n$ is taken for ease of exposition; this assumption is not essential \cite{DK1,DKM,DHL1}.} As was demonstrated in \cite{KM}, the Riemann-Hilbert problem admits a unique solution. Furthermore, the biorthogonal polynomial $P_n(z)$ can be identified with the $1$-$1$ entry of $\boldsymbol{Y}$:
    \begin{equation}
        [\boldsymbol{Y}]_{11}(z) = P_n(z).
    \end{equation}

    In fact, we can say slightly more: we can actually determine the full solution to the RHP \eqref{BOP-RHP} in terms of the biorthogonal polynomials $P_n$ and their recurrence coefficients. This solution is established in Proposition \ref{RHP-determination-prop}. As was pointed out in \cite{KM}, the RHP \eqref{BOP-RHP} is really a Riemann-Hilbert problem for a family of multiple orthogonal polynomials (MOPs). In the language of multiple orthogonal polynomials, this Riemann-Hilbert problem gives rise to a \textit{non-normal} family of MOPs. This is due to the fact that, for $n=3k$, the $21$-entry of $\boldsymbol{Y}$ (which corresponds to the MOP with indices $(k-1,k,k)$, and therefore should be of degree $k-1+k+k=3k-1$), is not of degree $3k-1$, but is instead of degree $3k-2$.

    To establish our main results, we rely on a steepest descent analysis of the RHP \eqref{BOP-RHP}. We encountered some new difficulties in the construction of local parametrices: the usual techniques used in the 1-matrix model are no longer applicable \cite{DK0,BleherDeano}, and one must use the technique of `shrinking discs' \cite{DKZ,DG}. However, this still is not enough, and one must further modify the local parametrices to account for this. The approach we use exploits the symmetry of the model RHP, and introduces `nested' parametrices, a technique which is related to the works \cite{Molag0,Molag}. However, it seems that our method is new. We find it notable that the matching problem appearing in \cite{DG} concerning a $(4,2)$-type critical point was manageable without the use of nested parametrices. We expect that the method we have developed in this work is essential for dealing with higher order criticalities in the $2$-matrix model.

\subsection{Strategy of the proof and notations}
We organize the rest of this work as follows. In Section \ref{section:CPF}, we state several propositions (as well as Theorem \ref{propA}), and defer their proofs until later in the section. Then, using these Propositions, we prove Theorem \ref{main-theorem}. The proof of most of these Propositions is straightforward, and is carried out already later in Section \ref{section:CPF}; however, Theorem \ref{propA} requires more work. The rest of the paper is then devoted to the proof of Theorem \ref{propA}. In Section \ref{section:RHPalgebraic}, we establish some algebraic properties of the RHP \eqref{BOP-RHP} we will need. In Section \ref{section:SpectralCurve}, we construct the \textit{spectral curve}, an important ingredient in the steepest descent analysis. Finally, in Section \ref{section:RHP}, we prove Theorem \ref{propA}, which completes the proof of the main Theorem \ref{main-theorem}.

Throughout this work, we will employ many notations without comment. We list them here for the convenience of the reader.
\begin{itemize}
    \item $\omega :=e^{\frac{2\pi i}{3}} = -\frac{1}{2} + i\frac{\sqrt{3}}{2}$ is the principal third root of unity.
    \item $E_{ij}$ will denote the $4\times 4$ matrix with a $1$ in the $(i,j)^{th}$ entry, and zeroes elsewhere.
    \item The Pauli matrices
        \begin{equation*}
            \sigma_{1} :=
            \begin{pmatrix}
                0 & 1\\
                1 & 0
            \end{pmatrix},\qquad\qquad
            \sigma_{2} :=
            \begin{pmatrix}
                0 & -i\\
                i & 0
            \end{pmatrix},\qquad\qquad
            \sigma_{3} :=
            \begin{pmatrix}
                1 & 0\\
                0 & -1
            \end{pmatrix}.
        \end{equation*}
    We introduce some additional notations for the third Pauli matrix $\sigma_3$: we define
        \begin{equation*}
            z^{C\sigma_3} :=
            \begin{pmatrix}
                z^C & 0\\
                0 & z^{-C}
            \end{pmatrix},
            \qquad\qquad
            e^{f(z)\sigma_3} :=
            \begin{pmatrix}
                e^{f(z)} & 0\\
                0 & e^{-f(z)}
            \end{pmatrix}.
        \end{equation*}
    \item $\hat{\sigma}_{ij}$ is the $4\times 4$ matrix whose action by conjugation on matrices permutes the $i^{th}$ and $j^{th}$ row/column. For instance,
    $\hat{\sigma}_{24}$ is
        \begin{equation*}
            \hat{\sigma}_{24} = 
            \begin{pmatrix}
                1 & 0 & 0 & 0\\
                0 & 0 & 0 & 1\\
                0 & 0 & 1 & 0\\
                0 & 1 & 0 & 0
            \end{pmatrix}.
        \end{equation*}
    \item To ease notations, blocks of zeros in matrices will be simply denoted by $0$, when there is no cause for ambiguity. For instance, if 
    $A$ is a $3\times 3$ matrix, then all of the following matrices are equivalent:
    \begin{equation*}
        \begin{pmatrix}
            1 & 0 & 0 & 0\\
            0 & & & \\
            0 & & A & \\
            0 & & & 
        \end{pmatrix} = 
        \begin{pmatrix}
            1 & 0_{3\times 1}\\
            0_{1\times 3} & A
        \end{pmatrix} =
        \begin{pmatrix}
            1 & 0\\
            0 & A
        \end{pmatrix}.
    \end{equation*}
    \item If $A$ is an $n\times n$ matrix, and $B$ is an $m\times m$ matrix, then we define the following $(n+m)\times(n+m)$ matrix:
    \begin{equation*}
        A\oplus B := 
        \begin{pmatrix}
            A & 0\\
            0 & B
        \end{pmatrix}.
    \end{equation*}
    \item If $X(z)$ is the solution to an RHP defined on the oriented contour $\gamma$, we let $J_X(z) :\gamma\to \CC$ denote its jump matrix. In other words,
        \begin{equation*}
            X_+(z) = X_-(z)J_X(z),\qquad z\in \gamma.
        \end{equation*}
\end{itemize}

\subsection{Acknowledgments}
MD and NH were supported by the European Research Council (ERC), Grant Agreement No. 101002013. Part of this work was completed while MD held a Chaire d’Excellence from the Fondation Sciences Math\'{e}matiques de Paris (FSMP). MD and NH thank the LPSM at Sorbonne University for its
hospitality during this period.

\section{Calculation of the critical partition function}\label{section:CPF}
\subsection{Proof of Theorem \ref{main-theorem}}
Here, we prove Theorem \ref{main-theorem}. This is accomplished through a series of Propositions, which we state here, and prove in the indicated sections. The proof of Theorem \ref{main-theorem} also relies on the result of Theorem \ref{propA}, which we also postpone until later.

\begin{prop}\label{propB}
    Fix $\ell\in \ZZ$. With $t,\tau,H$ as in Definition \ref{scaling-variables-def}, and under the assumption that 
    $R_n(\tau,t,H;n)$, $f_n(\tau,t,H;n)$ admit the asymptotic expansions of Theorem \ref{propA}, we have that
        \begin{align}
            R_{n+\ell}(\tau,t,H;n) &\sim \lambda_n^2\left(\mathfrak{r}_c + \sum_{k=2}^{\infty} \mathfrak{r}_k(\eta_n,\mu_n,\nu_n)(n+\ell)^{-k/7}\right),\label{R-nl}\\
            f_{n+\ell}(\tau,t,H;n) &\sim \lambda_n^2\left(\mathfrak{f}_c + \sum_{k=2}^{\infty} \mathfrak{f}_k(\eta_n,\mu_n,\nu_n)(n+\ell)^{-k/7}\right),\label{f-nl}
        \end{align}
    where $\lambda_n^2 :=\left(1+\frac{\ell}{n}\right)$, and $(\eta_n,\mu_n,\nu_n) = (\eta,\mu,\nu) + o(1)$ are the unique solution to the implicit equations
        \begin{equation}
            \lambda_n^{-2}t(\eta_n,\nu_n;n+\ell)=t(\eta,\nu;n),\qquad \tau(\eta_n,\nu_n;n+\ell)=\tau(\eta,\nu;n), \qquad H(\mu_n;n+\ell)=H(\mu;n)
        \end{equation}
    in a neighborhood of $(\eta_n,\mu_n,\nu_n) = (\eta,\mu,\nu)$, which exist for $n$ sufficiently large.
\end{prop}
Since the expansions for $R_n$, $f_n$ held uniformly, the above proposition implies that:
\begin{cor} For any fixed $\ell\in \ZZ$,
    \begin{align*}
        R_{n+\ell}(\tau,t,H;n) &\sim \mathfrak{r}_c + \sum_{k=2}^{\infty} \mathfrak{r}_k^{(\ell)}(\eta,\mu,\nu) n^{-k/7},\\
        f_{n+\ell}(\tau,t,H;n) &\sim \mathfrak{f}_c + \sum_{k=2}^{\infty} \mathfrak{f}_k^{(\ell)}(\eta,\mu,\nu) n^{-k/7},
    \end{align*}
where the functions $\mathfrak{r}_k^{(\ell)}$, $\mathfrak{f}_k^{(\ell)}$ are explicit differential polynomials in $\mathfrak{r}_k$, $\mathfrak{f}_k$, $\eta,\mu,\nu$, and $\ell$.
\end{cor}
 Proposition \ref{propB} is established in Section \ref{propB-proof}. This Proposition, in conjunction with a formula for ${\bf d}\log Z_n(\tau,t,H)$ in terms of a \textit{finite number} of recurrence coefficients, will be enough to prove the main theorem. Indeed, we can establish such a formula:
\begin{prop}\label{propC}
Let
    \begin{equation}
        W_n(\tau,t,H;N):=\frac{N^2}{4}\bigg[S_n+e^Ht(R_{n}S_{n-1}+R_{n-2}S_{n-1}+R_{n}R_{n-1}R_{n-2}+(R_{n+1}+R_n+R_{n-2}+R_{n-3})S_n)\bigg],
    \end{equation}
    and $\tilde{W}_n = \tilde{W}_n(\tau,t,H) := W_n(\tau,t,-H;N)$, and
    \begin{align}
        \Omega_{\tau} &:= \frac{n(n-1)}{2\tau} + N^2\bigg[\frac{\tau}{2}\left(f_nf_{n+1}-\tilde{R}_n(R_{n+1}-R_{n-1})\right) -\frac{te^H}{2}\sum_{k=0}^{2} f_{n-k}\left(R_{n-k+1}R_{n-k+2} + S_{n-k+2}+S_{n-k+3}\right) \nonumber\\
        &+\frac{te^H}{2}f_n\left(R_{n-1}R_{n-2}+S_n+S_{n-1}\right)-\frac{te^H}{2}\sum_{k=-1}^{1}(R_{n+k-2}S_{n+k+1}+R_{n+k+1}S_{n+k})\frac{\tilde{R}_{n+k-2}}{f_{n+k-2}} + \nonumber\\
        &-\frac{te^H}{2}(R_nS_{n-1}+R_{n-3}S_n)\frac{\tilde{R}_n}{f_n}-\frac{\tau}{2}(R_{n-1}R_n+S_n+S_{n+1})\frac{\tilde{R}_n\tilde{R}_{n-1}}{f_nf_{n-1}}-\frac{t^2}{2\tau}\sum_{k=0}^{2}S_{n+k}\tilde{S}_{n+k-3}\bigg], \label{dlogZ_dtau}\\
        \Omega_{t} &:= -\frac{n^2}{2t} + \frac{\tau}{2t}\Omega_{\tau} + \frac{1}{t}(W_{n+1}+\tilde{W}_{n+1}) +\frac{N}{4t}n(R_n+\tilde{R}_n),\label{dlogZ_dt}\\
        \Omega_{H} &:= -(W_{n+1}-\tilde{W}_{n+1})-\frac{N}{4}(n-2)\left(R_n-\tilde{R}_n\right).\label{dlogZ_dH}
    \end{align}
    Then
    \begin{equation}
        {\bf d}\log Z_n(\tau,t,H) = \Omega_{\tau} d\tau + \Omega_t dt + \Omega_H dH.
    \end{equation}
\end{prop}
Proposition \ref{propC} is proven in Section \ref{propC-proof}.
For now, let us assume these Propositions, and prove Theorem \ref{main-theorem}.
\begin{proof}
    \textit{(of Theorem \ref{main-theorem})}.
        By Theorem \ref{propA} and Proposition \ref{propB}, the recurrence coefficients $R_{n+\ell}, f_{n+\ell}$, carry an expansion in $n^{-1/7}$, and 
        furthermore the first few terms in this expansion are explicitly known. We first insert these expansions into the discrete string equation\footnote{We have eliminated here from Equations \eqref{dString-a}--\eqref{dString} the variables $S_k,\tilde{S}_k,\tilde{R}_k$.}: 
            \begin{align}
                0 &= \tau R_n - f_n\left[1 + \frac{te^{-H}}{\tau}\sum_{k=-1}^{1}f_{n+k}\left(1+te^H(R_{n+k-1}+R_{n+k}+R_{n+k+1})\right)\right],\label{dS-loc-A}\\
                1 &= \tau f_n + R_n + \frac{t^2}{\tau}\left[\sum_{k=-1}^{1}f_{n+k+1}f_{n+k}f_{n+k-1}\right] + te^HR_n(R_{n+1}+R_n+R_{n-1})\label{dS-loc-B}
            \end{align}
        This leads a number of differential identities on the coefficients $\mathfrak{r}_k,\mathfrak{f}_k$, which must be computed up to order $n^{-1}$ (i.e. the identities will involve $\mathfrak{r}_7,\mathfrak{f}_7$ at most). The first few terms of \eqref{dS-loc-A}, 
        \eqref{dS-loc-B}, are
            \begin{align*}
                \eqref{dS-loc-A}\quad \Leftrightarrow \quad 0 &=\mathcal{Q}(\mathfrak{r}_4,\mathfrak{f}_4,U,V;\eta,\mu,\nu)n^{-4/7} + \mathcal{Q}_{5,a}(\mathfrak{r}_5,\mathfrak{f}_5,\mathfrak{r}_4,\mathfrak{f}_4,U,V;\eta,\mu,\nu)n^{-5/7} + \OO(n^{-6/7}),\\
                \eqref{dS-loc-B}\quad \Leftrightarrow \quad 0 &=\mathcal{Q}(\mathfrak{r}_4,\mathfrak{f}_4,U,V;\eta,\mu,\nu)n^{-4/7} + \mathcal{Q}_{5,b}(\mathfrak{r}_5,\mathfrak{f}_5,\mathfrak{r}_4,\mathfrak{f}_4,U,V;\eta,\mu,\nu)n^{-5/7} + \OO(n^{-6/7}),
            \end{align*}
            with $\mathcal{Q}$, $\mathcal{Q}_{5,a},\mathcal{Q}_{5,b}$ being some explicit differential polynomials in the variables expressed above.
            The exact form of these polynomials is not so illuminating; for instance, 
            \begin{align*}
                \mathcal{Q}(\mathfrak{r}_4,\mathfrak{f}_4,U,V;\eta,\mu,\nu) &= \frac{5^{2/7}2^{3/7}}{12}\left(\frac{\partial^2}{\partial \nu^2} -3U+5\eta\right)(\mathfrak{r}_4+\mathfrak{f}_4)\\
                &+ \frac{5^{6/7}2^{2/7}}{720}\left(18U^3-138\eta^2U+230\eta^3-9(U')^2-108V^2-72\nu\right).
            \end{align*}
            What is important to observe is that we that
            \begin{equation*}
                \mathcal{Q} = 0, \qquad\qquad \mathcal{Q}_{5,a} = 0,\qquad\qquad \mathcal{Q}_{5,b} = 0.
            \end{equation*}
        Now, if we then insert the expansions of $R_{n+\ell}, f_{n+\ell}$ into the expression we obtained for ${\bf d}\log Z_n(\tau,t,H)$ via 
        Proposition \ref{propC}, we obtain that
            \begin{equation*}
                n^2{\bf d}\log \frac{Z_n(\tau,t,H)}{Z_{reg}(\tau,t,H)} = \frac{5^{9/7}2^{3/7}}{6} \mathcal{Q} d\eta \cdot n^{6/7} - \left[\frac{5^{2/7}2^{3/7}}{15}\mathcal{Q}_{5,a} - \frac{9}{10}5^{2/7}2^{3/7}\mathcal{Q}_{5,b} + \frac{37}{60}5^{3/7}2^{1/7}\frac{\partial}{\partial \nu}\mathcal{Q}\right] d\eta \cdot n^{5/7} + \OO(n^{4/7}),
            \end{equation*}
        and so we notice that $n^2{\bf d}\log \frac{Z_n(\tau,t,H)}{Z_{reg}(\tau,t,H)}$ is actually of order $\OO(n^{4/7})$ due to the relations we 
        obtained from the string equation. Continuing in this fashion, with some tedious calculation\footnote{We leave out the full calculation here because it is not very illuminating. The interested reader should have no trouble using a compute algebra system to verify the result described. We are happy to share our own Maple file upon request.} (we used Maple here), one can show that $n^2{\bf d}\log \frac{Z_n(\tau,t,H)}{Z_{reg}(\tau,t,H)}$ is of order $1$ in $n$, due to the cancellations occurring from the relations one derived from the string equation. At order $1$, we see that
            \begin{equation*}
                n^2{\bf d}\log \frac{Z_n(\tau,t,H)}{Z_{reg}(\tau,t,H)} = H_1d\nu + H_2d\mu + H_5d\eta + \OO(n^{-1/7}),
            \end{equation*}
        where $H_a$ are the Hamiltonians of the $(3,4)$ string equation. This completes the proof.
\end{proof}

\subsection{Proof of Proposition \ref{propB}} \label{propB-proof}


Let us first state the following simple Lemma:
\begin{lemma}
    For any $\lambda >0$ the recurrence coefficients $R_n(\tau,t,H;N)$ and $f_n(\tau,t,H;N)$ satisfy the relations
    \begin{equation}
       \lambda^2R_n(\tau,\lambda^2t,H;\lambda^{2}N) = R_n(\tau,t,H;N),\qquad\qquad \lambda^2f_n(\tau,\lambda^2t,H;\lambda^{2}N) = f_n(\tau,t,H;N).
    \end{equation}
\end{lemma}
\begin{proof}
    This follows immediately from the definition of these coefficients, and an appropriate change of variables.
\end{proof}
    We now proceed to the proof of  Proposition \ref{propB}.
\begin{proof}\textit{(Of  Proposition \ref{propB}).}
    Theorem \ref{propA} (whose conclusion we have assumed) tells us that, for any fixed $\ell\in \ZZ$,
        \begin{equation*}
        R_{n+\ell}\bigg(\tau(\eta_n,\nu_n;n+\ell),t(\eta_n,\nu_n;n+\ell),H(\mu_n;n+\ell);n+\ell\bigg) \sim \mathfrak{r}_c + \sum_{k=1}^{\infty}\mathfrak{r}_k(\eta_n,\mu_n,\nu_n)(n+\ell)^{-k/7},
    \end{equation*}
    which holds for $(\eta_n,\mu_n,\nu_n) = (\eta,\mu,\nu) + o(1)$ as $n\to \infty$, by local uniformity of the asymptotic. Then,
    \begin{align*}
        &\lambda_n^{-2}R_{n+\ell}\bigg(\tau(\eta_n,\nu_n;n+\ell),t(\eta_n,\nu_n;n+\ell),H(\mu_n;n+\ell);n+\ell\bigg)\\
        &=\lambda_n^{-2} R_{n+\ell}\bigg(\tau(\eta_n,\nu_n;n+\ell),t(\eta_n,\nu_n;n+\ell),H(\mu_n;n+\ell);n\lambda_n^2\bigg)\\
        &=R_{n+\ell}\bigg(\tau(\eta_n,\nu_n;n+\ell),\lambda_n^{-2}t(\eta_n,\nu_n;n+\ell),H(\mu_n;n+\ell);n\bigg).\\
    \end{align*}
We choose $(\eta_n,\mu_n,\nu_n)$ such that
    \begin{align*}
        \tau(\eta_n,\nu_n;n+\ell) = \tau(\eta,\nu;n), \qquad \lambda_n^{-2}t(\eta_n,\nu_n;n+\ell) = t(\eta,\nu;n),\qquad H(\mu_n;n+\ell) = H(\mu;n)
    \end{align*}
Such $(\eta_n,\mu_n,\nu_n)$ exist for $n$ sufficiently large in a neighborhood of $(\eta,\mu,\nu)$ by the implicit function theorem. Indeed, if we set
    \begin{align*}
        F_1(\eta_n,\mu_n,\nu_n)&:=n^{4/7}[\tau(\eta_n,\nu_n;n+\ell) -\tau(\eta,\nu;n)],\\
        F_2(\eta_n,\mu_n,\nu_n)&:=n^{4/7}[\lambda_n^{-2}t(\eta_n,\nu_n;n+\ell) - t(\eta,\nu;n)],\\
        F_3(\eta_n,\mu_n,\nu_n)&:=n^{5/7}[H(\mu_n;n+\ell) -H(\mu;n)],
    \end{align*}
then
    \begin{equation*}
        \det\begin{pmatrix}
            \partial_{\eta_n} F_1 & \partial_{\mu_n} F_1 & \partial_{\nu_n} F_1\\
            \partial_{\eta_n} F_2 & \partial_{\mu_n} F_2 & \partial_{\nu_n} F_2\\
            \partial_{\eta_n} F_3 & \partial_{\mu_n} F_3 & \partial_{\nu_n} F_3
        \end{pmatrix}\bigg|_{\eta_n=\eta,\mu_n=\mu,\nu_n=\nu} = \frac{25}{5184}2^{2/7}5^{6/7} + \OO(n^{-2/7}),
    \end{equation*}
where the error is uniform on compact subsets of $(\eta,\mu,\nu)$. So, provided $n$ is sufficiently large, a unique solution exists. This completes the proof.
\end{proof}

\begin{remark}
    Explicitly, the first few terms of $\eta_n,\mu_n,\nu_n$ are given by
    \begin{align*}
        \eta_n &= \eta + \frac{3}{10} 2^{4/7}5^{5/7}(1+\kappa)\ell n^{-5/7} + \OO(n^{-1}),\\
        \mu_n &= \mu + \frac{5}{7}\mu\ell n^{-1} + \OO(n^{-2}),\\
        \nu_n &= \nu + \frac{1}{2}5^{1/7}2^{5/7}\ell n^{-1/7} -\frac{5}{6}2^{1/7}5^{3/7}\eta(\kappa-7/5)\ell n^{-3/7} + \frac{5}{72}2^{4/7}5^{5/7}(\kappa-3)(10\kappa-9)\eta^2n^{-5/7} + \OO(n^{-1}).
    \end{align*}
\end{remark}

\subsection{Proof of Proposition \ref{propC}}\label{propC-proof}
In this section, we prove Proposition \ref{propC}, which demonstrates that the differential of $\log Z_n$ can be written in terms of a finite number of recurrence coefficients. Our approach involves exploiting a so-called \textit{derivative Christoffel-Darboux identity}, i.e. an identity which allows one to express the derivative in $z$ or $w$ of the kernel
    \begin{equation}
        K_n(z,w) := e^{N\Phi(z,w;\tau,t,H)}\sum_{k=0}^{n-1}\frac{P_n(z)Q_n(w)}{h_n}
    \end{equation}
in terms of a finite sum of biorthogonal polynomials. This kind of formula was first established in \cite{BEH}, and has been applied in a wider range of random matrix problems more recently in \cite{ABK,BLY}. The following proof is essentially adapted from \cite{BLY} to the present context, for the convenience of the reader.
\begin{lemma}\label{CD-lemma} \textit{(Derivative Christoffel-Darboux identity).}
    \begin{align*}
        \frac{\partial}{\partial z} K_n(z,w) &= N\tau e^{N\Phi(z,w)}\left[-\frac{P_{n-1}(z)Q_n(w)}{h_{n-1}} + \frac{\tilde{R}_nP_n(z)Q_{n-1}(w)}{h_n} + \sum_{j=0}^{2}\frac{\tilde{S}_{n+j}P_{n+j}(z)Q_{n+j-3}(w)}{h_{n+j}} \right],\\
        \frac{\partial}{\partial w} K_n(z,w) &=N\tau e^{N\Phi(z,w)}\left[-\frac{P_n(z)Q_{n-1}(w)}{h_{n-1}} + \frac{R_nP_{n-1}(z)Q_n(w)}{h_n} + \sum_{j=0}^{2}\frac{S_{n+j}P_{n+j-3}(z)Q_{n+j}(w)}{h_{n+j}} \right].
    \end{align*}
\end{lemma}
\begin{proof}
    We only prove the first identity; the second follows from similar arguments. Let $V(z;t):=\frac{1}{2}z^2+\frac{t}{4}z^4$,
    and define functions
        \begin{equation}
            \psi_n(z) = \frac{P_n(z)}{h_{n-1}}e^{-NV(z;e^{H}t)},\qquad\qquad \phi_n(w) = Q_n(w)e^{-NV(z;e^{-H}t)},
        \end{equation}
    which satisfy 
        \begin{equation}
            \iint \psi_n(z) \phi_n(w)e^{N\tau zw} dzdw = \delta_{nm}.
        \end{equation}
    Then, define the semi-infinite vectors
    \begin{equation}
        \Psi(z) := \langle \psi_0(z),\psi_1(z),\cdots \rangle^T,\qquad\qquad \Phi(w) := \langle \phi_0(w),\phi_1(w),\cdots \rangle^T,
    \end{equation}
and the semi-infinite matrix
    \begin{equation}
        \Pi_n:=\text{diag }(\underbrace{1,\cdots,1}_{n},0,\cdots), \qquad\qquad \Pi_n \psi_j(z) = \psi_j(z) \cdot \boldsymbol{1}_{j<n},
        \qquad \Pi_n \phi_j(w) = \phi_j(w) \cdot \boldsymbol{1}_{j<n}.
    \end{equation}
With this notation, we can write $K_n(z,w)$ as
    \begin{equation*}
        K_n(z,w) = e^{N\tau zw} \Psi^T(z)\Pi_n\Phi(w).
    \end{equation*}
There exist matrices $A,B$ such that
    \begin{align*}
        \frac{\partial}{\partial z}\Psi(z) &= A\Psi(z),\qquad\qquad z\psi_j(z) = \sum_{k\geq 0} A_{jk}\psi_k(z),\\
        w\Phi(w) &= B\Phi(w),\qquad\qquad w\phi_j(w) = \sum_{k\geq 0} B_{jk}\phi_k(w).
    \end{align*}
Let us compute the matrix element $B_{jk}$. Integrating by parts,
    \begin{align*}
        B_{jk} &= \iint \psi_j(z)w\phi_k(w)e^{N\tau zw}dzdw = \frac{1}{N\tau}\iint \psi_j(z)\phi_k(w)\frac{\partial}{\partial z}\left(e^{N\tau zw}\right)dzdw\\
        &= -\frac{1}{N\tau}\iint\left(\frac{\partial}{\partial z}\psi_j(z)\right)\phi_k(w)e^{N\tau zw}dzdw = -\frac{1}{N\tau} A_{kj}.
    \end{align*}
In other words,
    \begin{equation}
        B = -\frac{1}{N\tau} A^T.
    \end{equation}
So,
    \begin{align*}
        \left(w+\frac{1}{N\tau}\frac{\partial}{\partial z}\right)\Psi^T(z)\Pi_n\Phi(w) &= \Psi^T(z)\Pi_n B\Phi(w) + \frac{1}{N\tau}\Psi^T(z)A^T\Pi_n\Phi(w) = \Psi^T(z)\left[\Pi_n,B\right]\Phi(w),
    \end{align*}
where $\left[\Pi_n,B\right] := \Pi_n B - B \Pi_n$ is the matrix commutator. It follows that
    \begin{align*}
        \frac{\partial}{\partial z} K_n(z,w) &= \frac{\partial}{\partial z}\left(e^{N\tau zw} \Psi^T(z)\Pi_n \Phi(w)\right) = N\tau e^{N\tau zw}\left(w + \frac{1}{N\tau}\frac{\partial}{\partial z}\right)\Psi^T(z)\Pi_n \Phi(w)\\
        &= N\tau e^{N\tau zw}\Psi^T(z)\left[\Pi_n,B\right]\Phi(w).
    \end{align*}
Now, the matrix $B$ can be computed by appealing to the recurrence relation for the polynomials $Q_n(w)$:
    \begin{equation}
        B_{jk} = 
        \begin{cases}
            1, & j = k+1,\\
            \tilde{R}_k, & j = k-1,\\
            \tilde{S}_k, & j = k-3,\\
            0, & \textit{otherwise}.
        \end{cases}
    \end{equation}
So, we obtain that
    \begin{equation*}
        \frac{\partial}{\partial z}K_n(z,w) = N\tau e^{N\tau zw}\left[-\psi_{n-1}(z)\phi_n(w) + \tilde{R}_n\psi_n(z)\phi_{n-1}(w) + \sum_{j=0}^{2}\tilde{S}_{n+j} \psi_{n+j}(z)\phi_{n+j-3}(w)\right].
    \end{equation*}
Re-expressing $\psi_k(z)$, $\phi_k(w)$ in terms of the biorthogonal polynomials completes the proof.
\end{proof}

We now make use of this formula to prove Proposition \ref{propC}.

\begin{proof} \textit{Proposition \ref{propC}: Formula \eqref{dlogZ_dtau}}.
        Obviously, we have that $\Omega_{\tau} = \frac{\partial}{\partial \tau} \log Z_n(\tau,t,H;N)$. By definition, we have that
        \begin{align*}
            \frac{\partial}{\partial \tau} \log Z_n(\tau,t,H;N) &= \frac{n(n-1)}{2\tau} + \frac{\partial}{\partial \tau} \left(\sum_{k=0}^{n-1}\log h_n(\tau,t,H;N)\right)\\
            &=\frac{n(n-1)}{2\tau} + \sum_{k=0}^{n-1}\frac{1}{h_{k}}\iint P_k(z)Q_k(w)\frac{\partial}{\partial \tau}e^{N\Phi(z,w)} dzdw\\
            &=\frac{n(n-1)}{2\tau} + \iint \left(\sum_{k=0}^{n-1}\frac{ P_k(z)Q_k(w)}{h_{k}}\right) Nxy e^{N\Phi(z,w)} dzdw\\
            &=\frac{n(n-1)}{2\tau} + \iint \frac{\partial}{\partial x} \left(\frac{N}{2}x^2y\right) K_n(x,y) dzdw\\
            &=\frac{n(n-1)}{2\tau} - \frac{N}{2}\iint x^2y\frac{\partial}{\partial x} K_n(x,y) dzdw.
        \end{align*}
    By Lemma \ref{CD-lemma}, we can rewrite $\frac{\partial}{\partial x} K_n(x,y)$ in terms of a finite number of biorthogonal polynomials. By applying the recurrence relations \eqref{P-recurrence}, \eqref{Q-recurrence}, one can rewrite the integrand as a finite sum of the form
        \begin{equation*}
            x^2y\frac{\partial}{\partial x} K_n(x,y) = \sum c_{kj}P_k(z)Q_j(w),
        \end{equation*}
    which can be evaluated explicitly by biorthogonality. 
\end{proof}

In principle, one can apply identical calculations to obtain formulae for $\Omega_{t},\Omega_{H}$. However, we will use an alternative approach, which yields more succinct formulae. 
    
    We consider the following alternative family of monic biorthogonal polynomials:
    \begin{equation}
        \iint p_n(x)q_m(y)e^{N\hat{\Phi}(x,y)}dxdy = \mathfrak{h}_n\delta_{nm},
    \end{equation}
    where
        \begin{equation*}
            \hat{\Phi}(x,y,\varsigma,\mathfrak{t},h) :=\varsigma x y - \frac{e^{-h}}{2}\mathfrak{t}x^2 - \frac{e^{h}}{2}\mathfrak{t}y^2 - \frac{1}{4}x^4-\frac{1}{4}y^4.
        \end{equation*}
    These biorthogonal polynomials satisfy the recurrence relations
        \begin{equation*}
        xp_n(x) = p_{n+1}(x) + r_np_{n-1}(x) + s_np_{n-3}(x),\qquad yq_n(y) = q_{n+1}(y) + \tilde{r}_n q_{n-1}(y) + \tilde{s}_nq_{n-3}(y).
    \end{equation*}
    We write $\mathfrak{h}_n = \mathfrak{h}_n(\varsigma,\mathfrak{t},h;N)$ to denote the dependence
    of $\mathfrak{h}_n$ on the parameters of this model, and similarly for the recurrence coefficients. We can in fact relate the recurrence coefficients and norming constants of these polynomials to those of the ones define by \eqref{BOP-def}; we will then exploit this relation
    to obtain a simpler formula for the logarithmic derivatives of the partition function. To this end, we present the following lemma:
        \begin{lemma}\label{lemmaA1}
            The following relations between the norming constants $\mathfrak{h}_n$ and $h_n$ hold:
                \begin{align}
                    h_n(\tau,t,H;N) = t^{-(n+1)/2}\mathfrak{h}_n\left(\tau t^{-1/2},t^{-1/2},\frac{1}{2}H;N\right).
                \end{align}
            Furthermore, the following relations between the recurrence coefficients hold:
                \begin{align}
                    R_n(\tau,t,H)&= t^{-1/2}e^{-H/2}r_n\left(\tau t^{-1/2},t^{-1/2},\frac{1}{2}H\right),\quad\tilde{R}_n(\tau,t,H) = t^{-1/2}e^{H/2}\tilde{r}_n\left(\tau t^{-1/2},t^{-1/2},\frac{1}{2}H\right),\\
                    S_n(\tau,t,H)&= e^{-H}t^{-1} s_n\left(\tau t^{-1/2},t^{-1/2},\frac{1}{2}H\right),\quad\tilde{S}_n(\tau,t,H)= e^{H} t^{-1}\tilde{s}_n\left(\tau t^{-1/2},t^{-1/2},\frac{1}{2}H\right).
                \end{align}
        \end{lemma}
        \begin{proof}
            The proof of each of the above identities follows from a simple change of variables argument. For instance, by substituting 
            \begin{align*}
                \varsigma = \tau t^{-1/2}, \quad \mathfrak{t}=t^{-1/2}, \quad h= \frac{1}{2}H, \qquad x=t^{1/4}e^{H/4}z,\quad y=t^{1/4}e^{-H/4}w
            \end{align*}
            in the definition of $\mathfrak{h}_n$, one obtains that
            \begin{equation*}
                \mathfrak{h}_n\left(\tau t^{-1/2},t^{-1/2},\frac{1}{2}H;N\right) = t^{(n+1)/2}\int P_n(z)Q_n(w)e^{N\Phi(z,w;\tau,t,H)}dzdw = t^{(n+1)/2}h_n(\tau,t,H;N),
            \end{equation*}
            which yields the first identity.
        \end{proof}
    We additionally will need a formula for the derivatives of $\mathfrak{h}_n$ in $\mathfrak{t}$ and $h$:
    \begin{lemma}\label{lemmaA2}
        The following identities hold:
        \begin{align}
            \frac{\partial}{\partial \mathfrak{t}} \log \mathfrak{h}_k(\varsigma,\mathfrak{t},h;N) = -\frac{N}{2}\left[e^{-h}(r_k+r_{k+1}) + e^{h}(\tilde{r}_k+\tilde{r}_{k+1})\right],\\
            \frac{\partial}{\partial h} \log \mathfrak{h}_k(\varsigma,\mathfrak{t},h;N) = -\frac{N}{2}\mathfrak{t}\left[e^{-h}(r_k+r_{k+1}) - e^{h}(\tilde{r}_k+\tilde{r}_{k+1})\right].
        \end{align}
    \end{lemma}
    \begin{proof}
        We sketch the proof of the first identity only; the second follows from identical arguments. By direct calculation,
            \begin{align*}
                \frac{\partial}{\partial \mathfrak{t}} \log \mathfrak{h}_k(\varsigma,\mathfrak{t},h;N) 
                &= \frac{1}{\mathfrak{h}_k}\iint p_n(x)q_n(y)\frac{\partial}{\partial \mathfrak{t}}e^{N\hat{\Phi}(x,y)}dxdy\\
                &=-\frac{N}{2}\frac{1}{\mathfrak{h}_k}\iint p_n(x)q_n(y)\left(e^{-h}x^2+e^{h}y^2\right)e^{N\hat{\Phi}(x,y)}dxdy,
            \end{align*}
        and the last integral can be evaluated explicitly by expanding the integrand in biorthogonal polynomials using the recurrence relation.
    \end{proof}
        
    We also establish the following lemma, which will be useful later on in exchanging $n$-dependent sums for finite ones:
        \begin{lemma}\label{lemmaA3}
            Define
            \begin{equation*}
                p_n(x) = x^n - \chi_nx^{n-2} + \OO(x^{n-4}),\qquad\qquad q_n(y) = y^n - \tilde{\chi}_ny^{n-2} + \OO(y^{n-4}),
            \end{equation*}
        Then
            \begin{align}
            \chi_{n} &= \frac{1}{2}\left[A_{n}+nr_{n-1}\right],
        \end{align}
    where 
        \begin{align}
            A_n &= Ne^{-h}\mathfrak{t}s_n+N[r_nr_{n-1}r_{n-2} +r_ns_{n-1} + (r_{n+1}+r_{n}+r_{n-2}+r_{n-3})s_n + r_{n-2}s_{n+1}].
        \end{align}
    Similar formulae hold for the `tilded' variables by replacing $h\to -h$.
        \end{lemma}
    \begin{proof}
        We prove the `untilded' identities; the identities with tildes follow from identical calculations. First, observe that, using the recurrence relations for the biorthogonal polynomials, at order $\OO(x^{n-1})$:
            \begin{equation*}
                \chi_{n+1} = \chi_n + r_n.
            \end{equation*}
        On the other hand, let us write
            \begin{align*}
                p_n'(x) = np_{n-1}(x) + A_np_{n-3}(x) + \OO(x^{n-5}).
            \end{align*}
        Since we can express the monomials $x^{n-1}$, $x^{n-3}$ as
            \begin{align*}
                x^{n-1} = p_{n-1}+\chi_{n-1} p_{n-3}+ \OO(x^{n-5}), \qquad \qquad x^{n-3}= p_{n-3}+\OO(x^{n-5}),
            \end{align*}
        we find that
            \begin{align*}
                p_n'(x) &= nx^{n-1} -(n-2)\chi_nx^{n-3} + \OO(x^{n-5})\\
                        &= n\left(p_{n-1}+\chi_{n-1}p_{n-3} + \OO(x^{n-5})\right) - (n-2)\chi_{n}\left(p_{n-3}+\OO(x^{n-5})\right) + \OO(x^{n-5})\\
                        &= n p_{n-1} + \left(n\chi_{n-1}-(n-2)\chi_n\right)p_{n-3} + \OO(x^{n-5})\\
                        &= n p_{n-1} + \left(2\chi_{n}-nr_{n-1}\right)p_{n-3} + \OO(x^{n-5}),
            \end{align*}
        and so
            \begin{equation*}
                A_n = 2\chi_{n}-nr_{n-1},\qquad \text{ or } \qquad \chi_n = \frac{1}{2}\left(A_n + nr_{n-1}\right).
            \end{equation*}
            Let us now determine $A_n$. Integrating by parts in $x$,
        \begin{align*}
            A_n\mathfrak{h}_{n-3} &= \iint p_n'q_{n-3}e^{N\hat{W}}dxdy = -N\iint (\varsigma y-e^{-h}\mathfrak{t}x-x^3)p_nq_{n-3}e^{N\hat{W}}dxdy\\
                       &= Ne^{-h}\mathfrak{t}s_n\mathfrak{h}_{n-3}+N[r_nr_{n-1}r_{n-2} +r_ns_{n-1} + (r_{n+1}+r_{n}+r_{n-2}+r_{n-3})s_n + r_{n-2}s_{n+1}]\mathfrak{h}_{n-3},
        \end{align*}
        which completes the proof.
    \end{proof}

    With these Lemmas in place, we can proceed with the rest of the proof of Proposition \ref{propC}.
    \begin{proof} \textit{Proposition \ref{propC}: Formulae \eqref{dlogZ_dt} and \eqref{dlogZ_dH}}.
        Using the formula \eqref{BOP-partition-function} for the partition function, Lemma \ref{lemmaA1} yields that
        \begin{align*}
            \log Z_n(\tau,t,H;N) & = \log\left(n!C_{n,N}\right) + \frac{n(n-1)}{2}\log\tau + \sum_{k=0}^{n-1}\log h_k(\tau,t,H;N)\\
            &=\log\left(n!C_{n,N}\right) + \frac{n(n-1)}{2}\log\tau - \frac{n(n+1)}{4}\log t + \sum_{k=0}^{n-1}\log \mathfrak{h}_k\left(\tau t^{-1/2},t^{-1/2},\frac{1}{2}H;N\right).
        \end{align*}
    So, we can write
        \begin{align*}
            \frac{\partial}{\partial \tau} \log Z_n(\tau,t,H;N) &= \frac{n(n-1)}{2\tau}+ t^{-1/2}\sum_{k=0}^{n-1}\frac{1}{\mathfrak{h}_k}\frac{\partial \mathfrak{h}_k}{\partial \varsigma}\left(\tau t^{-1/2},t^{-1/2},\frac{1}{2}H;N\right),\\
            \frac{\partial}{\partial t} \log Z_n(\tau,t,H;N) &= - \frac{n(n+1)}{4t}- \frac{t^{-3/2}}{2}\sum_{k=0}^{n-1}\left[\frac{\tau}{\mathfrak{h}_k}\frac{\partial \mathfrak{h}_k}{\partial \varsigma}\left(\tau t^{-1/2},t^{-1/2},\frac{1}{2}H;N\right) + \frac{1}{\mathfrak{h}_k}\frac{\partial \mathfrak{h}_k}{\partial \mathfrak{t}}\left(\tau t^{-1/2},t^{-1/2},\frac{1}{2}H;N\right) \right]\\
            &=-\frac{n^2}{2t} +\frac{\tau}{2t}\frac{\partial}{\partial \tau} \log Z_n(\tau,t,H;N) - \frac{t^{-3/2}}{2} \sum_{k=0}^{n-1}\frac{1}{\mathfrak{h}_k}\frac{\partial \mathfrak{h}_k}{\partial \mathfrak{t}}\left(\tau t^{-1/2},t^{-1/2},\frac{1}{2}H;N\right),\\
            \frac{\partial}{\partial H} \log Z_n(\tau,t,H;N) &= \frac{1}{2}\sum_{k=0}^{n-1}\frac{1}{\mathfrak{h}_k}\frac{\partial \mathfrak{h}_k}{\partial h}\left(\tau t^{-1/2},t^{-1/2},\frac{1}{2}H;N\right).
        \end{align*}
    Now, by Lemma \ref{lemmaA2}, the sums in the above can be written as
        \begin{align*}
            \sum_{k=0}^{n-1}\frac{1}{\mathfrak{h}_k}\frac{\partial \mathfrak{h}_k}{\partial \mathfrak{t}}\left(\varsigma,\mathfrak{t},h;N\right) &= -\frac{N}{2}\left[e^{-h}r_n+e^{h}\tilde{r}_n\right] -N\left[e^{-h}\chi_n+e^{h}\tilde{\chi}_n\right],\\
            \sum_{k=0}^{n-1}\frac{1}{\mathfrak{h}_k}\frac{\partial \mathfrak{h}_k}{\partial h}\left(\varsigma,\mathfrak{t},h;N\right) &= \frac{N}{2}\mathfrak{t}\left[e^{-h}r_n-e^{h}\tilde{r}_n\right] -N\mathfrak{t}\left[e^{-h}\chi_n-e^{h}\tilde{\chi}_n\right],
        \end{align*}
    where $\chi_n$ are as in Lemma \ref{lemmaA3}. Using the result of \ref{lemmaA3}, we can rewrite these sums as finite sums of the coefficients
    $r_n$, $s_n$, $\tilde{r_n}$, $\tilde{s}_n$, and we can further apply Lemma \ref{lemmaA1} to rewrite this expression subsequently in terms
    of $R_n$, $S_n$, $\tilde{R}_n$, $\tilde{S}_n$.
    \end{proof}

    \section{Riemann-Hilbert problem: Algebraic Results}\label{section:RHPalgebraic}
Before proceeding to asymptotic analysis of the RHP \eqref{BOP-RHP}, we first collect some results on the finite-$n$ problem. Our goal here is to show that \textit{the recurrence coefficients $R_n, f_n$ can be expressed explicitly in terms of the solution to the Riemann-Hilbert problem}. From the work of \cite{KM}, we already have the relation
    \begin{equation*}
        [\boldsymbol{Y}]_{11}(z) = P_n(z),
    \end{equation*}
but the remaining entries of the Riemann-Hilbert problem are left unspecified. We adopt the following notations: we write the large $z$ (resp. $w$) expansions of the biorthogonal polynomials $P_n(z)$, $Q_n(w)$ as
    \begin{equation}\label{chi-kappa-def}
        P_n(z) = z^{n}-\chi_nz^{n-2} + \OO(z^{n-4}),\qquad\qquad Q_n(w) = w^{n}-\tilde{\chi}_nw^{n-2}+\OO(w^{n-4}).
    \end{equation}
\begin{remark}\label{monomial-remark}
    The following representations of monomials will be useful:
    \begin{align*}
        z^{n-1} &= P_{n-1}(z) + \chi_{n-1}P_{n-3}(z) + \OO(z^{n-5}),\\
        z^{n-3} &= P_{n-3}(z) + \OO(z^{n-5}),
    \end{align*}
    Furthermore, if one applies the recurrence relation \eqref{P-recurrence} to the expansion
    \eqref{chi-kappa-def}, then one obtains the identity
        \begin{equation}\label{chi-kappa-recurrence}
            \chi_{n+1} = \chi_{n}+R_n.
        \end{equation}
    Similar identities may be derived for the `tilded' variables.
\end{remark}
Our first lemma re-expresses the variable $\chi_n$, its `tilded' counterpart in terms of the recurrence coefficients of the biorthogonal polynomials.
\begin{lemma}
    Consider $\chi_n$ as defined in \eqref{chi-kappa-def}. Then
        \begin{align}
            \chi_n &= \frac{1}{2}[A_n+nR_{n-1}],
        \end{align}
    where 
        \begin{align}
            A_n &=NS_n + Ne^Ht\big((R_{n+1}+R_{n}+R_{n-2}+R_{n-3})S_n + R_{n-2}S_{n+1} + R_nS_{n-1}\big),
        \end{align}
    A similar identity holds for $\tilde{\chi}_n$, which can be obtained by replacing $H$ by $-H$ in the above.
\end{lemma}
\begin{proof}
The proof is identical to Lemma \ref{lemmaA3}, so we omit it.
%
\end{proof}

Let us now introduce a special family of polynomials $\pi_n(w)$, whose role will soon become clear. For $k=0,1,2$, $j\in \ZZ_{+}$, set
    \begin{equation}
        \pi_{3k+j}(w) := \exp\left[\frac{N}{2}w^2+\frac{N}{4}te^{-H}w^4\right]\frac{d^k}{dw^k}\left( w^j\exp\left[-\frac{N}{2}w^2-\frac{N}{4}te^{-H}w^4\right]\right).
    \end{equation}
\begin{lemma}\label{lemma:pi-expansion}
    The $\pi_{3k+j}(w)$ admits the following expansion in the biorthogonal polynomials $Q_n(w)$:
        \begin{equation*}
            \pi_{3k+j}(w) = \left(-\frac{Nt}{e^{H}}\right)^k\left[ Q_{3k+j}(w) + \left(\frac{ke^H}{t} + \tilde{\chi}_{3k+j}\right)Q_{3k+j-2}(w) + \OO(w^{3k+j-4})\right].
        \end{equation*}
\end{lemma}
\begin{proof}
    As $w\to \infty$, one can show that
        \begin{equation*}
            \pi_{3k+j}(w) = \left(-\frac{Nt}{e^{H}}\right)^k\left[w^{3k+j} + \frac{k e^{H}}{t}w^{3k+j-2} +  \OO(w^{3k+j-4})\right].
        \end{equation*}
    In light of Remark \ref{monomial-remark}, we can re-express the monomials $w^k$ in terms of biorthogonal polynomials. This completes the proof.
\end{proof}

We can now prove the following Proposition:
    \begin{prop}\label{RHP-determination-prop}
        For $n=3k$, the solution to the RHP \eqref{BOP-RHP} is given by
            \begin{align}
                [\boldsymbol{Y}(z)]_{\cdot,1}&= \langle P_{3k}(z), \gamma_{3k-1}\alpha_{3k-1}P_{3k-1}(z) + \gamma_{3k-3}P_{3k-3}(z),\gamma_{3k-2}P_{3k-2}(z),\gamma_{3k-1}P_{3k-1}(z)\rangle^T,\\
                [\boldsymbol{Y}(z)]_{\cdot,j}&= \int_{\Gamma} \frac{\boldsymbol{Y}_{\cdot,1}(\zeta)\hat{w}_{j-2}(\zeta)}{\zeta-z}\frac{d\zeta}{2\pi i},
            \end{align}
        where
            \begin{align}
                \alpha_{3k-1} &:= -\left[\frac{(k-1)e^H}{t} + \tilde{\chi}_{3k+1}\right],\\
                \gamma_{3k+j} &:= -\frac{2\pi i}{h_{3k+j}} \left(\frac{\tau e^H}{t}\right)^k,\qquad j=0,1,2.
            \end{align}
    \end{prop}
\begin{proof}
    The proof is essentially an extension of the ideas of \cite{KM}. We rely heavily on the following \textit{fundamental identity}, which holds for any polynomial $f(x)$, and any $k\in \ZZ_+$, $j=0,1,2$:
        \begin{equation}\label{fundamental-identity}
            \iint f(z)\pi_{3k+j}(w) e^{N\Phi(z,w)}dzdw = (-N\tau)^k\int f(z)z^k\hat{w}_j(z)dz.
        \end{equation}
    This identity can be readily verified using integration by parts. By the results of \cite{KM}, we already have
        \begin{equation*}
            [\boldsymbol{Y}(z)]_{\cdot,1} = \langle P_{3k}(z), P_{3k-1}^{(0)}(z), P_{3k-1}^{(1)}(z), P_{3k-1}^{(2)}(z)\rangle^T,
        \end{equation*}
    where $P_{3k-1}^{(j)}(z)$ are some polynomials of degree $\leq 3k-1$, and $P_{3k}(z)$ is the $(3k)^{th}$ \textit{biorthogonal} polynomial. Let us now determine the polynomials $P_{3k-1}^{(j)}(z)$ explicitly.
    We will furnish the proof of the formula for $P_{3k-1}^{(0)}$; the proof for the other polynomials may be derived in a similar fashion.  Imposing the asymptotic condition \eqref{BOP-normalization} implies the following conditions on this polynomial:
        \begin{align*}
            \int P_{3k-1}^{(0)}(z) z^j \hat{w}_0(z) dz &= 0,\qquad\qquad j=0,...,k-2,\\
            \int P_{3k-1}^{(0)}(z) z^j \hat{w}_1(z) dz &= 0,\qquad\qquad j=0,...,k-1,\\
            \int P_{3k-1}^{(0)}(z) z^j \hat{w}_2(z) dz &= 0,\qquad\qquad j=0,...,k-1,
        \end{align*}
    as well as the normalizing condition
        \begin{equation*}
            -\frac{1}{2\pi i}\int P_{3k-1}^{(0)}(z) z^{k-1}\hat{w}_0(z)dz = 1.
        \end{equation*}
    The identity \eqref{fundamental-identity} implies that the first $3$ sets of conditions are equivalent to
        \begin{align}
            \iint P_{3k-1}^{(0)}(z) w^j e^{N\Phi(z,w)}dzdw &= 0\label{star}\\
            \iint P_{3k-1}^{(0)}(z) \pi_{3k-2}(w) e^{N\Phi(z,w)}dzdw &= 0,\label{star-star}\\
            \iint P_{3k-1}^{(0)}(z) \pi_{3k-1}(w) e^{N\Phi(z,w)}dzdw &= 0,\label{star-star-star}
        \end{align}
    where we have used the fact that $\{\pi_{\ell}(w)\}_{\ell=0}^{3k-4}$ form a basis for the space of polynomials of degree $\leq 3k-4$. The conditions \eqref{star} imply that
        \begin{equation*}
            P_{3k-1}^{(0)}(z) = AP_{3k-1}(z) + BP_{3k-2}(z) + CP_{3k-3}(z),
        \end{equation*}
    for some undetermined constants $A,B,C$. Applying Lemma \ref{lemma:pi-expansion} to \eqref{star-star}, we find that
        \begin{align*}
            0 &= \iint P_{3k-1}^{(0)}(z)\pi_{3k-2}(w) e^{N\Phi(z,w)}dzdw\\
              &=\left(-\frac{Nt}{e^{H}}\right)^{k-1}\iint [AP_{3k-1}(z) + BP_{3k-2}(z) + CP_{3k-3}(z)]\left[ Q_{3k-2}(w) + \OO(y^{3k-4})\right] e^{N\Phi(z,w)}dzdw\\
              &=\left(-\frac{Nt}{e^{H}}\right)^kBh_{3k-2},
        \end{align*}
    which implies $B = 0$. Again applying Lemma \ref{lemma:pi-expansion}, this time to \eqref{star-star-star},
        \begin{align*}
            0 &= \iint P_{3k-1}^{(0)}(z)\pi_{3k-1}(w) e^{N\Phi(z,w)}dzdw\\
            &=\left(-\frac{Nt}{e^{H}}\right)^{k-1}\iint [AP_{3k-1}(z) + CP_{3k-3}(z)][Q_{3k-1}(w) + \left(\frac{(k-1)e^H}{t} + \tilde{\chi}_{3k+1}\right)Q_{3k+3}(w) + \OO(y^{3k-5})]e^{N\Phi(z,w)}dzdw\\
            &=\left(-\frac{Nt}{e^{H}}\right)^{k-1}\left[Ah_{3k-1} + C\left(\frac{(k-1)e^H}{t} + \tilde{\chi}_{3k+1}\right)h_{3k-3} \right],
        \end{align*}
    which gives a relation between $A$ and $C$. Finally, applying \eqref{fundamental-identity} to the normalizing condition, we get that
        \begin{align*}
            1 &= -\frac{1}{2\pi i}\int [AP_{3k-1}(z)+CP_{3k-3}(z)]z^{k-1}\hat{w}_0(z)dz\\
            &=-\frac{\left(-N\tau\right)^{-k+1}}{2\pi i}\iint[AP_{3k-1}(z)+CP_{3k-3}(z)]\pi_{3k-3}(w)e^{N\Phi(z,w)}dzdw\\
            &=-\left(\frac{t}{\tau e^H}\right)^{k-1}\frac{1}{2\pi i}\iint[AP_{3k-1}(z)+CP_{3k-3}(z)][Q_{3k-3}(w)+\OO(w^{3k-5})]e^{N\Phi(z,w)}dzdw\\
            &=-\left(\frac{t}{\tau e^H}\right)^{k-1}\frac{1}{2\pi i}C h_{3k-3}.
        \end{align*}
    Thus, we can determine $A,C$ to be
        \begin{equation*}
            C = -\frac{2\pi i}{h_{3k-3}}\left(\frac{\tau e^H}{t}\right)^{k-1},\qquad\qquad A = \frac{2\pi i}{h_{3k-1}}\left(\frac{\tau e^H}{t}\right)^{k-1}\left[\frac{(k-1)e^H}{t} + \tilde{\chi}_{3k+1}\right].
        \end{equation*}
\end{proof}

Using the above results, we obtain at last the following Proposition:
    \begin{prop}\label{prop:RHP-coeff-expression}
        Let $n$ be a multiple of $3$, and set
            \begin{equation}
                Y^{(1)} := \Res_{z=\infty} \boldsymbol{Y}(z)z^{-n\hat{\sigma}},
            \end{equation}
        where $\hat{\sigma}$ is as in \eqref{hat-sigma-def}. Then
            \begin{align}
                f_n &= \frac{\tau e^{H}}{t} [Y^{(1)}]_{12}[Y^{(1)}]_{41},\\
                \tilde{R}_n &= \frac{\tau e^{H}}{t}\left([Y^{(1)}]_{43} - [Y^{(1)}]_{32}\right).
            \end{align}
    \end{prop}
\begin{proof}
    The proof follows by direct calculation. Using the result of Proposition \ref{RHP-determination-prop}, one can compute that
        \begin{align*}
            [Y^{(1)}]_{12} &= -\frac{1}{2\pi i}\int P_{3k}(z) z^{k} \hat{w}_0(z)dz,\\
            [Y^{(1)}]_{41} &= -\frac{2\pi i}{h_{3k-1}} \left(\frac{\tau e^H}{t}\right)^{k-1},\\
            [Y^{(1)}]_{43} &= \frac{1}{h_{3k-1}} \left(\frac{\tau e^H}{t}\right)^{k-1}\int P_{3k-1}(z)z^{k} \hat{w}_1(z)dz,\\
            [Y^{(1)}]_{32} &=  \frac{1}{h_{3k-2}} \left(\frac{\tau e^H}{t}\right)^{k-1}\int P_{3k-2}(z)z^{k} \hat{w}_0(z)dz.
        \end{align*}
    We can explicitly evaluate each of the above integrals by applying the fundamental identity \eqref{fundamental-identity}, and using Lemma
    \ref{lemma:pi-expansion} and biorthogonality. We first compute that
        \begin{align*}
            [Y^{(1)}]_{12} &= -\frac{1}{2\pi i}\int P_{3k}(z) z^{k} \hat{w}_0(z)dz
                           = -\frac{(-N\tau)^{-k}}{2\pi i}\iint P_{3k}(z)\pi_{3k}(w) e^{N\Phi(z,w)}dzdw\\
                           &=-\frac{1}{2\pi i}\left(\frac{t}{\tau e^H}\right)^k \iint P_{3k}(z)[Q_{3k}(w) + \OO(w^{3k-2})]e^{N\Phi(z,w)}dzdw
                           =-\frac{h_{3k}}{2\pi i}\left(\frac{t}{\tau e^H}\right)^k.
        \end{align*}
    So, it follows that
        \begin{equation*}
           \frac{\tau e^{H}}{t} [Y^{(1)}]_{12}[Y^{(1)}]_{41} = \bigg[-\frac{h_{3k}}{2\pi i}\left(\frac{t}{\tau e^H}\right)^k\bigg]\bigg[-\frac{2\pi i}{h_{3k-1}} \left(\frac{\tau e^H}{t}\right)^{k-1}\bigg] = \frac{h_{3k}}{h_{3k-1}} = f_{n},
        \end{equation*}
    since $n = 3k$. Now, by similar arguments,
        \begin{align*}
            [Y^{(1)}]_{43} &= \frac{1}{h_{3k-1}} \left(\frac{\tau e^H}{t}\right)^{k-1}\int P_{3k-1}(z)z^{k} \hat{w}_1(z)dz
            =\frac{(-N\tau)^{k}}{h_{3k-1}} \left(\frac{\tau e^H}{t}\right)^{k-1}\iint P_{3k-1}(z) \pi_{3k+1}(w) e^{N\Phi(z,w)}dzdw\\
            &= \frac{t}{\tau e^{H}}\left(\frac{k e^H}{t}+\tilde{\chi}_{3k+1}\right),
        \end{align*}
    and
        \begin{align*}
            [Y^{(1)}]_{32} &=\frac{1}{h_{3k-2}} \left(\frac{\tau e^H}{t}\right)^{k-1}\int P_{3k-2}(z)z^{k} \hat{w}_0(z)dz=\frac{t}{\tau e^{H}}\left(\frac{k e^H}{t}+\tilde{\chi}_{3k}\right).
        \end{align*}
    Finally, 
        \begin{align*}
            \frac{\tau e^{H}}{t}\bigg([Y^{(1)}]_{43} - [Y^{(1)}]_{32}\bigg) &= \frac{\tau e^{H}}{t}\bigg(\bigg[\frac{t}{\tau e^{H}}\left(\frac{k e^H}{t}+\tilde{\chi}_{3k+1}\right)\bigg] - \bigg[\frac{t}{\tau e^{H}}\left(\frac{k e^H}{t}+\tilde{\chi}_{3k}\right)\bigg]\bigg)=\tilde{\chi}_{3k+1}-\tilde{\chi}_{3k} = \tilde{R}_n,
        \end{align*}
    where we have used the relation $\tilde{\chi}_{n+1}-\tilde{\chi}_{n} = \tilde{R}_n$.
\end{proof}

Obviously, since $R_n(\tau,t,H;N) = \tilde{R}_n(\tau,t,-H;N)$, we can also express $R_n$ in terms of the solution to the Riemann-Hilbert problem, and so we have accomplished our goal.

\section{Spectral curves}\label{section:SpectralCurve}
In this section, we construct the spectral curve and associated functions needed for the second transformation (the so-called ``g-function transformation''). However, as we shall see, the \textit{true} spectral curve is not quite the correct object, as the degeneration of the
spectral curve as one approaches criticality is not something readily handled by standard Riemann-Hilbert analysis. This kind of obstacle has appeared before in the literature \cite{Bleher-Kuijlaars,DKZ,DG}; the solution is to construct a \textit{modified} spectral curve, which approaches the \textit{true} spectral curve when the parameters of the model approach the multicritical point. We first show that
the true spectral curve can be constructed with the desired properties, and then proceed to construct the modified spectral curve.

\subsection{Critical spectral curve}
Let us first construct the \textit{true} spectral curve associated to our problem, at the multicritical point \eqref{multicritical-point}.
This has essentially already been performed in Part I of this work \cite{DHL1}. The main idea
was to construct an appropriate uniformization $(z(u),Y(u))$ of the spectral curve, so that the corresponding functions $\Omega(z) = \int Y(z) dz$ have the correct asymptotics. Then, using a general lemma (cf. \cite{DHL1}, Lemma 3.4), along with a local computation around the branch points, one can show that the inequalities necessary for the opening of lenses hold. An overview of the construction of the spectral curve can be found in \cite{DHL1}. We simply present the uniformization of this curve here, and state
the relevant results which hold. Consider the following pair of rational functions:
    \begin{align}
        z(u) &:= \sqrt{\frac{6}{5}}\left(u+\frac{2}{u}-\frac{1}{3u^3}\right),\label{z-coord-critical}\\
        Y(u) =z(u^{-1})&:= \sqrt{\frac{6}{5}}\left(u^{-1}+2u-\frac{1}{3}u^3\right).\label{Y-coord-critical}
    \end{align}
The Riemann surface parameterized by the functions $(z(u),Y(u))$ is called the \textit{spectral curve}. This curve is $4$-sheeted, considered as a branched 
covering of the plane over the $z$-coordinate. It has $2$ finite cubic branch points at $z =\pm \alpha = \sqrt{\frac{6}{5}}$, in addition to a cubic branch point at infinity. We glue the sheets of the spectral curve as follows: sheets $1$ and $2$ are glued along
the interval $[-\alpha,\alpha]$, and sheet $2$ is additionally glued to sheet $3$ (respectively, $4$) along the interval $(-\infty,-\alpha]$ (respectively, along the interval $[\alpha,\infty$). This is depicted in Figure \ref{fig:Critical-Curve}. We further define uniformization mappings $u_j(z), j=1,2,3,4$, as functions inverse to $z(u)$, i.e. so that $z(u_j(z))$ is the identity mapping in sheet $j$, away from the branch cuts.

\begin{figure}[t]
    \centering
    \begin{overpic}[scale=.23]{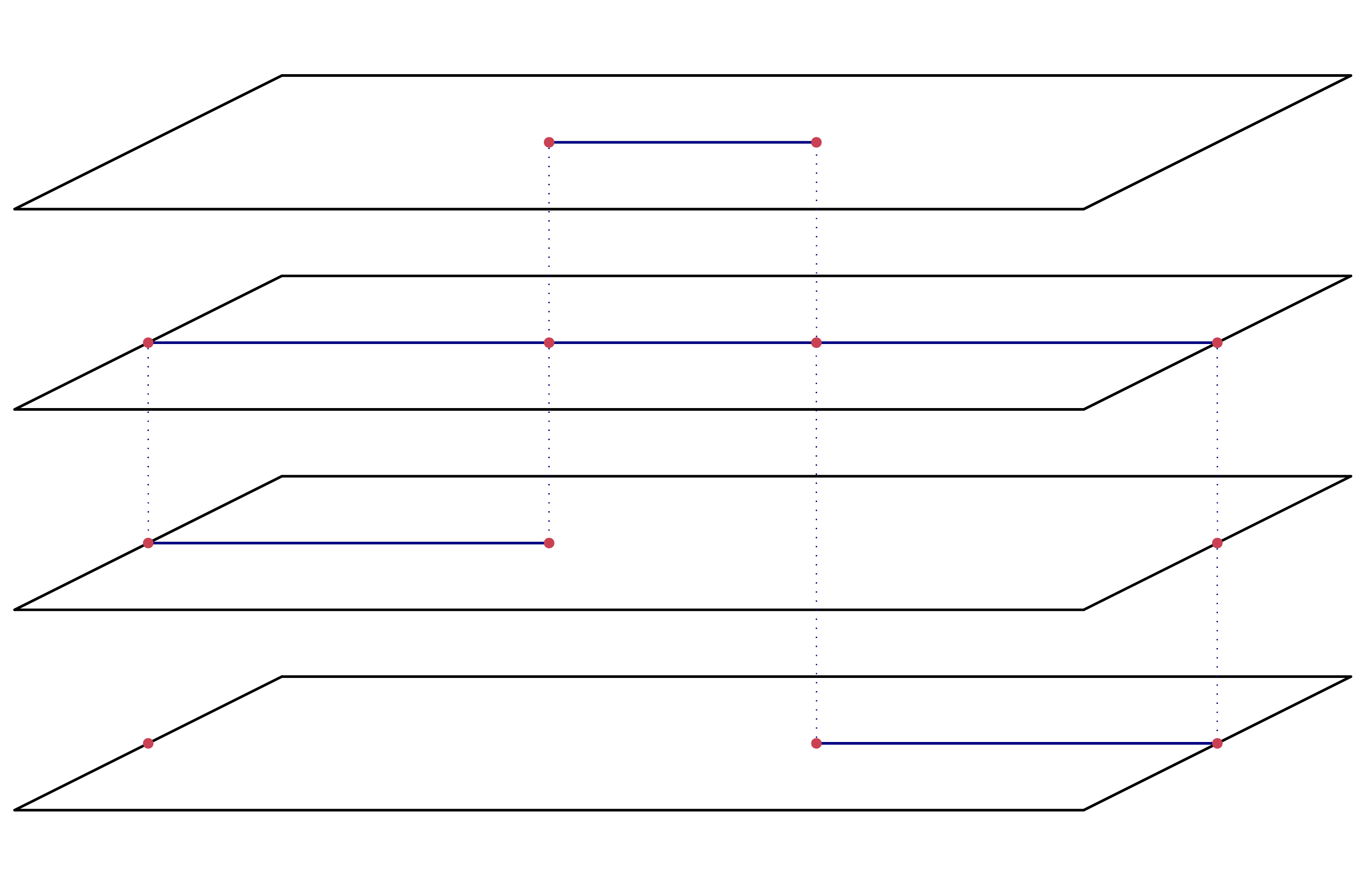}
        \put (85,54) {\large \bf I}
        \put (85,40) {\large \bf II}
        \put (85,25) {\large \bf III}
        \put (85,11) {\large \bf IV}
        \put (57,54) {$+\alpha$}
        \put (37,54) {$-\alpha$}
    \end{overpic}
    \caption{The branching structure of the multicritical spectral curve.}
    \label{fig:Critical-Curve}
\end{figure}
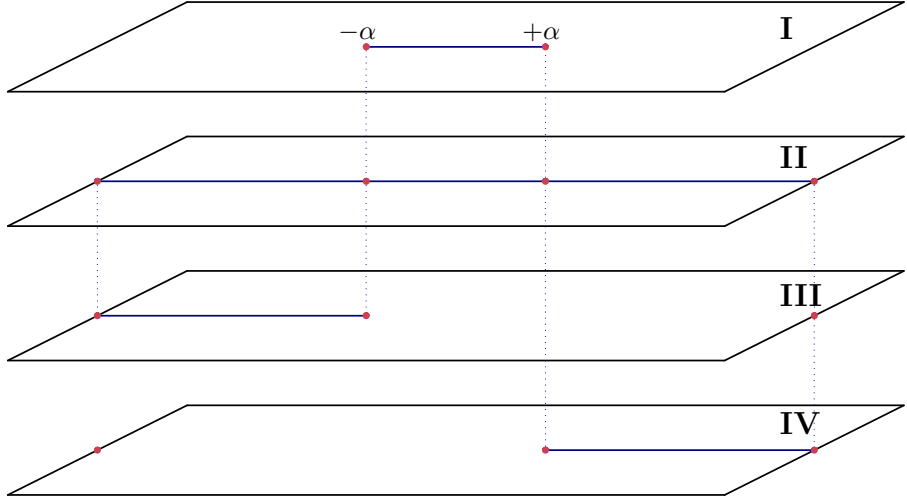

There is a special function we will need for the later Riemann-Hilbert analysis, which will `match' the exponentially growing terms in
the asymptotics of $\boldsymbol{Y}$. We set
    \begin{align}
        \tau_c \Omega(u) &= \tau_c \int Y dz = \tau_c \int Y(u) z'(u) du =  -\frac{1}{40}u^4 + \frac{2}{5}u^2 -\frac{3}{40}u^{-4} - \log u.\label{Omega-uniformizer}
    \end{align}
We can consider $\Omega$ as a function on the spectral curve in the variable $z$ by setting its value on each sheet to be
    \begin{equation}
        \Omega_j(z) := \Omega(u_j(z)),\qquad\qquad j = 1,...,4.
    \end{equation}
If we expand these functions at infinity, we obtain the following proposition, which is again proven in \cite{DHL1}:

    \begin{prop}[\cite{DHL1}, Proposition 2.3]
        \textit{(Asymptotics of $\Omega_j(z)$ as $z\to\infty$).} The functions $\Omega_j(z)$ admit the following expansions:
            \begin{align}
                    \tau_c \Omega_1(z) &= \frac{e^{H_c}t_c}{4}z^4 + \frac{1}{2}z^2 - \log z + \ell_0 + \OO(z^{-2}), \label{Omega-Sheet-1-infty}\\
                    \tau_c \Omega_{2}(z) &= \begin{cases}
                       -\frac{3\omega^2}{4}\frac{\tau_c^{4/3}}{(-e^{-H_c}t_c)^{1/3}}z^{4/3} - \frac{\omega}{2} \frac{\tau_c^{2/3}}{(-e^{-H_c}t_c)^{2/3}} z^{2/3} + \frac{1}{3}\log z + \ell_1 + \frac{\omega^2 C_1}{z^{2/3}} +  \OO\left(\frac{1}{z^{4/3}}\right), \qquad \text{Im } z > 0, \\
                        -\frac{3\omega}{4}\frac{\tau_c^{4/3}}{(-e^{-H_c}t_c)^{1/3}}z^{4/3} - \frac{\omega^2}{2} \frac{\tau_c^{2/3}}{(-e^{-H_c}t_c)^{2/3}} z^{2/3} + \frac{1}{3}\log z + \ell_1 + \frac{\omega C_1}{z^{2/3}} +  \OO\left(\frac{1}{z^{4/3}}\right), \qquad \text{Im } z < 0,
                        \end{cases} \label{Omega-Sheet-2-infty}\\
                        \tau_c \Omega_{3}(z) &= -\frac{3}{4}\frac{\tau_c^{4/3}}{(-e^{-H_c}t_c)^{1/3}}z^{4/3} - \frac{1}{2} \frac{\tau_c^{2/3}}{(-e^{-H_c}t_c)^{2/3}} z^{2/3} + \frac{1}{3}\log z + \ell_1 + \frac{C_1}{z^{2/3}} +  \OO\left(\frac{1}{z^{4/3}}\right),\label{Omega-Sheet-3-infty}\\
                        \tau_c \Omega_{4}(z) &= \begin{cases}
                        -\frac{3\omega}{4}\frac{\tau_c^{4/3}}{(-e^{-H_c}t_c)^{1/3}}z^{4/3} - \frac{\omega^2}{2} \frac{\tau_c^{2/3}}{(-e^{-H_c}t_c)^{2/3}} z^{2/3} + \frac{1}{3}\log z + \ell_1 + \frac{\omega C_1}{z^{2/3}} +  \OO\left(\frac{1}{z^{4/3}}\right), \qquad \text{Im } z > 0, \\
                        -\frac{3\omega^2}{4}\frac{\tau_c^{4/3}}{(-e^{-H_c}t_c)^{1/3}}z^{4/3} - \frac{\omega}{2} \frac{\tau_c^{2/3}}{(-e^{-H_c}t_c)^{2/3}} z^{2/3} + \frac{1}{3}\log z + \ell_1 + \frac{\omega^2 C_1}{z^{2/3}} +  \OO\left(\frac{1}{z^{4/3}}\right),\qquad \text{Im } z < 0.
                        \end{cases}\label{Omega-Sheet-4-infty}
                \end{align}
            Here, the constants $C_1,\ell_0,\ell_1$ are
                \begin{equation}\label{critical-C-L}
                    C_1 =\left(\frac{2}{1875}\right)^{1/3}, \qquad 
                    \ell_0 = \log\sqrt{\frac{6}{5}} -\frac{11}{6},\qquad
                    \ell_1 = -\frac{1}{3}\log\sqrt{\frac{6}{5}}-\frac{21}{10}.
                \end{equation}
    \end{prop}

\begin{lemma}[\cite{DHL1}, Lemma 2.1]\label{global-inequalities}
    Define the vector field
    \begin{equation}
        \hat{{\bf n}} := \frac{\nabla \text{Im }z(u)}{||\nabla \text{Im }z(u)||}.
    \end{equation} 
(note that this is the normal vector to the preimages of the branch cuts in the uniformizing plane under the mapping $z(u)$).
    Then, the function
        \begin{equation}
            \nabla \text{Re }\Omega(u) \cdot \hat{{\bf n}} = \frac{\partial}{\partial n} \text{Re }\Omega(u)
        \end{equation}
    is of constant sign on each connected component of the preimages of the branch cuts.
\end{lemma}
The proof again may be found in \cite{DHL1}. The key idea behind this lemma is that it reduces the proof of the inequalities we need to a $2$-step procedure: $1.$ Show that the curves defining the branch cuts and $\text{Im } z(u) = 0$ and the corresponding components of
$\nabla \text{Re }\Omega(u) \cdot \hat{{\bf n}} = \text{Im } Y(u) = 0$ do not intersect, and $2.$ show that the inequalities we need hold locally at the branch points. Once we have shown that the inequalities hold near the branch points, since $1.$ holds these
inequalities extends to the entire branch cuts.

We thus have only to check that the inequalities we need hold near the branch points.

\begin{prop}[\cite{DHL1}, Analog of Proposition 3.5]\label{local-expansions-pm-alpha}
    The functions $\Omega_j(z)$ have the following expansions about the branch points $z = \pm \alpha$:
    Near $z = +\alpha$:
    \begin{align} 
            \Omega_{1}(z) &= \eta_+(z) -\frac{48}{35}\left(\frac{2}{\alpha}\right)^{7/3}(z-\alpha)^{7/3} + \OO\left((z-\alpha)^{8/3}\right),\label{Sheet1-alpha-plus}\\
            \Omega_{2}(z) &= \begin{cases}
                            \eta_+(z) -\frac{48 \omega^2}{35}\left(\frac{2}{\alpha}\right)^{7/3}(z-\alpha)^{7/3} + \OO\left((z-\alpha)^{8/3}\right), \qquad \text{Im } z > 0, \\
                            \eta_+(z) -\frac{48 \omega}{35}\left(\frac{2}{\alpha}\right)^{7/3}(z-\alpha)^{7/3} + \OO\left((z-\alpha)^{8/3}\right), \qquad \text{Im } z < 0,
                            \end{cases}\label{Sheet2-alpha-plus}\\
            \Omega_{4}(z) &= \begin{cases}
                            \eta_+(z) -\frac{48 \omega}{35}\left(\frac{2}{\alpha}\right)^{7/3}(z-\alpha)^{7/3} + \OO\left((z-\alpha)^{8/3}\right), \qquad \text{Im } z > 0, \\
                            \eta_+(z) -\frac{48 \omega^2}{35}\left(\frac{2}{\alpha}\right)^{7/3}(z-\alpha)^{7/3} + \OO\left((z-\alpha)^{8/3}\right), \qquad \text{Im } z < 0,
                            \end{cases}\label{Sheet4-alpha-plus}
        \end{align}
    and $\Omega_3(z)$ is regular in a neighborhood of $z = +\alpha$.\\
    Near $z = -\alpha$:
    \begin{align} 
            \Omega_{1}(z) &= \begin{cases}
                            \eta_-(z) +\frac{48 \omega}{35}\left(\frac{2}{\alpha}\right)^{7/3}(z+\alpha)^{7/3} + \OO\left((z+\alpha)^{8/3}\right), \qquad \text{Im } z > 0, \\
                            \eta_-(z) +\frac{48 \omega^2}{35}\left(\frac{2}{\alpha}\right)^{7/3}(z+\alpha)^{7/3} + \OO\left((z+\alpha)^{8/3}\right), \qquad \text{Im } z < 0,
                            \end{cases}\label{Sheet1-alpha-minus}\\
            \Omega_{2}(z) &= \begin{cases}
                            \eta_-(z) + \frac{48 \omega^2}{35}\left(\frac{2}{\alpha}\right)^{7/3}(z+\alpha)^{7/3} + \OO\left((z+\alpha)^{8/3}\right), \qquad \text{Im } z > 0, \\
                            \eta_-(z) + \frac{48 \omega}{35}\left(\frac{2}{\alpha}\right)^{7/3}(z+\alpha)^{7/3} + \OO\left((z+\alpha)^{8/3}\right), \qquad \text{Im } z < 0,
                            \end{cases}\label{Sheet2-alpha-minus}\\
            \Omega_{3}(z) &= \eta_-(z) + \frac{48}{35}\left(\frac{2}{\alpha}\right)^{7/3}(z+\alpha)^{7/3} + \OO\left((z+\alpha)^{8/3}\right), \label{Sheet4-alpha-minus}
        \end{align}
     and $\Omega_4(z)$ is regular in a neighborhood of $z = -\alpha$.
     Here, $\eta_{\pm}(z):= \frac{6}{5} \pm \alpha (z\mp\alpha) -  \frac{1}{2}(z\mp\alpha)^2$.
\end{prop}
\begin{proof}
    At the multicritical point, the uniformizing coordinate takes the form \eqref{z-coord-critical}. Expanding in the uniformizing coordinate around the images of the branch points 
    $u = \pm 1$, we obtain that
        \begin{equation*}
            \Omega(z(u)) - \eta_{\pm}(z(u)) = \mp\frac{48}{35}(u\mp 1)^7 + \OO((u\mp 1)^8).
        \end{equation*}
    Expanding $z(u)$ in the uniformizing coordinate,
        \begin{equation*}
            z(u) \mp \alpha = \frac{\alpha}{2} (u\mp 1)^3 + \OO( (u\pm 1)^4 ).
        \end{equation*}
    Thus, locally, $u\pm1$ has the expansion (with to be determined choice of branch)
        \begin{equation*}
            u\pm1 \sim \left(\frac{2}{\alpha}\right)^{1/3}(z \pm \alpha)^{1/3}.
        \end{equation*}
    The choice of branch is determined by the argument of $u\pm 1$ in the uniformizing plane lying in the correct sheet, along with the specification that $(z\pm\alpha)^{1/3}$ is always taken with the principal branch cut. This yields the result.
\end{proof}

    \begin{prop}[\cite{DHL1}, Analog of Proposition 3.6]\label{multicritical-inequalities-1}
        Let $\Omega_j(z) = \phi_j(z) + i\psi_j(z)$, $j = 1,2,3,4$. Then, the following inequalities hold:
            \begin{enumerate}
                \item  $\phi_4(z) - \phi_2(z) > 0$ for $z$ in a lens around $[\alpha,\infty)$, 
                \item  $\phi_3(z) - \phi_2(z) > 0$ for $z$ in a lens around $(-\infty,-\alpha]$,
                \item  $\phi_2(z) - \phi_{1}(z) > 0$ for $z$ in a lens around $[-\alpha,\alpha]$.
            \end{enumerate}
    \end{prop}
    \begin{proof} 
    Since, by the preceding lemmas, the direction of steepest descent of $\text{Re } \Omega(u)$ is constant along each of the connected components of $\text{Im } z = 0$,  we have only to check that the necessary inequalities hold near the branch points $z = \pm \alpha$. 
    
    Near $z = +\alpha$, using Proposition \ref{local-expansions-pm-alpha}, we have that
        \begin{equation*}
            \phi_4(z) - \phi_2(z) = \text{Re }\left[\Omega_4(z) - \Omega_2(z)\right] = 
                \begin{cases}
                    \text{Re } \big[\frac{48(\omega^2-\omega)}{35}\left(\frac{2}{\alpha}\right)^{7/3}(z-\alpha)^{7/3} + \OO\left((z-\alpha)^{8/3}\right)\big], \qquad \text{Im } z > 0, \\
                    \text{Re } \big[\frac{48(\omega-\omega^2)}{35}\left(\frac{2}{\alpha}\right)^{7/3}(z-\alpha)^{7/3} + \OO\left((z-\alpha)^{8/3}\right)\big], \qquad \text{Im } z < 0, \\
                \end{cases}
        \end{equation*}
    from which we can read off that $\phi_4(z) - \phi_2(z) > 0$ in the sector $|\arg (z-\alpha)|<\frac{3\pi}{7}$ for $|z-\alpha|$ sufficiently small; this proves $1.$
    
    Near $z = -\alpha$, again using Proposition \ref{local-expansions-pm-alpha}, we have that
            \begin{equation*}
                \phi_3(z) - \phi_2(z) = \text{Re } \big[\Omega_3(z)- \Omega_2(z)\big] = 
                \begin{cases}
                    \text{Re } \big[-\frac{48(\omega^2-1)}{35}\left(\frac{2}{\alpha}\right)^{7/3}(z+\alpha)^{7/3} + \OO\left((z+\alpha)^{8/3}\right)\big], \qquad \text{Im } z > 0, \\
                    \text{Re } \big[-\frac{48(\omega-1)}{35}\left(\frac{2}{\alpha}\right)^{7/3}(z+\alpha)^{7/3} + \OO\left((z+\alpha)^{8/3}\right)\big], \qquad \text{Im } z < 0, \\
                \end{cases}
            \end{equation*}
        from which we can read off that $\phi_3(z) - \phi_2(z) > 0$ in the sector $\frac{4\pi}{7}<|\arg (z+\alpha)|\leq \pi$ for $|z+\alpha|$ sufficiently small, so that $2.$ follows.
        
        Finally, we determine the behavior of $\phi_2(z)-\phi_1(z)$ around either branch point; near $z=+\alpha$, say. We have that
            \begin{equation*}
                \phi_2(z) - \phi_1(z) = \text{Re } \big[\Omega_2(z)- \Omega_1(z)\big] = 
                \begin{cases}
                    \text{Re } \big[\frac{48(1-\omega^2)}{35}\left(\frac{2}{\alpha}\right)^{7/3}(z-\alpha)^{7/3} + \OO\left((z-\alpha)^{8/3}\right)\big], \qquad \text{Im } z > 0, \\
                    \text{Re } \big[\frac{48(1-\omega)}{35}\left(\frac{2}{\alpha}\right)^{7/3}(z-\alpha)^{7/3} + \OO\left((z-\alpha)^{8/3}\right)\big], \qquad \text{Im } z < 0; \\
                \end{cases}
            \end{equation*}
        this expansion tells us that $\phi_2(z)-\phi_1(z) > 0$ in particular in the sector $\frac{4\pi}{7}<|\arg (z-\alpha)|\leq \pi$, for $|z-\alpha|$ sufficiently small.
        A similar analysis near $z=-\alpha$ shows that $\phi_2(z)-\phi_1(z) > 0$ in the sector $|\arg (z-\alpha)| <\frac{3\pi}{7}$ for $|z+\alpha|$ sufficiently small. 
        Thus, we have established locally the inequalities $1.$ through $3.$; as a consequence of Proposition \ref{global-inequalities}, the inequalities hold in a neighborhood of each of the segments of the real line.
    \end{proof}
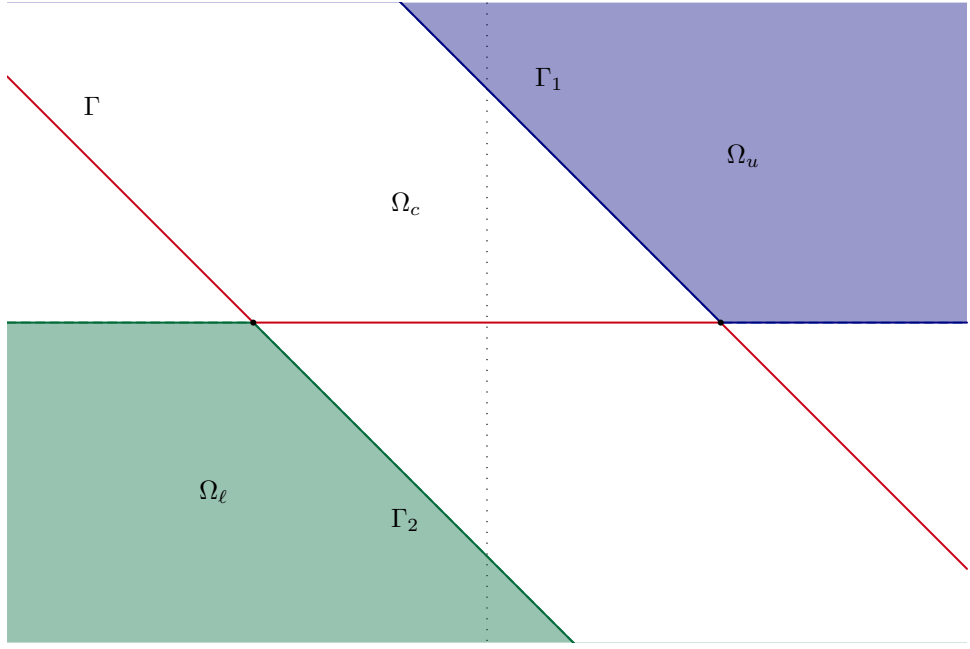
\begin{figure}
    		\begin{center}
		    \begin{overpic}[scale=.2]{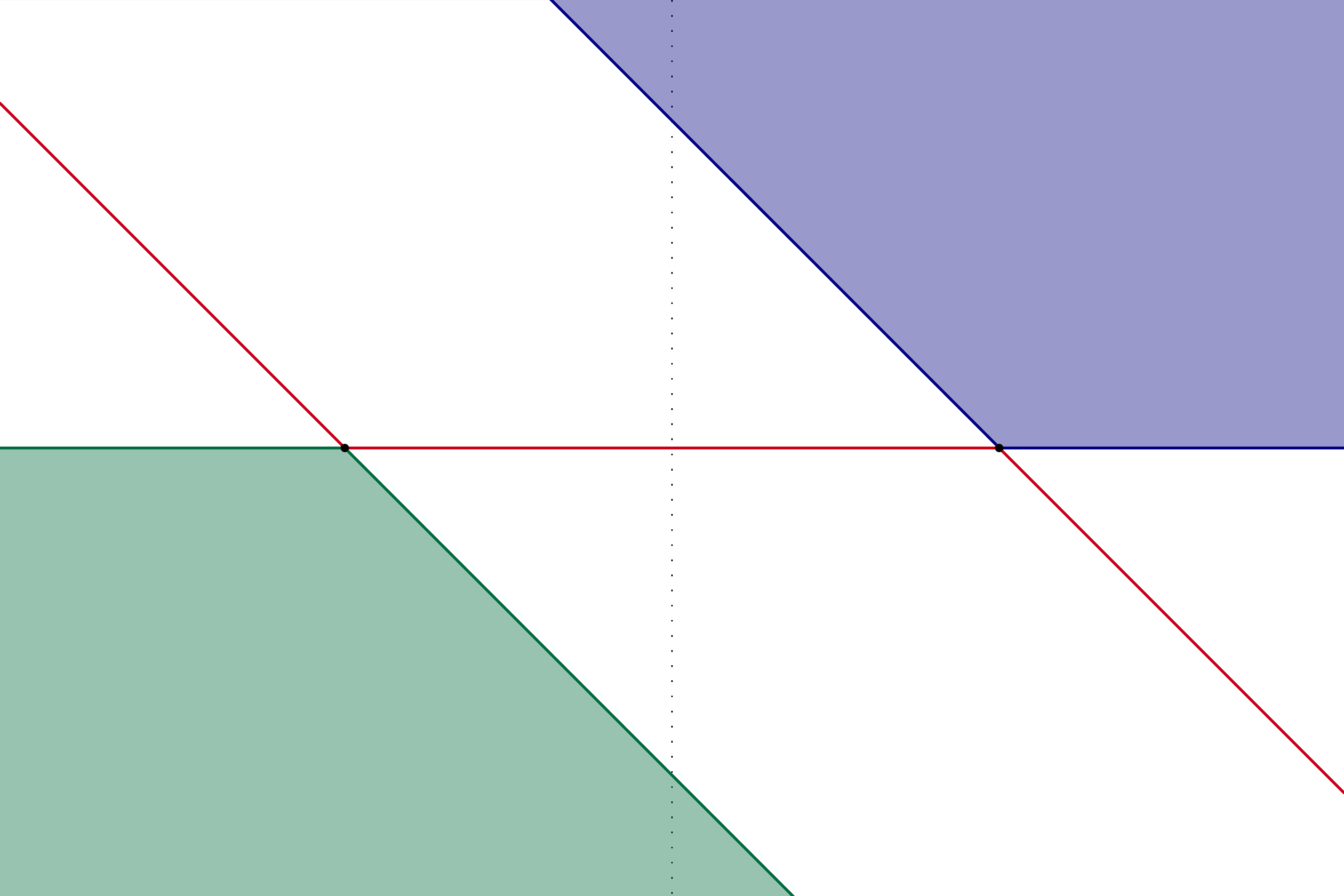}
		      \put (8,55) {$\Gamma$}
		      \put (40,12) {$\Gamma_2$}
		      \put (55,58) {$\Gamma_1$}
            \put (20,15) {$\Omega_{\ell}$}
            \put (75,50) {$\Omega_{u}$} 
            \put (40,45) {$\Omega_{c}$}
		    \end{overpic}
		\end{center}
    \caption{The contours $\Gamma$, $\Gamma_1$, and $\Gamma_2$, along with the regions $\Omega_u$, $\Omega_c$, and $\Omega_{\ell}$.}
    \label{fig:Gamma-Figure}
\end{figure}
We also will need some inequalities to hold off the real axis; for this, we will need to redefine the contour $\Gamma$, and additionally define a 
couple new contours\footnote{Note that all definitions here involve the parameter $\alpha$, the endpoint of the branch cuts. In the subsequent sections, $\alpha = \alpha(\eta,\mu,\nu;n)$ will be a function of the deformation variables, and we will slightly modify the contours to pass through this parameter-dependent $\alpha$ without further comment.}. We redefine the contour $\Gamma$ (where $\boldsymbol{Y}$ has jumps) to start at $e^{\frac{3\pi i}{4}}\cdot\infty$, passes through $z=-\alpha$, and continues along the real axis until it reaches $z=+\alpha$, where it then goes off to infinity in the direction $e^{-\frac{\pi i}{4}}$. In addition, we also define two new contours $\Gamma_1$ and $\Gamma_2$ as follows. The contour $\Gamma_{2}$ begins at $-\infty$, and travels along the real axis until it reaches $z=-\alpha$, and then goes off to infinity again in the sector $-\frac{3\pi}{7}<\arg(z+\alpha) < 0$. We define $\Gamma_1$ to start at infinity in the sector $0 < \arg (z-\alpha) < \frac{3\pi}{7}$, approaches and passes through $z=\alpha$, and then goes off to infinity again along the positive real axis. By following the arguments of \cite{DHL1}, one can show directly that:
    \begin{prop}[\cite{DHL1}, Analog of Proposition 3.7] \label{multi-critical-inequalities-2} 
        Let $\Omega_j(z) = \phi_j(z) + i\psi_j(z)$, $j = 1,2,3,4$. Then, the following inequalities hold:
        \begin{enumerate}
            \item  $\phi_2(z) - \phi_4(z) > 0$ for $z \in \Gamma_1 \cap \{\text{Im } z > 0\}$,
            \item  $\phi_2(z) - \phi_3(z) > 0$ for $z \in \Gamma_2 \cap \{\text{Im } z < 0\}$,
            \item  $\phi_1(z) - \phi_2(z) > 0$ for $z \in \Gamma \setminus \{\text{Im } z = 0\}$.
        \end{enumerate}
    \end{prop}
For convenience, we additionally define the region below $\Gamma_2$ to be $\Omega_{\ell}$, the region above $\Gamma_1$ to be $\Omega_{u}$, and the remaining part of the plane to be $\Omega_c$. These regions and the contours $\Gamma$, $\Gamma_1,\Gamma_2$ are depicted in Figure \ref{fig:Gamma-Figure}. 

\subsection{Modified spectral curve}
With the relevant properties of the critical spectral curve established, we turn our attention to the construction of the \textit{modified spectral curve}. Construction of such a curve addressed before in the literature \cite{Bleher-Kuijlaars,DKZ,DG}. The main idea is to define a identify
a curve which carries the same asymptotic behavior as the true spectral curve as we approach criticality, but is characterized by the condition that the branch points are fixed at $\pm \alpha$ over the varying parameters. This yields a curve whose behavior near its branch points is wrong; however, since we are going to patch in local parametrices anyways, this is admissible.

To ease notations in the what follows, we will use the shorthand 
    \begin{align*}
        \vec{t} := (\tau,t,H), \qquad\qquad \text{ and } \qquad\qquad \vec{t}_c = (\tau_c,t_c,H_c).
    \end{align*}
We will only need to calculate the modified spectral curve for $\vec{t}$ in a small neighborhood of the multicritical curve, i.e. with $\vec{t}$ as in Definition \ref{scaling-variables-def}. As it turns out, the conditions uniquely fix for us a choice of modified curve:

\begin{lemma}
   There exist functions 
     \begin{equation}
         A(\vec{t}), B(\vec{t}), J(\vec{t}), K(\vec{t}),
         L(\vec{t}), M(\vec{t}),
     \end{equation} 
     (real) analytic in a neighborhood of $\vec{t}=\vec{t}_c$, such that if we define
        \begin{align}
            z(u;\vec{t}) &:= A(\vec{t}) \left(u + \frac{2}{u} - \frac{1}{3u^3}\right), \label{z-mod-spectralcurve-1}\\
            Y(u;\vec{t}) &:= B(\vec{t}) \left(\frac{1}{u} + J(\vec{t}) u + K(\vec{t})u^3\right) \nonumber\\ &+ L(\vec{t})\left(\frac{1}{u-1} +\frac{1}{u+1} \right) + M(\vec{t}) \left(\frac{1}{(u-1)^2} - \frac{1}{(u+1)^2}\right),\label{z-mod-spectralcurve-2}
        \end{align}
then, the following Conditions hold in a neighborhood of $\vec{t}=\vec{t}_c$, whenever $\vec{t} = (\tau,t,H)$ are as in Definition \ref{scaling-variables-def}:

\begin{enumerate}
    \item[C1.] When $\vec{t}=\vec{t}_c$, the functions $z(u;\vec{t}_c)$, $Y(u;\vec{t}_c)$ coincide with the parameterization of the true multicritical spectral curve \eqref{z-coord-critical}, \eqref{Y-coord-critical}. In other words,
        \begin{equation}
            A(\vec{t}_c)=\sqrt{\frac{6}{5}}, \quad B(\vec{t}_c)=\sqrt{\frac{6}{5}}, \quad J(\vec{t}_c) = 2, \quad K(\vec{t}_c) = -\frac{1}{3}, \quad L(\vec{t}_c) = 0, \quad M(\vec{t}_c) = 0.
        \end{equation}
    \item[C2.] For $z = z(u;\vec{t})$, $Y = Y(u;\vec{t})$, the expansions
        \begin{align}
            te^{H} z(u)^3 + z(u) - \frac{1}{z(u)} -\tau Y(u) &= \OO(u^{-2}), \qquad u\to \infty,\\
            te^{-H} Y(u)^3 + Y(u) - \frac{1}{Y(u)} -\tau z(u) &= \OO(u^{2}), \qquad u\to 0,
        \end{align}
    hold, uniformly in a neighborhood of $\vec{t}=\vec{t}_c$,
    \item[C3.] Define $\Omega(u) := \int Y(u) z'(u) du$, uniformization coordinates $u_j(z)$, $j=1,...,4$ as
    before, and 
        \begin{equation}\label{modified-Omega-functions}
            \Omega_j(z;\vec{t}) := \Omega(u_j(z)).
        \end{equation}
    Then, defining $\alpha = \alpha(\vec{t}) := z(1;\vec{t})$, the expansion of $\Omega_j(z;\vec{t})$, $j=1,2,4$ about 
    $z = \alpha$ has vanishing coefficient of the term $(z-\alpha)^{4/3}$. Similarly, the expansion of $\Omega_j(z;\vec{t})$, $j=1,2,3$ about 
    $z = -\alpha$ has vanishing coefficient of the term $(z+\alpha)^{4/3}$.
\end{enumerate}
\end{lemma}
\begin{proof}
Assume the form of $z(u;\vec{t}),Y(u;\vec{t})$ is as in Equations \eqref{z-mod-spectralcurve-1}, \eqref{z-mod-spectralcurve-2}. Additionally, suppose $(\tau,t,H)$ are as in Definition \ref{scaling-variables-def}. The requirement that Condition C2 holds imposes the following equations on $A(\vec{t})$ through $M(\vec{t})$:
    \begin{align*}
        0 &= e^Ht A^3 - BK\tau,\\
        0 &= 6e^HtA^3 - BJ\tau + A,\\
        0 &= 2A + 11e^HtA^3 - 1/A - \tau(B + 2L),\\
        0 &= te^{-H}B^3 + \frac{1}{3}\tau A,\\
        0 &= B + 3e^{-H}tB^2(BJ - 2L + 4M) - 2\tau A,\\
        0 &= BJ - 2L + 4M + 3tB e^{-H}\bigg[(J^2+K)B^2 - 2\big([1+2J]L - 4[J+1]M\big)B + (4M-2L)^2\bigg] - \frac{1}{B} - \tau A.
    \end{align*}
One of the above equations is redundant; the remaining $5$ equations can be solved uniquely for for $A$, $J$, $K$, $L$, and $M$ as rational 
functions of $t,\tau,H$, and the function $B(\vec{t})$:
    \begin{align*}
        A &= -\frac{3te^{-H}B^3}{\tau}, \qquad \qquad J = -\frac{3te^{-2H}B^2}{\tau^4}(54B^6t^3 + e^H\tau^2),\\
        K &= -\frac{27t^4e^{-2H}B^8}{\tau^4},\qquad\qquad L= -e^{-2H}\frac{891B^{12} t^5 + 18B^6e^{H}t^2\tau^2 + 3B^4e^{2H} t\tau^4 - e^{3H}\tau^4}{6B^3t \tau^4},\\
        M &= - e^{-2H}\frac{405 B^{12}t^5 + 9B^6e^{H}t^2\tau^2 + 9B^4e^{2H}t\tau^4 + B^2e^{3H}\tau^4 - e^{3H}\tau^4}{12\tau^4B^3 t}.
    \end{align*}
Now, Condition C3 requires that the expansion of $\Omega_j(z)$, $j=1,2,4$ about $z=\alpha$ has no term of the form $(z-\alpha)^{4/3}$, and similarly
the expansion of $\Omega_j(z)$, $j=1,2,3$ about $z= -\alpha$ should have no term of the form $(z+\alpha)^{4/3}$. Because of the symmetry of the
coefficients of the poles in $Y(u)$, we only need to verify the first condition holds; the condition at $z= -\alpha$ will then hold automatically. The
requirement that the coefficient of the term $(z-\alpha)^{4/3}$ in the expansion of the $\Omega_j(z)$'s implies the equation
    \begin{equation*}
       \frac{2^{1/3}(5e^{3H}\tau^4 + 7e^{3H}\tau^4B^2 + 603e^{2H}\tau^4 tB^4 + 1575e^Ht^2\tau^2B^6 + 132111t^5B^{12})}{1536B^3e^{2H}t\tau^4(-tB^3/(\tau e^H))^{1/3}} = 0,
    \end{equation*}
which in particular will hold if the numerator vanishes identically. The theorem we are trying to prove has thus been reduced to proving the existence
of an implicit function $B(\vec{t})$ of the equation
    \begin{equation} \label{implicit-equation-B}
        \Phi(\vec{t},B) := 5e^{3H}\tau^4 + 7e^{3H}\tau^4B^2 + 603e^{2H}\tau^4 tB^4 + 1575e^Ht^2\tau^2B^6 + 132111t^5B^{12} = 0,
    \end{equation}
When $\vec{t}=\vec{t}_c$, $B = \sqrt{\frac{6}{5}}$ is the unique positive solution to the above equation. The implicit function theorem guarantees the existence of a function $B(\vec{t})$ in a neighborhood of $\vec{t}=\vec{t}_c$ satisfying Equation \eqref{implicit-equation-B},
provided that 
    \begin{equation*}
        \frac{\partial \Phi}{\partial B} \bigg|_{\vec{t}=\vec{t}_c,B=\sqrt{6/5}} \neq 0.
    \end{equation*} 
Indeed, one can calculate that $\frac{\partial \Phi}{\partial B} \vert_{\vec{t}=\vec{t}_c,B=\sqrt{6/5}} = -3\sqrt{\frac{6}{5}}$, and so an implicit function exists. 
Since $A,J,K,L$, and $M$ all depend rationally on $B$, $\vec{t}$, we can obtain Taylor expansions of these functions as
functions of $\vec{t}$ as well. 
When $\vec{t} =\vec{t}_c$, it is immediately apparent that $z(u;\vec{t}_c)$, $Y(u;\vec{t}_c)$ coincide with the parametrization of 
the true multicritical spectral curve, so that Condition C1 is satisfied. This completes the proof of the Lemma.
\end{proof}

\begin{remark}
Explicitly, the first few terms in the expansion of $B$ are
    \begin{equation*}
        B = \sqrt{\frac{6}{5}}\left(1 + \frac{2279}{320}(t-t_c) + \frac{581}{576}(\tau-\tau_c) + \frac{9}{64}H + \OO(2)\right).
    \end{equation*}
Evaluated on the choice of $\vec{t}$ given by \eqref{scaling-variables-def}, as $n\to \infty$, $B$ has the expansion
    \begin{equation*}
        B = \sqrt{\frac{6}{5}}\left(1 - \frac{9}{5}c_5\eta n^{-2/7} + \frac{9}{25}c_5^2\eta^2 n^{-4/7} + \frac{9}{64}c_2\mu n^{-5/7}\eta +\OO(n^{-6/7})\right).
    \end{equation*}
It is also interesting to observe that
    \begin{align*}
        L = -\sqrt{\frac{6}{5}}\left(2c_2\mu n^{-5/7} + \OO(n^{-6/7})\right),\qquad\qquad M = -\sqrt{\frac{6}{5}}\left(\frac{4}{3}c_1\nu n^{-6/7} + \OO(n^{-1})\right),
    \end{align*}
i.e. first and second order poles in $Y$ `couple' to the parameters $\mu,\nu$, respectively.
\end{remark}

Note that since the $z$-coordinate of the modified spectral curve is the same (up to the overall normalizing factor) the $z$-coordinate of the \textit{true} critical
spectral curve, for $n$ sufficiently large. Thus, we define sheets $1$--$4$ in an identical manner. We also define
the uniformization coordinates $u_j(z)$ for $j=1,...,4$ in the same way. The below lemma then follows at once from condition C2:
\begin{prop}
For $n$ sufficiently large, and with $\vec{t}$ as in Definition \ref{scaling-variables-def}, the functions $\Omega_j(z;\vec{t})$ have the following large $z$ asymptotics:
        \begin{align}
            \tau \Omega_1(z;\vec{t}) &= \frac{e^{H}t}{4}z^4 + \frac{1}{2}z^2 - \log z + \ell_0 + \OO(z^{-2}),\label{modified-omega-1-asymptotic}\\
            \tau \Omega_{2}(z;\vec{t}) &= \begin{cases}
                -\frac{3\omega^2}{4}\frac{\tau^{4/3}}{(-e^{-H}t)^{1/3}}z^{4/3} - \frac{\omega}{2} \frac{\tau^{2/3}}{(-e^{-H}t)^{2/3}} z^{2/3} + \frac{1}{3}\log z + \ell_1 + \frac{\omega^2 C_1}{z^{2/3}} +  \OO\left(\frac{1}{z^{4/3}}\right), \qquad \text{Im } z > 0, \\
                -\frac{3\omega}{4}\frac{\tau^{4/3}}{(-e^{-H}t)^{1/3}}z^{4/3} - \frac{\omega^2}{2} \frac{\tau^{2/3}}{(-e^{-H}t)^{2/3}} z^{2/3} + \frac{1}{3}\log z + \ell_1 + \frac{\omega C_1}{z^{2/3}} +  \OO\left(\frac{1}{z^{4/3}}\right), \qquad \text{Im } z < 0,
                    \end{cases}\\
                \tau \Omega_{3}(z;\vec{t}) &= -\frac{3}{4}\frac{\tau^{4/3}}{(-e^{-H}t)^{1/3}}z^{4/3} - \frac{1}{2} \frac{\tau^{2/3}}{(-e^{-H}t)^{2/3}} z^{2/3} + \frac{1}{3}\log z + \ell_1 + \frac{C_1}{z^{2/3}} +  \OO\left(\frac{1}{z^{4/3}}\right),\\
                \tau \Omega_{4}(z;\vec{t}) &= \begin{cases}
                -\frac{3\omega}{4}\frac{\tau_c^{4/3}}{(-e^{-H}t)^{1/3}}z^{4/3} - \frac{\omega^2}{2} \frac{\tau^{2/3}}{(-e^{-H}t)^{2/3}} z^{2/3} + \frac{1}{3}\log z + \ell_1 + \frac{\omega C_1}{z^{2/3}} +  \OO\left(\frac{1}{z^{4/3}}\right), \qquad \text{Im } z > 0, \label{modified-omega-4-asymptotic}\\
                -\frac{3\omega^2}{4}\frac{\tau^{4/3}}{(-e^{-H}t)^{1/3}}z^{4/3} - \frac{\omega}{2} \frac{\tau^{2/3}}{(-e^{-H}t)^{2/3}} z^{2/3} + \frac{1}{3}\log z + \ell_1 + \frac{\omega^2 C_1}{z^{2/3}} +  \OO\left(\frac{1}{z^{4/3}}\right),\qquad \text{Im } z < 0.
                    \end{cases}
        \end{align}
        Note that the constants $C_1 = C_1(\vec{t})$, $\ell_0 = \ell_0(\vec{t})$, and $\ell_1 = \ell_1(\vec{t})$ depend on $\vec{t}$.
\end{prop}

Crucially, the modified spectral curve has been constructed so that outside a sufficiently small neighborhood of $z = \pm \alpha$, the lensing 
inequalities we proved earlier still hold. This seemingly introduces a new issue: the inequalities necessary to open lenses do not hold nearby the
branch points. However, since we are going to introduce local discs around the branch points anyways, we actually do not need these inequalities to
hold here. This is the objective of introducing the modified curve: we now have a workable form for the local parametrices discs about the branch 
points (to be constructed later in this section), while retaining the required lensing inequalities outside of these discs. We now formally state a Lemma that guarantees the lensing inequalities still hold outside of local discs about $z = \pm \alpha$.

\begin{prop}\label{global-inequalities-modified}
    For fixed $\epsilon > 0$ sufficiently small, there exists $r_{\epsilon} >0$ such that, if we set $\Delta := \{z : |z-\alpha| < r_{\epsilon} \} \cup\{z : |z+\alpha| < r_{\epsilon} \}$, and $\Omega_j(z;\vec{t}) = \phi_j(z;\vec{t}) + i\psi_j(z;\vec{t})$, $j=1,2,3,4$, the inequalities 
            \begin{enumerate}
                \item  $\phi_4(z;\vec{t}) - \phi_2(z;\vec{t}) > 0$ for $z \in \big(\Gamma_{1,u}\cup\Gamma_{1,l}\big) \setminus \Delta$,
                \item  $\phi_3(z;\vec{t}) - \phi_2(z;\vec{t}) > 0$ for $z \in \big(\Gamma_{2,u}\cup\Gamma_{2,l}\big) \setminus \Delta$,
                \item  $\phi_2(z;\vec{t}) - \phi_{1}(z;\vec{t}) > 0$ for $z \in \big(\Gamma_{u}\cup\Gamma_{l}\big) \setminus \Delta$,
            \end{enumerate}
    hold, whenever $n$ is sufficiently large. 
    Furthermore, with the same notations as above, the inequalities
        \begin{enumerate}
            \item  $\phi_2(z;\vec{t}) - \phi_4(z;\vec{t}) > 0$ for $z \in \Gamma_1 \cap (\{\text{Im } z > 0\} \setminus \Delta)$,
            \item  $\phi_2(z;\vec{t}) - \phi_3(z;\vec{t}) > 0$ for $z \in \Gamma_2 \cap (\{\text{Im } z < 0\} \setminus \Delta)$,
            \item  $\phi_1(z;\vec{t}) - \phi_2(z;\vec{t}) > 0$ for $z \in (\Gamma \setminus \{\text{Im } z = 0\})\setminus \Delta$,
        \end{enumerate}
    hold, whenever $n$ is sufficiently large. 
\end{prop}
\begin{proof}
    We shall prove the first of the above $6$ inequalities, i.e. that 
    $\phi_4(z;\vec{t}) - \phi_2(z;\vec{t}) > 0$ for $z \in \big(\Gamma_{1,u}\cup\Gamma_{1,l}\big) \setminus \Delta$. The proof of the remaining inequalities follows from identical arguments. 
    
    Denote the minimum of the function $\phi_4(z;\vec{t}) - \phi_2(z;\vec{t})$ on $\big(\Gamma_{1,u}\cup\Gamma_{1,l}\big) \setminus \Delta$ by $m = m(\vec{t})$. $m$ is a continuous function of $n^{-1/7}$; furthermore, since we have taken the minimum outside $\Delta$, there exists a constant $c>0$  such that $m(\vec{t}_c) > c > 0$, by Proposition \ref{multicritical-inequalities-1}. By continuity of $m$, it follows that $\phi_4(z;\vec{t}) - \phi_2(z;\vec{t}) > 0$ on $\big(\Gamma_{1,u}\cup\Gamma_{1,l}\big) \setminus \Delta$ for $n$ sufficiently large.
\end{proof}

\section{Proof of Theorem \ref{propA}: Riemann-Hilbert analysis}\label{section:RHP}
We now turn our attention to the Riemann-Hilbert problem which defines the biorthogonal polynomials. The goal of this section is to prove Theorem \ref{propA}. Once this is established, we the proof of the main theorem will be complete.

\subsection{Steepest descent: preliminary transformation}

With the all of the necessary theorems for the modified spectral curve now in our hands, we are ready to begin the steepest descent analysis. Most of this analysis is directly analogous to the construction in \cite{DHL1}. The only substantial difference here is the
calculation of the parametrices in the final transformation. Before coming to the transformation which replaces the exponentially growing part of $\boldsymbol{Y}$ with oscillatory jumps, we will need a `preparatory' transformation. The idea of this transformation is introduced in \cite{DK1}. The main observation is that the weights appearing in the jump matrix satisfy a modified version of the so-called Pearcey equation:
    \begin{equation}
        \frac{t}{N^2\tau^2}e^{-H} f'''(z) + f'(z) - N\tau^2zf(z) = 0,
    \end{equation}
which has solutions given by the Pearcey-type integrals
    \begin{equation}\label{Pearcey-integrals}
        w_j(z) := \int_{\gamma_j} \exp\big[ N(\tau z w  \underbrace{-\frac{1}{2}w^2 - \frac{t}{4}e^{-H}w^4}_{-V(e^{-H}w)} )\big] dw,
    \end{equation}
where the contours $\gamma_j$ are as in Figure \ref{Gamma-curves}.

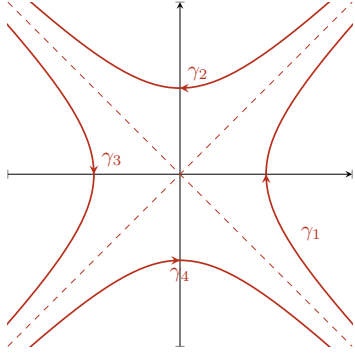
\begin{figure}
		\begin{center}
		\begin{tikzpicture}[scale=.8]
			\begin{axis}[
			axis lines=middle,
				xticklabels=\empty,
				yticklabels=\empty,
				axis equal image,
    				xmin = -2, xmax = 2,
    				ymin = -2, ymax = 2]
				 \addplot [BrickRed,thick,domain=-2:2,samples=50,
					postaction={decorate},
				 	decoration={markings, 
		 				mark=at position 0.5  with {\arrow{stealth}}}							 
				 ]({cosh(x)},{sinh(x)}) node[above right,pos=.4] {$\gamma_1$};
				 \addplot [BrickRed,thick,domain=-2:2,samples=50,
					postaction={decorate},
				 	decoration={markings, 
		 				mark=at position 0.5  with {\arrowreversed{stealth}}}				 
				 ]({sinh(x)},{cosh(x)}) node[above right,pos=.5] {$\gamma_2$};
				 \addplot [BrickRed,thick,domain=-2:2,samples=50,
					postaction={decorate},
				 	decoration={markings, 
		 				mark=at position 0.5  with {\arrowreversed{stealth}}}				 
				 ]({-cosh(x)},{sinh(x)}) node[above right,pos=.5] {$\gamma_3$};
				 \addplot [BrickRed,thick,domain=-2:2,samples=50,
					postaction={decorate},
				 	decoration={markings, 
		 				mark=at position 0.5  with {\arrow{stealth}}}				 
				 ]({sinh(x)},{-cosh(x)}) node[below,pos=.5] {$\gamma_4$};
				 \addplot [BrickRed,dashed,domain=-2:2,samples=50]({x},{x});
				 \addplot [BrickRed,dashed,domain=-2:2,samples=50]({x},{-x});
			\end{axis}
			\end{tikzpicture}
		\end{center}
    \caption{The contours $\gamma_j$, $j = 1,...,4$.}
    \label{Gamma-curves}
\end{figure}

We can therefore analyze the asymptotics of these functions using classical steepest descent analysis. This has already been performed in \cite{DHL1}, and the calculation here is no different. We thus present the result here without further comment. 

Define row vectors 
\begin{equation}
    \vec{w}_j(z) := \left(w_j(z), \frac{w'_j(z)}{N\tau},\frac{w''_j(z)}{(N\tau)^2}\right),
\end{equation}
(here, $' = \frac{d}{dz}$), and define the following $3\times 3$ matrix $\boldsymbol{W}(z)$:

    \begin{equation}
        \boldsymbol{W}(z) =
        \begin{cases}
            \begin{pmatrix}
            -\vec{w}_2(z) \\
            \vec{w}_3(z) \\
            \vec{w}_1(z)
            \end{pmatrix},& z \in\Omega_u, \\
            \begin{pmatrix}
                \vec{w}_3(z) + \vec{w}_4(z) \\
                \vec{w}_3(z) \\
                \vec{w}_1(z)
            \end{pmatrix}, & z \in \Omega_c,\\
            \begin{pmatrix}
            \vec{w}_4(z) \\
            \vec{w}_3(z) \\
            \vec{w}_1(z)
            \end{pmatrix}, & z \in \Omega_{\ell}.
        \end{cases}
    \end{equation}
(The regions $\Omega_{u},\Omega_{\ell},$ and $\Omega_c$ are depicted in Figure \ref{fig:Gamma-Figure}). We then have the following proposition:
\begin{prop}[\cite{DHL1}, Proposition 4.4]
    $\boldsymbol{W}(z)$ is the unique solution to the following Riemann-Hilbert problem:
    
    \begin{enumerate}
        \item $\boldsymbol{W}(z)$ is analytic in $\CC \setminus (\Gamma_1\cup \Gamma_2)$, with boundary values
            \begin{equation}
                \boldsymbol{W}_{+}(z) = 
                \begin{cases}
                    \begin{pmatrix}
                    1 & 0 & 1 \\
                    0 & 1 & 0 \\
                    0 & 0 & 1
                \end{pmatrix}\boldsymbol{W}_{-}(z), & z \in \Gamma_1,\\
                    \begin{pmatrix}
                    1 & 1 & 0 \\
                    0 & 1 & 0 \\
                    0 & 0 & 1
                \end{pmatrix}\boldsymbol{W}_{-}(z), & z \in \Gamma_2,\\
                \end{cases}
            \end{equation}
        \item At infinity, $\boldsymbol{W}(z)$ is normalized as
    \begin{equation}
        \boldsymbol{W}(z) = c_N \cdot e^{N\Theta(z)} A(z) B(z)\hat{K}\bigg[\mathbb{I}_{3\times 3} + \OO\left(\frac{1}{z}\right)\bigg], \qquad z\to \infty,
    \end{equation}
    where the constant $c_N:= i\sqrt{\frac{2\pi}{3N}}\exp\left(\frac{N}{6t}e^H\right)$, the matrix $\Theta(z)$ is defined to be
    \begin{equation} \label{Lambda-Matrix}
        \Theta(z) = 
        \begin{cases}
            \text{diag}(\lambda_3(z), \lambda_1(z),\lambda_2(z)), & \text{Im } z > 0,\\
            \text{diag}(\lambda_2(z), \lambda_1(z),\lambda_3(z)), & \text{Im } z <0,
        \end{cases} 
    \end{equation}
where the functions $\lambda_k(z)$ are defined by the \textit{exact} formulas
    \begin{equation}\label{little-theta-hat}
        \lambda_k(z) = -\frac{3\omega^{k-1}}{4} \frac{\tau^{4/3}}{(-te^{-H})^{1/3}}z^{4/3} - \frac{\omega^{1-k}}{2}\frac{\tau^{2/3}}{(-te^{-H})^{2/3}}z^{2/3}.
    \end{equation}
and finally, the matrices $A(z)$, $B(z)$, and $\hat{K}$ are given by
     \begin{equation}
        A(z) = 
            \begin{cases}
            \begin{psmallmatrix}
                -\omega & 1 & \omega^2\\
                -1 & 1 & 1\\
                -\omega^2 & 1 & \omega
            \end{psmallmatrix}, & \text{Im } z > 0,\\
            \begin{psmallmatrix}
                 \omega^2 & -1 & -\omega\\
                -1 & 1 & 1\\
                -\omega & 1 & \omega^2\\
            \end{psmallmatrix}, & \text{Im } z < 0.
            \end{cases}
        \end{equation}
        \begin{equation}
        B(z) = 
        \begin{psmallmatrix}
            z^{-1/3} & 0 & 0\\
            0 & 1 & 0\\
            0 & 0 & z^{1/3}
        \end{psmallmatrix},
    \end{equation}
    \begin{equation}
        \hat{K} = 
        \begin{psmallmatrix}
            (-te^{-H})^{-1/6}\tau^{-1/3} & 0 & -\frac{n+27 te^{-H}}{54(-te^{-H})^{13/6}\tau^{1/3}}\\
            0 & (-te^{-H})^{-1/2} & 0\\
            0 & 0 & -(-te^{-H})^{-5/6}\tau^{1/3}
        \end{psmallmatrix}.
    \end{equation}
    \end{enumerate}
\end{prop}
We now define the transformation $\boldsymbol{Y}\mapsto\boldsymbol{X}$:

\begin{equation}
        \boldsymbol{X}(z) := 
            \begin{psmallmatrix}
                1 & 0 \\
                0 & c_N \hat{K}
            \end{psmallmatrix}
            \boldsymbol{Y}(z)
                \begin{pmatrix}
                    e^{-NV(z)} & 0 \\
                    0 & \boldsymbol{W}^{-1}(z)
                \end{pmatrix}.
    \end{equation}
$\boldsymbol{X}(z)$ is then a piecewise analytic on $\CC\setminus \left(\Gamma_1 \cup \Gamma_2 \cup \Gamma\right)$, and moreover:
    \begin{prop}[Analog of \cite{DHL1}, Proposition 4.8]
        $\boldsymbol{X}(z)$ solves the following RHP:
           \begin{equation}
            \boldsymbol{X}_{+}(z) = \boldsymbol{X}_{-}(z) \times 
            \begin{cases}
            \mathbb{I}-E_{24}, & z \in \Gamma_1,\vspace{3mm}\\
            \mathbb{I}-E_{23}, & z \in \Gamma_2,\vspace{3mm}\\
            \mathbb{I}+E_{12}, & z \in \Gamma,\vspace{3mm}\\
            \end{cases}
        \end{equation}
        
        Subject to the normalization condition
    \begin{equation}\label{X-asymptotics}
        \boldsymbol{X}(z) = 
            \bigg[\mathbb{I} + \OO\left(\frac{1}{z}\right)\bigg]\begin{psmallmatrix}
                1 & 0\\
                0 & B^{-1}(z) A^{-1}(z) 
            \end{psmallmatrix}
            z^{n\hat{\sigma}}e^{-N\Lambda(z)},
        \end{equation}
        where $\hat{\sigma}$ is as in \eqref{hat-sigma-def}, $\Lambda(z)$ is defined as the diagonal matrix
    \begin{equation} 
        \Lambda(z) =
        \begin{pmatrix}
            V(z) & 0\\
            0 & \Theta(z)
        \end{pmatrix},
    \end{equation}
where $\Theta(z)$ is the diagonal $3\times 3$ matrix defined by \eqref{Lambda-Matrix}.
    \end{prop}
The proof is identical to Proposition 4.8 of \cite{DHL1}, so we omit it here. 

\subsection{Steepest descent: $g$-functions and lens opening}
We are now ready to move on to next transformation, which will remove exponentially growing terms in $\boldsymbol{X}$. We define
    \begin{equation}\label{G-matrix}
        \boldsymbol{G}(z) := \text{diag }(e^{n\tau\Omega_1(z;\vec{t})},e^{n\tau\Omega_2(z;\vec{t})},e^{n\tau\Omega_3(z;\vec{t})},e^{n\tau\Omega_4(z;\vec{t})}),
    \end{equation}
where the functions $\Omega_j(z;\vec{t})$ are the \textit{modified} $\Omega_j$-functions arising from the modified spectral curve, 
as defined in Equation \eqref{modified-Omega-functions}. We set
    \begin{equation}
        \boldsymbol{U}(z) := \left[\mathbb{I} - nC_1\cdot E_{24}\right] e^{-nL}\boldsymbol{X}(z)\boldsymbol{G}(z),
    \end{equation}
where $\boldsymbol{G}(z)$ is as defined above, $C_1$ is the coefficient of $z^{-2/3}$ in the large $z$-expansion of the $\Omega_j(z)$'s, and $L$ is the diagonal constant (in $z$) matrix (see Equations \eqref{modified-omega-1-asymptotic}--\eqref{modified-omega-4-asymptotic})
    \begin{equation}
        L := \text{diag } (\ell_0, \ell_1, \ell_1, \ell_1). 
    \end{equation}
$\boldsymbol{U}(z)$ is analytic in $\CC \setminus (\Gamma \cup \Gamma_1 \cup \Gamma_2 )$, and:
\begin{prop}[\cite{DHL1}, Analog of Proposition 4.10]
        The function $\boldsymbol{U}(z)$ is the unique solution to the following Riemann-Hilbert problem:
            \begin{align}\label{U-jumps}
                \boldsymbol{U}_{+}(z) &= \boldsymbol{U}_{-}(z) \times \nonumber\\
                    &\begin{cases}
                        \vspace{2mm}
                        \mathbb{I} - E_{24}e^{-n\tau[\Omega_2(z) - \Omega_4(z)]}, 
                        & z \in \Gamma_1 \cap \{\text{Im } z > 0\},\\
                        \vspace{2mm}
                        \begin{psmallmatrix}
                            1 & 0 & 0 & 0 \\
                            0 & e^{-n\tau[\Omega_{2,-}(z)-\Omega_{4,-}(z)]} & 0 & -1\\
                            0 & 0 & 1 & 0 \\
                            0 & 0 & 0 & e^{n\tau[\Omega_{2,-}(z)-\Omega_{4,-}(z)]}
                        \end{psmallmatrix}, 
                        & z \in \Gamma_1 \cap \{\text{Im } z = 0\},\\
                        \vspace{2mm}
                        \mathbb{I} - E_{23}e^{-n\tau[\Omega_2(z) - \Omega_3(z)]}, 
                        & z \in \Gamma_2 \cap \{\text{Im } z < 0\},\\
                        \vspace{2mm}
                        \begin{psmallmatrix}
                            1 & 0 & 0 & 0 \\
                            0 & e^{-n\tau[\Omega_{2,-}(z) - \Omega_{3,-}(z)]} & -1 & 0\\
                            0 & 0 & e^{n\tau[\Omega_{2,-}(z) - \Omega_{3,-}(z)]} & 0 \\
                            0 & 0 & 0 & 1
                        \end{psmallmatrix}, 
                        & z \in \Gamma_2 \cap \{\text{Im } z = 0\},\\
                        \vspace{2mm}
                        \mathbb{I} + E_{12}e^{-n\tau[\Omega_1(z) - \Omega_2(z)]},
                        & z\in \Gamma \setminus \{\text{Im } z = 0\},\\
                        \vspace{2mm}
                        \begin{psmallmatrix}
                            e^{-n\tau[\Omega_{1,-}(z)-\Omega_{2,-}(z)]} & 1 & 0 & 0\\
                            0 & e^{n\tau[\Omega_{1,-}(z)-\Omega_{2,-}(z)]} & 0 & 0\\
                            0 & 0 & 1 & 0 \\
                            0 & 0 & 0 & 1
                        \end{psmallmatrix},
                        & z\in \Gamma \cap \{\text{Im } z = 0\}.
                    \end{cases}
            \end{align}
        The asymptotics of $\boldsymbol{U}(z)$ are given by
 \begin{equation} \label{U-asymptotics}
            \boldsymbol{U}(z) =
            \bigg[\mathbb{I} + \OO\left(\frac{1}{z^{1/3}}\right)\bigg]\begin{psmallmatrix}
                1 & 0\\
                0 & B^{-1}(z) A^{-1}(z) 
            \end{psmallmatrix}, \qquad \qquad z\to \infty.
    \end{equation}
\end{prop}
\begin{proof}
    The proof is almost identical to that of \cite{DHL1}, but there are some slight technical differences, which we would like to point out, as is the first place in the Riemann-Hilbert analysis that differs from that of \cite{DHL1}. This difference, as far as this 
    transformation is concerned, is superficial: all of the properties required of the functions $\Omega_j(z)$ for this transformation to 
    work also hold for the modified-$\Omega_j(z)$'s here. Indeed, since the functions $\Omega_j(z)$ have the required asymptotics,
    and have the same kind of boundary values, the proof of this proposition can be replicated directly from \cite{DHL1}.
\end{proof}

The previous transformations have led us to an identity-normalized Riemann-Hilbert problem, which has rapidly oscillating jumps on the real axis. This situation is typical in Riemann-Hilbert analysis: one must ``open lenses'' around the branch cuts. We separate this
lens opening into two parts: in the first part $\boldsymbol{U} \mapsto \boldsymbol{T}$, we open lenses around the unbounded branch cuts 
$(-\infty,-\alpha] \cup [\alpha,\infty)$. We then open lenses around the central cut $[-\alpha,\alpha]$.

The opening of lenses here is based on the factorization of the jump matrix
    \begin{equation}
        \begin{pmatrix}
            e^{-ng_+(z)} & -1 \\
            0 & e^{-n g_-(z)}
        \end{pmatrix}
        =
        \begin{pmatrix}
            1 & 0 \\
            -e^{-n g_-(z)} & 1
        \end{pmatrix}
        \begin{pmatrix}
            0 & -1 \\
            1 & 0
        \end{pmatrix}
        \begin{pmatrix}
            1 & 0 \\
            -e^{-n g_+(z)} & 1
        \end{pmatrix},
    \end{equation}
where $g_+(z)$, $g_-(z)$ are the boundary values of one of the functions $\Omega_3(z) -\Omega_2(z)$, $\Omega_4(z) - \Omega_2(z)$ from above/below the contour.

Following \cite{DHL1}, the lensing propositions we proved imply that there exist lens-shaped regions (see Figure \ref{fig:Lenses-1}) around $(-\infty,-\alpha]$ (respectively, $[\alpha,\infty)$) such that the differences $\text{Re }[\Omega_3-\Omega_2]$ (respectively, $\text{Re }[\Omega_4-\Omega_2]$) are positive in this region. Define $\Gamma_{1,u},\Gamma_{1,l}$ as the boundaries of the lensing region around $[\alpha,\infty)$ in the upper and lower half planes, and similarly define $\Gamma_{2,u},\Gamma_{2,l}$ as the boundaries of the lensing region around $(-\infty,-\alpha]$.
The sectors enclosed by these contours are labelled as follows:
    \begin{itemize}
        \item $\Sigma_{1,u}$ is the region enclosed by $\Gamma_{1,u}$ and $[\alpha,\infty)$,
        \item $\Sigma_{1,l}$ is the region enclosed by $\Gamma_{1,l}$ and $[\alpha,\infty)$,
        \item $\Sigma_{2,u}$ is the region enclosed by $\Gamma_{2,u}$ and $(-\infty,-\alpha]$,
        \item $\Sigma_{2,l}$ is the region enclosed by $\Gamma_{2,l}$ and $(-\infty,-\alpha]$.
    \end{itemize}
These contours are depicted in Figure \ref{fig:Lenses-1}. Define matrices
    \begin{equation}
        V_1(z) = 
        \begin{psmallmatrix}
            1 & 0 & 0 & 0\\
            0 & 1 & 0 & 0\\
            0 & 0 & 1 & 0\\
            0 & -e^{-n\tau[\Omega_4(z) - \Omega_2(z)]} & 0 & 1
        \end{psmallmatrix}, \qquad\qquad
        V_2(z) = 
        \begin{psmallmatrix}
            1 & 0 & 0 & 0\\
            0 & 1 & 0 & 0\\
            0 & -e^{-n\tau[\Omega_3(z) - \Omega_2(z)]} & 1 & 0\\
            0 & 0 & 0 & 1
        \end{psmallmatrix}.
    \end{equation}
We define the transformation ${\bf U} \mapsto {\bf T}$ by setting
    \begin{equation}
        \boldsymbol{T}(z) = 
        \begin{cases}
            \boldsymbol{U}(z) V_1^{-1}(z), & z \in \Sigma_{1,u},\\
            \boldsymbol{U}(z) V_1(z), & z \in \Sigma_{1,l},\\
            \boldsymbol{U}(z) V_2^{-1}(z), & z \in \Sigma_{2,u},\\
            \boldsymbol{U}(z) V_2(z), & z \in \Sigma_{2,l},\\
            \boldsymbol{U}(z), & \textit{elsewhere}.
        \end{cases}
    \end{equation}
Clearly, $\boldsymbol{T}(z)$ is a piecewise analytic function off of the contours $\Gamma_1 $,$\Gamma_2$, $\Gamma$, $\Gamma_{1,u}$, $\Gamma_{1,l}$, $\Gamma_{2,u}$ , 
and $\Gamma_{2,l}$. In fact, $\boldsymbol{T}(z)$ is the unique solution to the following Riemann-Hilbert problem:
    \begin{prop}[\cite{DHL1}, Analog of Proposition 5.1]
        The function $\boldsymbol{T}(z)$ is the unique solution to the following RHP:
            \begin{equation}\label{T-jumps}
                \boldsymbol{T}_{+}(z) = \boldsymbol{T}_{-}(z)
                \begin{cases}
                        \mathbb{I} - E_{42}e^{-n\tau[\Omega_4(z) - \Omega_2(z)]}, & z\in \Gamma_{1,u} \cup \Gamma_{1,l},\\
                       \mathbb{I} - E_{32}e^{-n\tau[\Omega_3(z) - \Omega_2(z)]}, & z\in \Gamma_{2,u} \cup \Gamma_{2,l},\\
                       \begin{psmallmatrix}
                            1 & 0 & 0 & 0\\
                            0 & 0 & 0 & -1\\
                            0 & 0 & 1 & 0\\
                            0 & 1 & 0 & 0
                        \end{psmallmatrix}, & z\in [\alpha,\infty),\\
                        \begin{psmallmatrix}
                            1 & 0 & 0 & 0\\
                            0 & 0 & -1 & 0\\
                            0 & 1 & 0 & 0\\
                            0 & 0 & 0 & 1
                        \end{psmallmatrix}, & z\in (-\infty,-\alpha],\\
                        \mathbb{I} - E_{24}e^{-n\tau[\Omega_2(z) - \Omega_4(z)]}, 
                        & z \in \Gamma_1 \cap \{\text{Im } z > 0\},\\
                        \mathbb{I} - E_{23}e^{-n\tau[\Omega_2(z) - \Omega_3(z)]}, 
                        & z \in \Gamma_2 \cap \{\text{Im } z < 0\},\\
                        \mathbb{I} + E_{12}e^{-n\tau[\Omega_1(z) - \Omega_2(z)]},
                        & z\in \Gamma \setminus [-\alpha,\alpha],\\
                        \vspace{2mm}
                        \begin{psmallmatrix}
                            e^{-n\tau[\Omega_{1,-}(z)-\Omega_{2,-}(z)]} & 1 & 0 & 0 \\
                            0 & e^{n\tau[\Omega_{1,-}(z)-\Omega_{2,-}(z)]} & 0 & 0\\
                            0 & 0 & 1 & 0 \\
                            0 & 0 & 0 & 1
                        \end{psmallmatrix},
                        & z\in [-\alpha,\alpha].
                \end{cases}
            \end{equation}
        Furthermore, $\boldsymbol{T}(z)$ has asymptotics
\begin{equation}\label{T-asymptotics}
           \boldsymbol{T}(z) = 
           \bigg[\mathbb{I} + \OO (z^{-1/3})\bigg]
            \begin{psmallmatrix}
                1 & 0\\
                0 & B^{-1}(z) A^{-1}(z) 
            \end{psmallmatrix}, \qquad\qquad z\to \infty.
\end{equation}
    \end{prop}
\begin{proof}
    The proof of this proposition follows immediately from the definition of $\boldsymbol{T}(z)$.
\end{proof}

We now open the lens around the segment $[-\alpha,\alpha]$. This is based on the following factorization of the jump matrix:
\begin{equation*}
    \begin{psmallmatrix}
        e^{-n\tau[\Omega_{2,+}(z)-\Omega_{1,+}(z)]} & 1\\
        0 & e^{-n\tau[\Omega_{2,-}(z)-\Omega_{1,-}(z)]}
    \end{psmallmatrix}
    =\begin{psmallmatrix}
        1 & 0\\
        e^{-n\tau[\Omega_{2,-}(z)-\Omega_{1,-}(z)]} & 1
    \end{psmallmatrix}
    \begin{psmallmatrix}
        0 & 1\\
       -1 & 0\\
    \end{psmallmatrix}
    \begin{psmallmatrix}
        1 & 0\\
        e^{-n\tau[\Omega_{2,+}(z)-\Omega_{1,+}(z)]} & 1\\
    \end{psmallmatrix}.
\end{equation*}
By the lensing propositions of Section \S3, there exists a lens-shaped region around $[-\alpha,\alpha]$ such that the difference $\text{Re }[\Omega_2 -\Omega_1](z) >0$. Define contours $\Gamma_u,\Gamma_l$ as the boundaries of this lens-shaped region in the upper and lower half planes, respectively. Further, put $\Sigma_u, \Sigma_l$ to be the regions enclosed by $[-\alpha,\alpha]$ and $\Gamma_u,\Gamma_l$, respectively.
Define the invertible matrix $V_0(z)$ in the lensed region $\Sigma_u \cup \Sigma_l$ by
    \begin{equation}
        V_0(z) := \begin{psmallmatrix}
        1 & 0 & 0 & 0\\
        e^{-n\tau[\Omega_{2}(z)-\Omega_{1}(z)]} & 1 & 0 & 0\\
        0 & 0 & 1 & 0 \\
        0 & 0 & 0 & 1
    \end{psmallmatrix}.
    \end{equation}

We define the piecewise analytic function $\boldsymbol{S}(z)$ by
    \begin{equation}
        \boldsymbol{S}(z) := 
        \begin{cases}
        \vspace{3mm}
        \boldsymbol{T}(z)V_0^{-1}(z), & z\in\Sigma_{l},\\
        \vspace{3mm}
        \boldsymbol{T}(z)V_0(z), & z\in\Sigma_{u},\\
    \vspace{3mm}
        \boldsymbol{T}(z), & otherwise.
        \end{cases}
    \end{equation}
In this case, $\boldsymbol{S}(z)$ is the unique solution to the following RHP:
    \begin{prop}[\cite{DHL1}, Analog of Proposition 5.2] \label{lensing-proposition}
        The function $\boldsymbol{S}(z)$ is the unique solution to the following Riemann-Hilbert problem:
        $\boldsymbol{S}(z)$ is piecewise analytic off the contours $\Gamma_1 $, $\Gamma_2$, $\Gamma$, $\Gamma_{1,u}$, $\Gamma_{1,l}$, $\Gamma_{2,u}$, $\Gamma_{2,l}$, $\Gamma_l$, and $\Gamma_u$, satisfying the jump condition
            \begin{equation} \label{S-jumps}
                \boldsymbol{S}_{+}(z) = \boldsymbol{S}_{-}(z)
                \begin{cases}
                       \mathbb{I} - E_{42}e^{-n\tau[\Omega_4(z) - \Omega_2(z)]}, & z\in \Gamma_{1,u} \cup \Gamma_{1,l},\\
                       \mathbb{I} - E_{32}e^{-n\tau[\Omega_3(z) - \Omega_2(z)]}, & z\in \Gamma_{2,u} \cup \Gamma_{2,l},\\
                       \begin{psmallmatrix}
                            1 & 0 & 0 & 0\\
                            0 & 0 & 0 & -1\\
                            0 & 0 & 1 & 0\\
                            0 & 1 & 0 & 0
                        \end{psmallmatrix}, & z\in [\beta,\infty),\\
                        \begin{psmallmatrix}
                            1 & 0 & 0 & 0\\
                            0 & 0 & -1 & 0\\
                            0 & 1 & 0 & 0\\
                            0 & 0 & 0 & 1
                        \end{psmallmatrix}, & z\in (-\infty,-\beta],\\
                        \begin{psmallmatrix}
                            0 & 1 & 0 & 0\\
                           -1 & 0 & 0 & 0\\
                            0 & 0 & 1 & 0 \\
                            0 & 0 & 0 & 1
                        \end{psmallmatrix} & z\in [-\alpha,\alpha],\\
                        \mathbb{I} - E_{24}e^{-n\tau[\Omega_2(z) - \Omega_4(z)]}, 
                        & z \in \Gamma_1 \cap \{\text{Im } z > 0\},\\
                        \mathbb{I} - E_{23}e^{-n\tau[\Omega_2(z) - \Omega_3(z)]}, 
                        & z \in \Gamma_2 \cap \{\text{Im } z < 0\},\\
                        \mathbb{I} + E_{12}e^{-n\tau[\Omega_1(z) - \Omega_2(z)]},
                        & z\in \Gamma \setminus [-\alpha,\alpha],\\
                        \mathbb{I} + E_{21}e^{-n\tau[\Omega_{2}(z)-\Omega_{1}(z)]}, & z\in \Gamma_{u}\cup\Gamma_{l}.\\
                \end{cases}
            \end{equation}
Furthermore, $\boldsymbol{S}(z)$ has asymptotics
\begin{equation} \label{S-asymptotics}
            \boldsymbol{S}(z) =
            \bigg[\mathbb{I} + \OO\left(z^{-1}\right)\bigg]
            \begin{psmallmatrix}
                1 & 0\\
                0 & B^{-1}(z) A^{-1}(z)
            \end{psmallmatrix},
            \qquad\qquad z \to \infty .
        \end{equation}
    \end{prop}
\begin{proof}
    Again, the proof of this proposition follows immediately from the definition of $\boldsymbol{S}(z)$. The fact that the stronger condition
        \begin{equation}
            \boldsymbol{S}(z) =
            \bigg[\mathbb{I} + \OO\left(z^{-1}\right)\bigg]
            \begin{psmallmatrix}
                1 & 0\\
                0 & B^{-1}(z) A^{-1}(z)
            \end{psmallmatrix},
            \qquad\qquad z \to \infty,
        \end{equation}
    holds is due to the fact that the jumps of $A(z)B(z)$ match the jumps of $\boldsymbol{S}(z)$ at infinity. Thus, 
    $\OO(z^{-1/3})$ can be replaced with $\OO(z^{-1})$ in the asymptotics of $\boldsymbol{S}(z)$.
\end{proof}

\begin{figure}
		\begin{center}
		    \begin{overpic}[scale=.2]{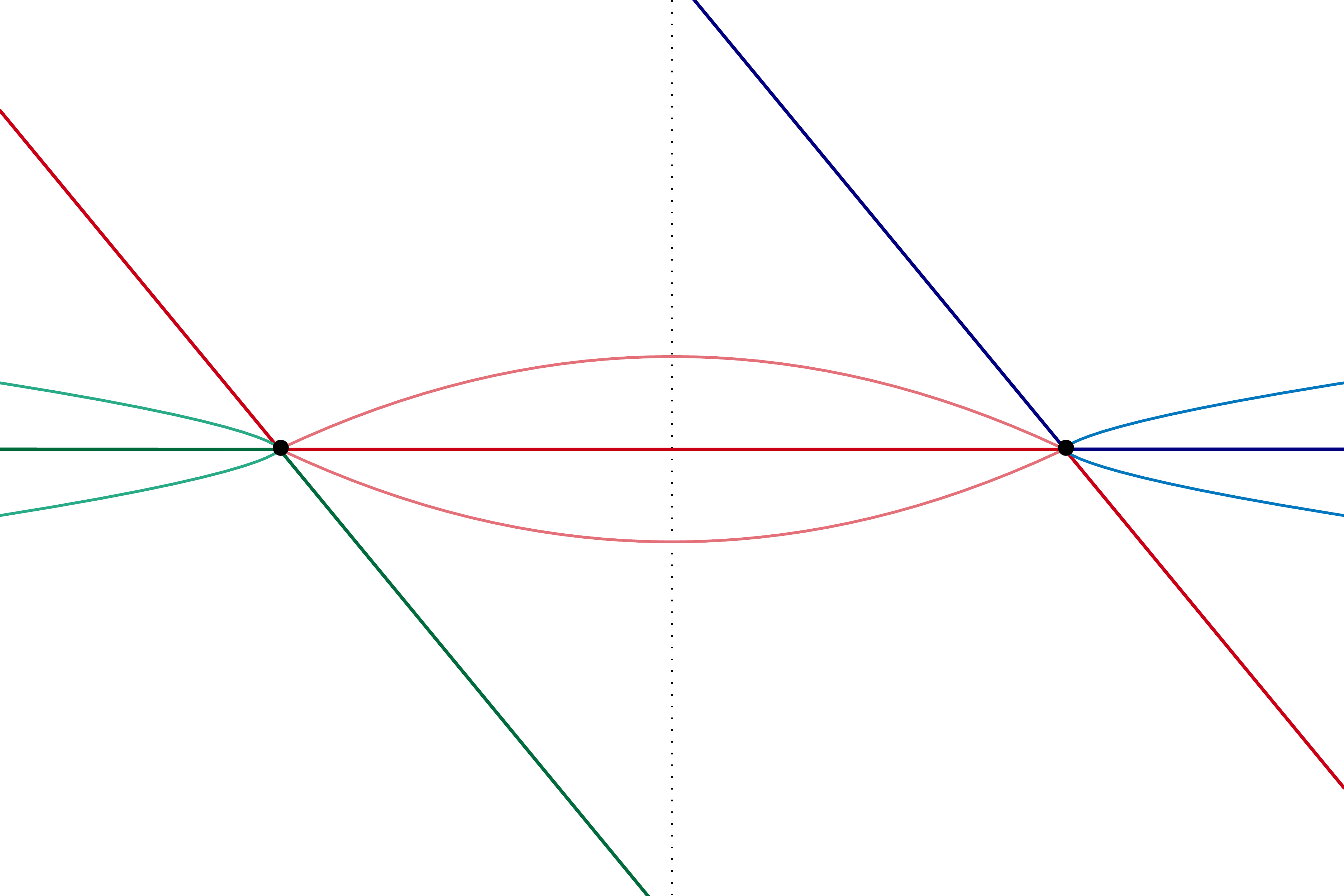}
		      \put (10,55) {$\Gamma$}
		      \put (52,43) {$\Gamma_u$}
		      \put (52,23) {$\Gamma_l$}
		      \put (35,5) {$\Gamma_2$}
		      \put (5,39) {$\Gamma_{2,u}$}
		      \put (5,23) {$\Gamma_{2,l}$}
		      \put (87,39) {$\Gamma_{1,u}$}
		      \put (87,26) {$\Gamma_{1,l}$}
		      \put (63,63) {$\Gamma_1$}
		    \end{overpic}
		\end{center}
    \caption{The opened lenses.}
    \label{fig:Lenses-1}
\end{figure}
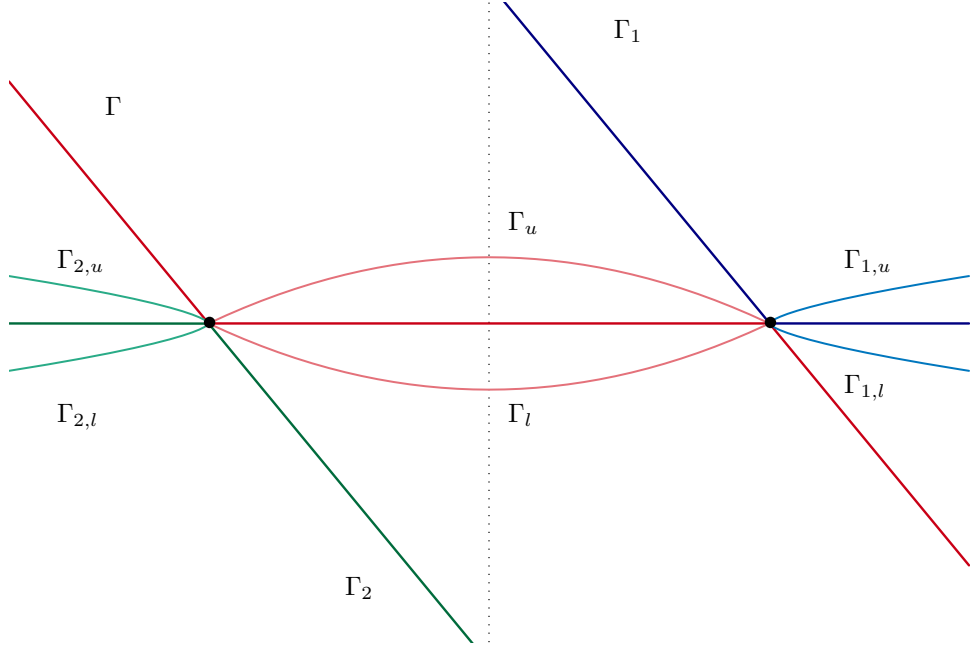

\subsection{Steepest descent: Global parametrix and symmetries}
Ignoring the exponentially small jumps of $\boldsymbol{S}(z)$, we arrive at the following model Riemann-Hilbert problem:

\begin{equation}\label{Model-RHP-a}
    \begin{cases}
        \text{$M$ is analytic in $\CC\setminus\RR$},\\
        M_+(z) = M_-(z)
        \begin{psmallmatrix}
            1 & 0 & 0 & 0\\
            0 & 0 & -1 & 0\\
            0 & 1 & 0 & 0\\
            0 & 0 & 0 & 1\\
        \end{psmallmatrix}, & z \in (-\infty,-\alpha],\\
        M_+(z) = M_-(z)
        \begin{psmallmatrix}
            0 & 1 & 0 & 0\\
            -1 & 0 & 0 & 0\\
            0 & 0 & 1 & 0\\
            0 & 0 & 0 & 1\\
        \end{psmallmatrix}, & z \in [-\alpha,\alpha],\\
        M_+(z) = M_-(z)
        \begin{psmallmatrix}
            1 & 0 & 0 & 0\\
            0 & 0 & 0 & -1\\
            0 & 0 & 1 & 0\\
            0 & 1 & 0 & 0\\
        \end{psmallmatrix}, & z \in [\alpha,\infty),\\
        M(z) = 
        \bigg[\mathbb{I} + \OO(z^{-1})\bigg]
            \begin{psmallmatrix}
                1 & 0\\
                0 & B^{-1}(z) A^{-1}(z)
            \end{psmallmatrix},
             & z \to \infty.
        \end{cases}
\end{equation}
Solutions to \eqref{Model-RHP-a} need not be unique. We can guarantee uniqueness by imposing
\begin{equation}\label{global-parametrix-constraint}
        M(z) = \OO( (z\mp \alpha)^{-1/3} ),  \qquad z\to \pm \alpha.
\end{equation}
One can directly construct a solution to this model problem, in a similar manner to \cite{DHL1}. We only present the result of 
this construction here, as it can be inferred from the result of \cite{DHL1} by taking a limit as $b\to a = 1$ there.

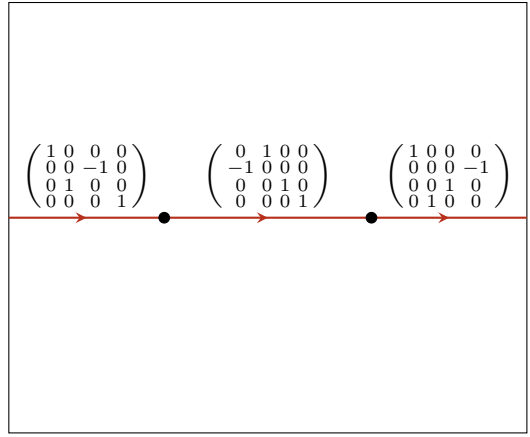
\begin{figure}
	\begin{center}
			\begin{tikzpicture}[scale=1]
			\begin{axis}[
				tick style={draw=none},
				xticklabels=\empty,
				yticklabels=\empty,
    				xmin = -2.5, xmax = 2.5,
    				ymin = -1.5, ymax = 1.5]
				 \addplot [BrickRed,thick,domain=-2.5:-1,samples=50,
					postaction={decorate},
				 	decoration={markings, 
		 				mark=at position 0.5  with {\arrow{stealth}}}				 
				 ]({x},{0}) node[black,above,pos=.5] 
				 {$ \begin{psmallmatrix}
            			1 & 0 & 0 & 0\\
            			0 & 0 & -1 & 0\\
            			0 & 1 & 0 & 0\\
            			0 & 0 & 0 & 1\\
        				\end{psmallmatrix}$};
				 \addplot [BrickRed,thick,domain=-1:1,samples=50,
					postaction={decorate},
				 	decoration={markings, 
		 				mark=at position 0.5  with {\arrow{stealth}}}				 
				 ]({x},{0}) node[black,above,pos=.5] 
				 {$\begin{psmallmatrix}
            			0 & 1 & 0 & 0\\
            			-1 & 0 & 0 & 0\\
            			0 & 0 & 1 & 0\\
            			0 & 0 & 0 & 1\\
        				\end{psmallmatrix}$};
				 \addplot [BrickRed,thick,domain=1:2.5,samples=50,
					postaction={decorate},
				 	decoration={markings, 
		 				mark=at position 0.5  with {\arrow{stealth}}}				 
				 ]({x},{0}) node[black,above,pos=.5] 
				 {$ \begin{psmallmatrix}
            			1 & 0 & 0 & 0\\
            			0 & 0 & 0 & -1\\
            			0 & 0 & 1 & 0\\
            			0 & 1 & 0 & 0\\
        				\end{psmallmatrix}$};
				 \addplot[mark=*] coordinates { (1,0) };
				 \addplot[mark=*] coordinates { (-1,0) };
			\end{axis}
			\end{tikzpicture}
		\end{center}
    \caption{The jumps of the solution to the model problem $M(z)$.}
    \label{fig:Global-Parametrix}
\end{figure}

\begin{prop}[\cite{DHL1}, Analog of Proposition 6.1]
    The solution to the model RHP defined by \eqref{Model-RHP-a} is
            \begin{equation}
                M_{jk}(z) = \begin{cases}
                    \psi_j(u_k(z)), & \text{Im } z >0,\\
                    \psi_j(u_\ell(z))S_{\ell k}, & \text{Im } z <0,\\
                \end{cases}
            \end{equation}
        (summation over $\ell$ is implied) where $u_k(z)$ is the restriction of the uniformizing coordinate to the $k^{th}$ sheet, and the $\psi_j(u)$ the following rational functions in the uniformizing plane:
        \begin{equation}
        \begin{cases}
            \psi_1(u) = \frac{u^2}{u^2-1}, \qquad & \qquad \psi_2(u) = 3^{2/3}A^{1/3}\frac{2-3u^2}{18u(u^2-1)} \\
            \psi_3(u) = \frac{1}{3(u^2-1)}, \qquad &\qquad \psi_4(u) = \frac{1}{3^{2/3}A^{1/3} }\frac{u}{u^2-1},
        \end{cases}
        \end{equation}
    and $S = \text{diag }(1, -1, 1, 1) = [\sigma_3 \oplus \mathbb{I}_2]$.
    \end{prop}

We also will make use of the following factorization of $M(z)$:
\begin{lemma}
    In a neighborhood of $z=\alpha$, the global parametrix can be factorized as
        \begin{equation}
            M(z) = M_{reg}(z) M_{sing}(z),
        \end{equation}
    where 
        \begin{equation}
            M_{reg}(z) = \mathbb{I} + \sum_{k=1}^{\infty} M_{reg,k}(z-\alpha)^k
        \end{equation}
    is convergent in a neighborhood of $z = \alpha$, and uniformly bounded in $n$. The matrix $M_{sing}(z)$ is given by
        \begin{equation}
            M_{sing}(z) = \hat{C}_0 (z-\alpha)^{\Delta} \mathcal{U},
        \end{equation}
    where
        \begin{equation*}
            \mathcal{U} :=
            \begin{psmallmatrix}
                1 & \omega^2 & 0 & \omega\\
                1 & 1 & 0 & 1\\
                0 & 0 & 1 & 0\\
                1 & \omega & 0 & 1
            \end{psmallmatrix},
        \end{equation*}
    $\Delta = \text{diag }\left(\frac{1}{3},0,0,-\frac{1}{3}\right)$, and $\hat{C}_0$ is a bounded (in $n$) invertible matrix.
\end{lemma}
\begin{proof}
    The proof is a direct calculation. $\hat{C}_0^{-1}M_{sing}(z)$ is an explicit matrix, and so one can show that the matrix
        \begin{equation*}
            M_{reg}(z) := M(z)M_{sing}^{-1}(z)\hat{C}_0
        \end{equation*}
    is holomorphic in a neighborhood of $z=\alpha$. We choose $\hat{C}_0$ so that $M_{reg}(0)=\mathbb{I}$; this matrix is bounded in $n$.
    This completes the proof.
\end{proof}

We also observe that $M(z)$ carries the following symmetry:
\begin{prop}[\cite{DHL1}, Analog of Proposition 6.1]\label{global-parametrix-prop} 
    $M(z)$ has the symmetry
        \begin{equation} \label{global-parametrix-symmetry}
            M(-z) = \left[\sigma_3 \oplus \sigma_3 \right] M(z) [\sigma_3 \oplus \sigma_1].
        \end{equation}
\end{prop}

This will be useful in the construction of our parametrices. Our new parametrix $\hat{M}(z)$ will be of the form
    \begin{equation}
        \hat{M}(z) = 
        \begin{cases}
            M(z), & z\in \CC \setminus (D_{\pm}),\\
            \mathcal{P}^{(+)}(z), & z\in D_+,\\
            \mathcal{P}^{(-)}(z), & z\in D_-.
        \end{cases}
    \end{equation}
The functions $\mathcal{P}^{(\pm)}(z)$ will be chosen to so that the following conditions are met:
    \begin{enumerate}
        \item $\mathcal{P}^{(\pm)}(z)$ matches the jumps of $\boldsymbol{S}(z)$ \textit{exactly} in the discs $D_{\pm}$,
        \item $\mathcal{P}^{(\pm)}(z) = [\mathbb{I} + \OO(n^{-\delta})] M(z)$, for $z\in D_{\pm}$, $n\to \infty$ for some $\delta >0$.
    \end{enumerate}
Before calculating the parametrices, it is useful to notice that we only have to calculate one of the parametrices; the other can be calculated by symmetry.
This statement is summarized in the following proposition:
\begin{prop}[\cite{DHL1}, Analog of Lemmas 6.7 and 6.8]\label{Parametrix-symmetry-prop}
    Suppose $\mathcal{P}^{(+)}(z)$ satisfies conditions 1. and 2. above. Then, $\mathcal{P}^{(-)}(z)$ can be constructed as
        \begin{equation}
            \mathcal{P}^{(-)}(z) = [\sigma_3\oplus \sigma_3]\mathcal{P}^{(+)}(-z)[\sigma_3 \oplus \sigma_1].
        \end{equation}
\end{prop}

\subsection{Steepest descent: Local parametrices}

Thus, it is sufficient to calculate the parametrix at $z=\alpha$ only, and use the above symmetry relation as the definition of $P^{(-)}(z)$. We therefore 
relabel 
    \begin{equation}
        P^{(+)}(z) =: P(z).
    \end{equation}

\begin{figure}
    \centering
    \begin{tikzpicture}[scale=1.4]
			\begin{axis}[tick style={draw=none},
				xticklabels=\empty,
				yticklabels=\empty,
    				xmin = -4, xmax = 4,
    				ymin = -3, ymax = 3,
    				axis equal image]
				 \addplot [BrickRed,thick,domain=0:4,samples=50,
					postaction={decorate},
				 	decoration={markings, 
		 				mark=at position 0.5  with {\arrow{stealth}}}							 
				 ]({x},{0}) node[black,above left,pos=1] {\tiny$\begin{psmallmatrix}
				     1& 0& 0& 0\\
				     0& 0& 0& -1\\
				     0& 0& 1& 0\\
				     0& 1& 0& 0\\
				 \end{psmallmatrix}$};
				 \addplot [BrickRed,thick,domain=0:4,samples=50,
					postaction={decorate},
				 	decoration={markings, 
		 				mark=at position 0.5  with {\arrowreversed{stealth}}}							 
				 ]({-x},{0}) node[black,above right,pos=1] {\tiny$\begin{psmallmatrix}
				     0& 1& 0& 0\\
				     -1& 0& 0& 0\\
				     0& 0& 1& 0\\
				     0& 0& 0& 1\\
				 \end{psmallmatrix}$};
				 \addplot [BrickRed,thick,domain=0:1.5,samples=50,
					postaction={decorate},
				 	decoration={markings, 
		 				mark=at position 0.5  with {\arrow{stealth}}}			 
				 ]({x^2},{-2.3*x}) node[black,below left,pos=.65] {\scriptsize $\mathbb{I} + E_{12}e^{-n\tau\delta\Omega_{12}(z)}$};
				 \addplot [BrickRed,thick,domain=0:1.5,samples=50,
					postaction={decorate},
				 	decoration={markings, 
		 				mark=at position 0.5  with {\arrowreversed{stealth}}}				 
				 ]({-x^2},{2.3*x}) node[black,above right,pos=.65] {\scriptsize $\mathbb{I} - E_{24}e^{-n\tau\delta\Omega_{24}(z)}$};
				 \addplot [BrickRed,thick,domain=0:2,samples=50,
					postaction={decorate},
				 	decoration={markings, 
		 				mark=at position 0.5  with {\arrowreversed{stealth}}}				 
				 ]({-x^2},{.6*x}) node[black,above right,pos=.9] {\scriptsize $\mathbb{I} +E_{21} e^{-n\tau\delta\Omega_{21}(z)}$};
				 \addplot [BrickRed,thick,domain=0:2,samples=50,
					postaction={decorate},
				 	decoration={markings, 
		 				mark=at position 0.5  with {\arrowreversed{stealth}}}				 
				 ]({-x^2},{-.6*x}) node[black,below right,pos=.9] {\scriptsize $\mathbb{I} +E_{21} e^{-n\tau\delta\Omega_{21}(z)}$};
				 \addplot [BrickRed,thick,domain=0:2,samples=50,
					postaction={decorate},
				 	decoration={markings, 
		 				mark=at position 0.5  with {\arrow{stealth}}}				 
				 ]({x^2},{.6*x}) node[black,above left,pos=.99] {\scriptsize $\mathbb{I} -E_{42} e^{-n\tau\delta\Omega_{42}(z)}$};
				 \addplot [BrickRed,thick,domain=0:2,samples=50,
					postaction={decorate},
				 	decoration={markings, 
		 				mark=at position 0.5  with {\arrow{stealth}}}				 
				 ]({x^2},{-.6*x}) node[black,below left,pos=.99] {\scriptsize $\mathbb{I} -E_{42} e^{-n\tau\delta\Omega_{42}(z)}$};
				 \addplot[mark=*] coordinates { (0,0) };
			\end{axis}
			\end{tikzpicture}
    \caption{The jumps of the Riemann-Hilbert problem for $\boldsymbol{S}(z)$ in a small disc at $z=\alpha$. We must match these jumps exactly. Note that only the rows (columns) $1,2,$ and $4$ participate nontrivially.}
    \label{fig:MulticriticalParametrix}
\end{figure}
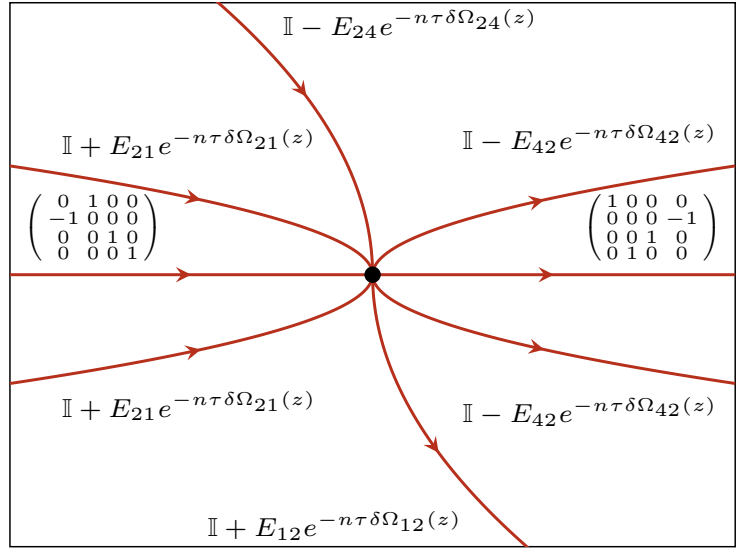

We are now ready to construct local parametrices around $z=\pm \alpha$. We utilize a pair of nested discs. For $\delta,\epsilon>0$ sufficiently small, fixed (and independent of $n$), we set
    \begin{align}
        \hat{\DD}_{\pm} &:=\{z\mid |z\mp\alpha|< \delta\},\\
        \DD_{\pm} &:= \{z\mid |z\mp \alpha|< n^{-2/7-\epsilon}\}.
    \end{align}

Define the local coordinate
    \begin{equation}
        \xi(z) = c_0(z-\alpha),
    \end{equation}
where $c_0 := \frac{1}{2}(40)^{1/14}\sqrt{3}$. The local parametrix will have the following form:
    \begin{equation}
        \mathcal{P}(z) := 
        \begin{cases}
            E_1(z)\hat{\sigma}_{34}\left[\hat{\Psi}(n^{3/7}\xi(z);n^{2/7}\boldsymbol{\eta}(z),n^{5/7}\boldsymbol{\mu}(z),n^{6/7}\boldsymbol{\nu}(z)) \oplus 1\right]\hat{\sigma}_{34}e^{n\boldsymbol{G}(z)}, & z\in \DD_+,\\
            E_2(z), & z\in \hat{\DD}_+\setminus\DD_+,
        \end{cases}
    \end{equation}
where:
    \begin{itemize}
        \item $\hat{\Psi}(\xi;\eta,\mu,\nu)$ is the model Riemann-Hilbert problem introduced in Appendix \ref{STRING-APPENDIX},
        \item $E_1(z)$ is analytic in $\DD_+$, to be determined,
        \item $\boldsymbol{G}(z)$ is the $g$-function, defined in \eqref{G-matrix},
        \item $\boldsymbol{\eta}(z;\vec{t})$, $\boldsymbol{\mu}(z;\vec{t})$, $\boldsymbol{\nu}(z;\vec{t})$ are appropriately chosen analytic functions, which satisfy, as $n\to \infty$,
            \begin{equation}\label{eta-mu-nu-limits}
                n^{2/7}\boldsymbol{\eta}(z;\vec{t})  = \eta + o(1),\qquad n^{5/7}\boldsymbol{\mu}(z;\vec{t}) = \mu + o(1),\qquad \boldsymbol{\nu}(z;\vec{t})  = \nu + o(1),
            \end{equation}
        uniformly\footnote{Observe that, since we chose $(\eta,\mu,\nu)$ to lie away from the singularity set of the solution to the model Riemann-Hilbert problem, $\hat{\Psi}(n^{3/7}\xi(z);n^{2/7}\boldsymbol{\eta}(z),n^{5/7}\boldsymbol{\mu}(z),n^{6/7}\boldsymbol{\nu}(z))$ is well-defined.} for $z\in \DD_+$. 
        \item $E_2(z)$ is analytic in $\hat{\DD}_+\setminus \DD_+$, to be determined.
    \end{itemize}
The method of using nested discs was already present in \cite{Molag0,Molag}. However, our construction differs from theirs in our construction of $E_2(z)$, where their methods break down in our situation.
The next proposition identifies suitable choices of $\boldsymbol{\eta}(z;\vec{t})$, $\boldsymbol{\mu}(z;\vec{t})$, $\boldsymbol{\nu}(z;\vec{t})$.
\begin{prop}\label{Omega-breakup}
         There exist functions $\boldsymbol{\eta}(z;\vec{t})$, $\boldsymbol{\mu}(\vec{t})$, $\boldsymbol{\nu}(z;\vec{t})$ and $K_0(z)$, analytic in a neighborhood of $z=\alpha$ such that
         \begin{align}
                \Omega_{1}(z) &= -\frac{3}{7}[\xi(z)]^{7/3} - \boldsymbol{\eta}(z;\vec{t})[\xi(z)]^{5/3} -\boldsymbol{\mu}(\vec{t})[\xi(z)]^{2/3} - \boldsymbol{\nu}(z;\vec{t})[\xi(z)]^{1/3} + K_0(z),\label{Omega-1-local}\\
                \Omega_{2}(z) &= 
                    \begin{cases}
                        -\frac{3}{7}\omega^2[\xi(z)]^{7/3} - \omega\boldsymbol{\eta}(z;\vec{t})[\xi(z)]^{5/3} -\omega\boldsymbol{\mu}(\vec{t})[\xi(z)]^{2/3} - \omega^2\boldsymbol{\nu}(z;\vec{t})[\xi(z)]^{1/3} + K_0(z), & \text{Im }z >0,\\
                        -\frac{3}{7}\omega[\xi(z)]^{7/3} - \omega^2\boldsymbol{\eta}(z;\vec{t})[\xi(z)]^{5/3} -\omega^2\boldsymbol{\mu}(\vec{t})[\xi(z)]^{2/3} - \omega\boldsymbol{\nu}(z;\vec{t})[\xi(z)]^{1/3} + K_0(z), & \text{Im }z <0,
                    \end{cases}\\
                \Omega_{4}(z) &= 
                    \begin{cases}
                        -\frac{3}{7}\omega[\xi(z)]^{7/3} - \omega^2\boldsymbol{\eta}(z;\vec{t})[\xi(z)]^{5/3} -\omega^2\boldsymbol{\mu}(\vec{t})[\xi(z)]^{2/3} - \omega\boldsymbol{\nu}(z;\vec{t})[\xi(z)]^{1/3} +K_0(z), & \text{Im }z >0,\\
                        -\frac{3}{7}\omega^2[\xi(z)]^{7/3} - \omega\boldsymbol{\eta}(z;\vec{t})[\xi(z)]^{5/3} -\omega\boldsymbol{\mu}(\vec{t})[\xi(z)]^{2/3} - \omega^2\boldsymbol{\nu}(z;\vec{t})[\xi(z)]^{1/3} +K_0(z) & \text{Im }z <0.
                    \end{cases}
            \end{align}
            Furthermore, provided $|z-\alpha| = \OO(n^{-2/7-\epsilon})$, then the limits \eqref{eta-mu-nu-limits} hold.
    \end{prop}
    \begin{proof}
       We prove that the equality \eqref{Omega-1-local} for $\Omega_1(z)$ holds; the proof for the other branches is identical. Write
        \begin{equation*}
            \Omega_{1}(z) = \sum_{k=0}^{\infty} C_k(\vec{t})(z-\alpha)^{k/3}.
        \end{equation*}
        We then temporarily define the auxiliary functions $K_0(z)$, $K_1(z)$, $K_2(z)$ by
            \begin{equation*}
            K_0(z) = \sum_{k=0}^{\infty} C_{3k} (z-\alpha)^k,\qquad
            K_1(z) = \frac{9}{28}(30)^{1/3}(z-\alpha)^2 + \sum_{k=0}^{\infty}C_{3k+1} (z-\alpha)^k,\qquad
            K_2(z) =\sum_{k=0}^{\infty}C_{3k+2} (z-\alpha)^k.
            \end{equation*}
        Clearly, the $K_j(z)$ are analytic in a neighborhood of $z=\alpha$, and by construction,
            \begin{align*}
                \Omega_{1}(z) &= -\frac{3}{7}[\xi(z)]^{7/3} + K_2(z)(z-\alpha)^{2/3} + K_1(z)(z-\alpha)^{1/3} + K_0(z)\\
                &= -\frac{3}{7}[\xi(z)]^{7/3} + c_0^{-2/3}K_2(z)[\xi(z)]^{2/3} + c_0^{-1/3} K_1(z)[\xi(z)]^{1/3} + K_0(z).
            \end{align*}
        We now define analytic functions
            \begin{align}
                \boldsymbol{\eta}(z;\vec{t}) &:= -c_0^{-5/3} \frac{K_2(z)-K_2(\alpha)}{z-\alpha},\\
                \boldsymbol{\mu}(\vec{t}) &:= -c_0^{-2/3} K_2(\alpha),\\
                \boldsymbol{\nu}(z;\vec{t})&:= -c_0^{-1/3} K_1(z).
            \end{align}
        By construction, it is immediate that equality \eqref{Omega-1-local} holds. It remains to see that, for $|z-\alpha| = \OO(n^{-2/7-\epsilon})$,
        that these functions have the desired limiting properties. Expanding $\boldsymbol{\eta}(z;\vec{t})$ about $z=\alpha$, we obtain that
            \begin{align*}
                \boldsymbol{\eta}(z;\vec{t}) = -c_0^{-5/3} C_5(\vec{t}) + \OO((z-\alpha)).
            \end{align*}
        Now, expanding $C_5(\vec{t})$ for large $n$, and using the fact that $|z-\alpha| = \OO(n^{-2/7-\epsilon})$,
            \begin{align*}
                n^{2/7}\boldsymbol{\eta}(z;\vec{t}) = \left[\eta + \OO(n^{-2/7})\right] + \OO(n^{-\epsilon}) = \eta + \OO(n^{-\epsilon}).
            \end{align*}
        Since $\epsilon>0$, we see that $n^{2/7}\boldsymbol{\eta}(z;\vec{t}) \to \eta$ as $n\to \infty$. 
        Next, expanding $\boldsymbol{\mu}(\vec{t})$ for large $n$, 
            \begin{equation*}
                n^{5/7}\boldsymbol{\mu}(\vec{t}) = \mu + \OO(n^{-2/7}),
            \end{equation*}
        so $n^{5/7}\boldsymbol{\mu}(\vec{t}) \to \mu$ as $n\to \infty$. Finally, expanding $\boldsymbol{\nu}(z;\vec{t})$ about $z=\alpha$,
            \begin{equation}
                \boldsymbol{\nu}(z;\vec{t}) =-c_0^{-1/3}\left[ C_1(\vec{t}) + \left[C_7(\vec{t}) + \frac{9}{28}(30)^{1/6}\right](z-\alpha)^2 + \OO((z-\alpha)^3)\right].
            \end{equation}
        where we have used the definition of $\boldsymbol{\nu}(z;\vec{t})$,  and the fact that $C_4\equiv 0$ by definition of the modified spectral curve. 
        Expanding the above for large $n$, we see that
            \begin{align*}
                -n^{6/7}c_0^{-1/3}C_1(\vec{t}) &= \nu + \OO(n^{-1/7}),\\
                -n^{6/7}c_0^{-1/3}\left[C_7(\vec{t}) + \frac{9}{28}(30)^{1/6}\right](z-\alpha)^2 &= \OO(n^{-2\epsilon}),\\
                n^{6/7}\OO((z-\alpha)^3) &= \OO(n^{-3\epsilon}),
            \end{align*}
        so provided $\epsilon>0$ is sufficiently small, $n^{6/7}\boldsymbol{\nu}(z;\vec{t}) = \nu + \OO(n^{-3\epsilon})$, and we have the convergence
        $n^{6/7}\boldsymbol{\nu}(\vec{t}) \to \nu$.   
    \end{proof}

We will begin our analysis with a `naive' choice of local parametrix, in which make modifications to the global parametrix only in the shrinking discs, i.e. we choose $E_2(z) = M(z)$. This will turn out to be insufficient, as the jumps we obtain on $\partial \DD_{\pm}$ are in fact growing with $n$. We will then slightly modify this construction in order to rectify this problem.

\begin{prop}
    Define
        \begin{align}
            P_{pre}(z) &:= E_{pre}(z)\left[\hat{\Psi}(n^{3/7}\xi(z);n^{2/7}\boldsymbol{\eta}(z;\vec{t}),n^{5/7}\boldsymbol{\mu}(\vec{t}),n^{6/7}\boldsymbol{\nu}(z;\vec{t})) \oplus 1\right]\hat{\sigma}_{34} e^{-n\mathcal{Q}(z) + n\boldsymbol{G}(z)},\\
            E_{pre}(z) &:=M(z) \hat{\sigma}_{34}[\hat{g}(\xi(z))\oplus 1]^{-1}\hat{\sigma}_{34},
        \end{align}
    (observe that $E_{pre}(z)$ is analytic in a neighborhood of $z=\alpha$), and set
        \begin{equation}
            \mathcal{P}_{pre}(z) := 
            \begin{cases}
                M(z), & z\in \hat{\DD}_+\setminus \DD_+,\\
                P_{pre}(z), & z\in \DD_+,
            \end{cases}
        \end{equation}
        \begin{equation}
        \hat{M}_{pre}(z) := 
        \begin{cases}
            M(z), & z\in \CC\setminus \hat{\DD}_{\pm},\\
            \mathcal{P}_{pre}(z), & z\in \hat{\DD}_+,\\
            [\sigma_3 \oplus \sigma_3]\mathcal{P}_{pre}(-z)[\sigma_3 \oplus \sigma_1], & z\in \DD_-.
        \end{cases}
    \end{equation}
Then, the jumps of the matrix function $\boldsymbol{S}(z)\hat{M}_{pre}^{-1}(z)$ across $\partial \DD_+$ are given by
    \begin{equation}
        \left[\boldsymbol{S}(z)\hat{M}_{pre}^{-1}(z)\right]_+ = \left[\boldsymbol{S}(z)\hat{M}_{pre}^{-1}(z)\right]_{-} \mathcal{D}(z),
    \end{equation}
where
    \begin{equation}
        \mathcal{D}(z) = M(z)\hat{\sigma}_{34}\left[\left(\sum_{k=0}^{\infty}\frac{\hat{\Psi}_k(n^{2/7}\boldsymbol{\eta}(z;\vec{t}),n^{5/7}\boldsymbol{\mu}(\vec{t}),n^{6/7}\boldsymbol{\nu}(z;\vec{t}))}{c_0^{k/3}(z-\alpha)^{k/3}} n^{-k/7}\right)\oplus 1\right]\hat{\sigma}_{34}M^{-1}(z).
    \end{equation}
\end{prop}
\begin{proof}
    This is a standard calculation. Since we will not use this choice of parametrix anyways, we omit the proof here.
\end{proof}

The matrix-valued function
    \begin{equation}\label{D-matrix}
        \mathcal{D}(z) := M(z)\hat{\sigma}_{34}\left[\left(\sum_{k=0}^{\infty}\frac{\hat{\Psi}_k(n^{2/7}\boldsymbol{\eta}(z;\vec{t}),n^{5/7}\boldsymbol{\mu}(\vec{t}),n^{6/7}\boldsymbol{\nu}(z;\vec{t}))}{c_0^{k/3}(z-\alpha)^{k/3}} n^{-k/7}\right)\oplus 1\right]\hat{\sigma}_{34}M^{-1}(z),
    \end{equation}
where $M(z)$ is the global parametrix, $\boldsymbol{\eta},\boldsymbol{\mu},\boldsymbol{\nu}$ are as in Proposition \ref{Omega-breakup}, and the
$\Psi_k$ are the coefficients of the expansion of the model RHP (cf. Equation \eqref{Psi-hat-asymptotics}), will be important in our analysis, but is \textit{not} close to the identity matrix, as one would hope. Thus, the methods described in \cite{DG,DKZ} do not apply directly, and we will have to resort to some alternative. 
The following structural lemma is the key to understanding how to circumvent this problem:
\begin{lemma}
        \begin{equation}\label{local-symmetry}
            \mathcal{D}(z) = M_{reg}(z)\hat{\sigma}_{34}\left[ \left(\sum_{k=0}^{\infty}\frac{\Lambda^k}{n^{k/7}} D_k\right)\oplus 1\right]\hat{\sigma}_{34}M_{reg}^{-1}(z),
        \end{equation}
    where 
        \begin{equation}\label{Lambda-matrix}
            \Lambda := 
            \begin{pmatrix}
                0 & 1 & 0\\
                0 & 0 & 1\\
                (z-\alpha)^{-1} & 0 & 0
            \end{pmatrix},
        \end{equation}
    and $D_k$ are $3\times 3$ diagonal matrices, which are uniformly bounded in $\hat{\DD}_+$ as $n\to \infty$.
\end{lemma}
    \begin{proof}
        The proof of this lemma is immediate from the definition of $M_{sing}(z)$, and the symmetry $\omega^{-k}\mathcal{S}^T\Psi_k\mathcal{S} = \Psi_k$ of the model problem (see Appendix \ref{Appendix-A}). The observation that the matrices $D_k$ are uniformly bounded follows from the fact that the entries of $D_k$ are comprised of the entries of $\Psi_k$, evaluated at $(\eta,\mu,\nu) = (\boldsymbol{\eta},\boldsymbol{\mu},\boldsymbol{\nu})$. Since we have chosen the point $(\eta,\mu,\nu)$ to be away from singularities of the model problem, this implies that the entries of $\Psi_k$ are uniformly bounded in its parameters.
    \end{proof}
The above matrix would be the jump across $\partial \DD_+$ if we had taken the parameterix in $\DD_+$ to be $P(z)$, and $M(z)$ outside of this disc. 
Here is where we see the issue with this choice: since the disc $\DD_+$ is $n$-dependent (and shrinks with $n$), the first few terms in the above expansion actually grow with $n$. The three `problem terms' are
    \begin{equation}\label{blowup-terms}
        \frac{\Lambda}{n^{1/7}} = \OO(n^{1/7+\epsilon}), \qquad \frac{\Lambda^2}{n^{2/7}} = \OO(n^{\epsilon}),\qquad \frac{\Lambda^4}{n^{4/7}} = \OO(n^{2\epsilon}).
    \end{equation}
The structure of the matrix inside the braces of the right hand side of \eqref{local-symmetry} can be exploited to build \textit{modified} local parametrices, which circumvent the above issue. We list the relevant properties of such matrices in the following elementary lemma:
\begin{lemma}
    Let $\Lambda$ be as in \eqref{Lambda-Matrix}.  Then:
    \begin{enumerate}
        \item If $D:=\text{diag }(d_1,d_2,d_3)$, then $D\Lambda^{j}=\Lambda^j\sigma^{j}[D]$, where $\sigma:=(1\,\, 3\,\, 2)$ is a permutation, and
        $\sigma[D]:=\text{diag }(d_{\sigma(1)},d_{\sigma(2)},d_{\sigma(3)})$,
        \item Put
        \begin{equation*}
            M_a:=\sum_{k=0}^{\infty}\frac{\Lambda^k}{n^{k/7}}D_k^a,\qquad\qquad M_b:=\sum_{k=0}^{\infty}\frac{\Lambda^k}{n^{k/7}}D_k^b,
        \end{equation*}
        where $D_k^a,D_k^b$ are diagonal matrices $k\in \ZZ_+$. Then
            \begin{equation*}
                M_aM_b = \sum_{k=0}^{\infty}\frac{\Lambda^k}{n^{k/7}}\left(\sum_{\ell=0}^k \sigma^{\ell}\left[D_{k-\ell}^a\right]D_{\ell}^b\right).
            \end{equation*}
    \end{enumerate}
\end{lemma}

\begin{prop}\label{U-def-prop}
    There exists an analytic and invertible series $U(z)$ in $\DD_+$ of the form
        \begin{equation*}
            U(z) = \sum_{k=0}^{\infty}\frac{\Lambda^k}{n^{k/7}}U_k,
        \end{equation*}
    where $U_k$ are diagonal, such that 
        \begin{equation}
            \mathcal{D}(z) U(z) = M_{reg}(z)\hat{\sigma}_{34}\left[ \left(\sum_{k=0}^{\infty}\frac{\Lambda^k}{n^{k/7}} \hat{D}_k\right)\oplus 1\right]\hat{\sigma}_{34}M_{reg}^{-1}(z),
        \end{equation}
    where
        \begin{equation}
            \Lambda^k\hat{D}_k = \OO(1),\qquad \qquad z\to \alpha.
        \end{equation}
\end{prop}
\begin{proof}
    Explicitly, we have that
        \begin{equation*}
            \hat{D}_k = \sum_{j=0}^{k} \sigma^j\left(D_j\right) U_{k-j};
        \end{equation*}
    since the above expression is linear in $U_k$, one can inductively choose the entries of $U_k$ to satisfy the hypotheses of the lemma. 
    Explicitly, we take
        \begin{equation}
            U_k =-\Lambda^{-1} \left(\oint_{\partial D_+} \sum_{j=1}^{k}\sigma^j\left(D_j(\xi)\right) U_{k-j}(\xi) \frac{d\xi}{2\pi i(\xi-z)}\right), \qquad k\geq 1,
        \end{equation}
    where $z$ is a point in the annulus $\hat{\DD}_+\setminus \DD_+$. By explicit computation, one finds that the first few matrices $U_k$ are given as in the statement of the 
    proposition. Invertibility of $U$ follows from similar arguments; since $U(\alpha) = \mathbb{I}$, one can explicitly construct its inverse matrix, which is also bounded in $n$ and thus convergent, since the terms growing in $n$ are nilpotent (cf. Equations \eqref{Lambda-matrix}, \eqref{blowup-terms}).
\end{proof}

\begin{figure}
    \centering
    \begin{overpic}[scale=.15]{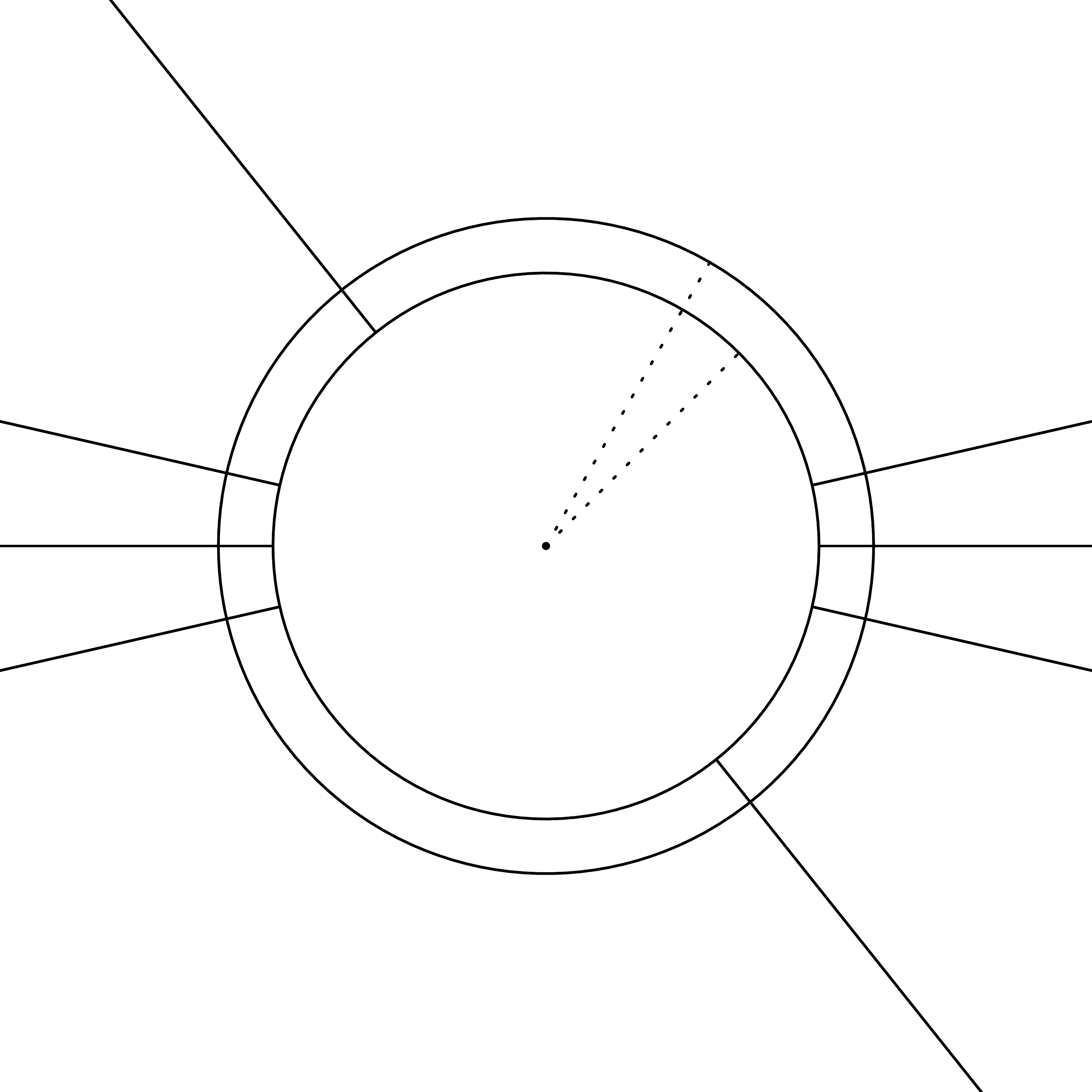}
        \put (58,55) {\footnotesize $n^{-2/7-\epsilon}$}
        \put (56,65) {\footnotesize $\delta$}
        \put (50,47) {\footnotesize $\alpha$}
        \put (65,77) {\footnotesize $\partial \hat{\mathbb{D}}_+$}
        \put (35,32) {\footnotesize $\partial \mathbb{D}_+$}
    \end{overpic}
    \caption{Jump contours near $z=\alpha$ of $\boldsymbol{R}(z)$. The two circles are oriented counterclockwise; all other contours are oriented from left to right. The jumps on the external rays are the jumps in Figure \ref{fig:Local-contours}, conjugated by the global parametrix.}
    \label{fig:Concentric-discs}
\end{figure}
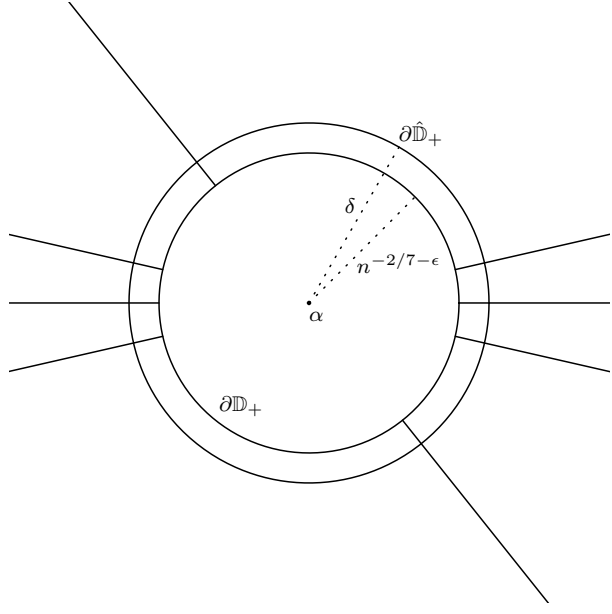

We now define the local parametrix that we will actually use in our analysis. We define the function $\mathcal{P}(z)$ in $\hat{\DD}_+$:
    \begin{equation}
        \mathcal{P}(z) := 
        \begin{cases}
            M_{reg}(z)U^{-1}(z)M_{sing}(z), & z \in \hat{\DD}_+\setminus \DD_+,\\
            P_0(z), & z\in\DD_+,
        \end{cases}
    \end{equation}
where:
    \begin{align}
        P_0(z)&:=E(z)\left[\hat{\Psi}(n^{3/7}\xi(z);n^{2/7}\boldsymbol{\eta}(z;\vec{t}),n^{5/7}\boldsymbol{\mu}(\vec{t}),n^{6/7}\boldsymbol{\nu}(z;\vec{t})) \oplus 1\right]\hat{\sigma}_{34} e^{-n\mathcal{Q}(z) + n\boldsymbol{G}(z)},\\
        E(z) &:= M_{reg}(z)\left(\mathcal{D}(z)U(z)\right)^{-1}M_{sing}(z) \hat{\sigma}_{34} \left[ g^{-1}(\xi(z)) \oplus 1\right],
    \end{align}
and $\mathcal{Q}(z) := \text{diag }(K_0(z), K_0(z), \Omega_3(z),K_0(z))$, with $K_0(z)$ is as defined in Proposition \ref{Omega-breakup}. Note that
$E(z)$ is well-defined, bounded and analytic in $\DD_+$, in particular because $\mathcal{D}(z)U(z)$ is, by Proposition \ref{U-def-prop}. We then
take
    \begin{equation}
        \hat{M}(z) := 
        \begin{cases}
            M(z), & z\in \CC\setminus \hat{\DD}_{\pm},\\
            \mathcal{P}(z), & z\in \hat{\DD}_+,\\
            [\sigma_3 \oplus \sigma_3]\mathcal{P}(-z)[\sigma_3 \oplus \sigma_1], & z\in \DD_-,
        \end{cases}
    \end{equation}
and finally define
    \begin{equation}
        \boldsymbol{R}(z) := \boldsymbol{S}(z)\hat{M}^{-1}(z).
    \end{equation}
Then $\boldsymbol{R}(z)$ is the solution to a small-norm Riemann-Hilbert problem:
    \begin{prop}
         On $\partial \hat{\DD}_{\pm}$, $\partial \DD_{\pm}$,
            \begin{equation}
                \boldsymbol{R}_+(z) = \boldsymbol{R}_-(z)\cdot
                \begin{cases}
                    J_{\boldsymbol{R}}^{(\pm)}(z), & z\in\partial \hat{\DD}_{\pm},\\
                    \mathbb{I}, & z\in\partial \DD_{\pm}
                \end{cases}
            \end{equation}
        where the matrices $J_{\boldsymbol{R}}^{(\pm)}(z)$ admit the $n\to \infty$ expansions
            \begin{equation}
                J_{\boldsymbol{R}}^{(\pm)}(z) \sim \mathbb{I} + \sum_{k=1}^{\infty} J_k^{(\pm)}(z)n^{-k/7}.
            \end{equation}
        Furthermore, the jumps of $\boldsymbol{R}(z)$ elsewhere are exponentially close to the identity matrix, and consequently,
            \begin{equation}
                \boldsymbol{R}(z) \sim \mathbb{I} + \sum_{k=1}^{\infty} \boldsymbol{R}_{k}(z) n^{-k/7}.
            \end{equation}
    \end{prop}
\begin{proof}
    By construction and our choice of Stokes parameters, the jumps of $P_0(z)$ match the jumps of $\boldsymbol{S}(z)$ \textit{exactly} inside the
    disc $\DD_+$. We have seen that
    
    The jumps of ${\bf R}(z)$ on $\partial \DD_+$ are
    \begin{align*}
        J_{\bf R}\big|_{\partial \DD_+}(z) &= P(z) M_{sing}(z)^{-1}U(z)M_{reg}^{-1}(z)\\
        &=M_{reg}(z)(\mathcal{D}(z)U(z))^{-1}\mathcal{D}(z)U(z)M_{reg}^{-1}(z)\\
        &=\mathbb{I}
    \end{align*}
The jumps of ${\bf R}(z)$ on $\partial \hat{\DD}_+$ are
    \begin{align*}
        J_{\bf R}\big|_{\partial \hat{\DD}_+}(z) &= M_{reg}(z)U^{-1}(z)M_{sing}(z)M^{-1}(z) = M_{reg}(z)U^{-1}(z)M_{reg}^{-1}(z) \sim \mathbb{I} + \sum_{k=1}^{\infty} J_k(z) n^{-k/7},
    \end{align*}
where the matrices $J_k(z)$ are $n$-independent.
\end{proof}

\subsection{Proof of Theorem \ref{propA}}
As a corollary of the analysis above, we obtain that
\begin{cor}
    The recurrence coefficients $R_n$, $f_n$ admit the following large-$n$ expansions:
    \begin{equation}\label{R-f-expansion}
        R_n(\tau,t,H) \sim \mathfrak{r}_c + \sum_{k=1}^{\infty} \mathfrak{r}_{k}(\eta,\mu,\nu)n^{-k/7},\qquad\qquad f_n(\tau,t,H) \sim \mathfrak{f}_c + \sum_{k=1}^{\infty} \mathfrak{f}_{k}(\eta,\mu,\nu)n^{-k/7},
    \end{equation}
where $\mathfrak{r}_c = \frac{12}{5}$, $\mathfrak{f}_c = \frac{6}{5}$, and the above expansions hold locally uniformly in $\eta,\mu,\nu$.
\end{cor}
\begin{proof}
    The results of the previous section show that
        \begin{equation}
            \boldsymbol{R}(z) \sim \mathbb{I} + \sum_{k=1}^{\infty} \boldsymbol{R}_k(z) n^{-k/7}
        \end{equation}
    for $|z|$ sufficiently large. Furthermore, we know from Proposition \ref{prop:RHP-coeff-expression} that
            \begin{align*}
                f_n &= \frac{\tau e^{H}}{t} [Y^{(1)}]_{12}[Y^{(1)}]_{41},\\
                \tilde{R}_n &= \frac{\tau e^{H}}{t}\left([Y^{(1)}]_{43} - [Y^{(1)}]_{32}\right).
            \end{align*}
    The result then follows from tracing back the transformations of the steepest descent analysis in the above formula; the asymptotic expansion in $n^{-1/7}$ arises from the fact that $\boldsymbol{R}(z)$ has an asymptotic expansion in $n^{-1/7}$, and the fact that we can relate the expansion of $\tilde{R}_n$ to that of $R_n$ by simply replacing $H\to -H$.
\end{proof}
This corollary provides half of the proof of Theorem \ref{propA}. The other part of the statement of this theorem regards the explicit structure of the first few leading terms of the expansion of $f_n$, $R_n$. In principle, one can furnish a proof of this fact by explicitly calculating $\boldsymbol{R}_1(z)$,
$\boldsymbol{R}_2(z)$, similarly to \cite{DK0}, for instance. However, we will take a different approach: we will utilize the fact that these recurrence coefficients satisfy the discrete equations \eqref{dString}, in order to extract the first few terms. With this idea in mind, we now complete the proof of Theorem \ref{propA}.

\begin{proof}\textit{(Of Theorem \ref{propA}).}
As we have already noted, we have only to show that
                \begin{align*}
                    \mathfrak{r}_{1}(\eta,\mu,\nu) &= \mathfrak{f}_{1}(\eta,\mu,\nu) = 0,\\
                    \mathfrak{r}_{2}(\eta,\mu,\nu) &= -\frac{c^2}{10}\left(3U(\eta,\mu,\nu)-\eta\right),\qquad \mathfrak{f}_{2}(\eta,\mu,\nu) =  -\frac{c^2}{10}\left(3U(\eta,\mu,\nu)+2\eta\right),\\
                    \mathfrak{r}_{3}(\eta,\mu,\nu) &= \frac{3c^3}{10}V(\eta,\mu,\nu),\qquad\quad\,\,\,\,\qquad \mathfrak{f}_{3}(\eta,\mu,\nu) = 0,
                \end{align*}
where $c:=5^{1/7}2^{5/7}$ is a positive constant. If we insert the expansions \eqref{R-f-expansion} into the discrete string equations \eqref{dS-loc-A},
\eqref{dS-loc-B}, we obtain at order $n^{-2/7}$ the relations
    \begin{equation*}
        0  = -\frac{5}{12}\mathfrak{f}_1(\mathfrak{r}_1-\mathfrak{f}_1),\qquad\qquad 0 = -\frac{5}{24}(\mathfrak{r}_1^2-\mathfrak{f}_1^2);
    \end{equation*}
which imply that $\mathfrak{r}_1 = \mathfrak{f}_1$. At order $n^{-3/7}$, the equations \eqref{dS-loc-A}, \eqref{dS-loc-B} then read that
    \begin{align*}
        0 &= 200^{1/7}\eta \mathfrak{r}_1+\frac{5}{3}\mathfrak{r}_1\mathfrak{f}_2-\frac{5}{3}\mathfrak{r}_1\mathfrak{r}_2 - \frac{25}{36}\mathfrak{r}_1^3,\\
        0 &= 200^{1/7}\eta \mathfrak{r}_1+\frac{5}{3}\mathfrak{r}_1\mathfrak{f}_2-\frac{5}{3}\mathfrak{r}_1\mathfrak{r}_2 + \frac{25}{108}\mathfrak{r}_1^3,
    \end{align*}
which, if we sum these two equations, implies that $\mathfrak{r}_1^3 = 0$, and so $\mathfrak{r}_1 = \mathfrak{f}_1 = 0$. Finally, since we know from the corollary that the functions $\mathfrak{r}_k,\mathfrak{f}_k$ are differential polynomials in $U,V$, we can deduce that the equations \eqref{dS-loc-A}, \eqref{dS-loc-B} at orders $n^{-4/7}$, $n^{-5/7}$ are equivalent to the string equations \eqref{string-equation}, upon identifying
    \begin{align*}
        \mathfrak{r}_{2}(\eta,\mu,\nu) &= -\frac{c^2}{10}\left(3U(\eta,\mu,\nu)-\eta\right),\qquad \mathfrak{f}_{2}(\eta,\mu,\nu) =  -\frac{c^2}{10}\left(3U(\eta,\mu,\nu)+2\eta\right),\\
        \mathfrak{r}_{3}(\eta,\mu,\nu) &= \frac{3c^3}{10}V(\eta,\mu,\nu),\qquad\quad\,\,\,\,\qquad \mathfrak{f}_{3}(\eta,\mu,\nu) = 0.
    \end{align*}

\end{proof}

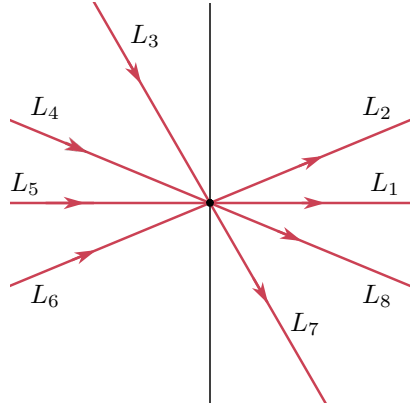
\begin{figure}
    \centering
    \begin{overpic}[scale=.25]{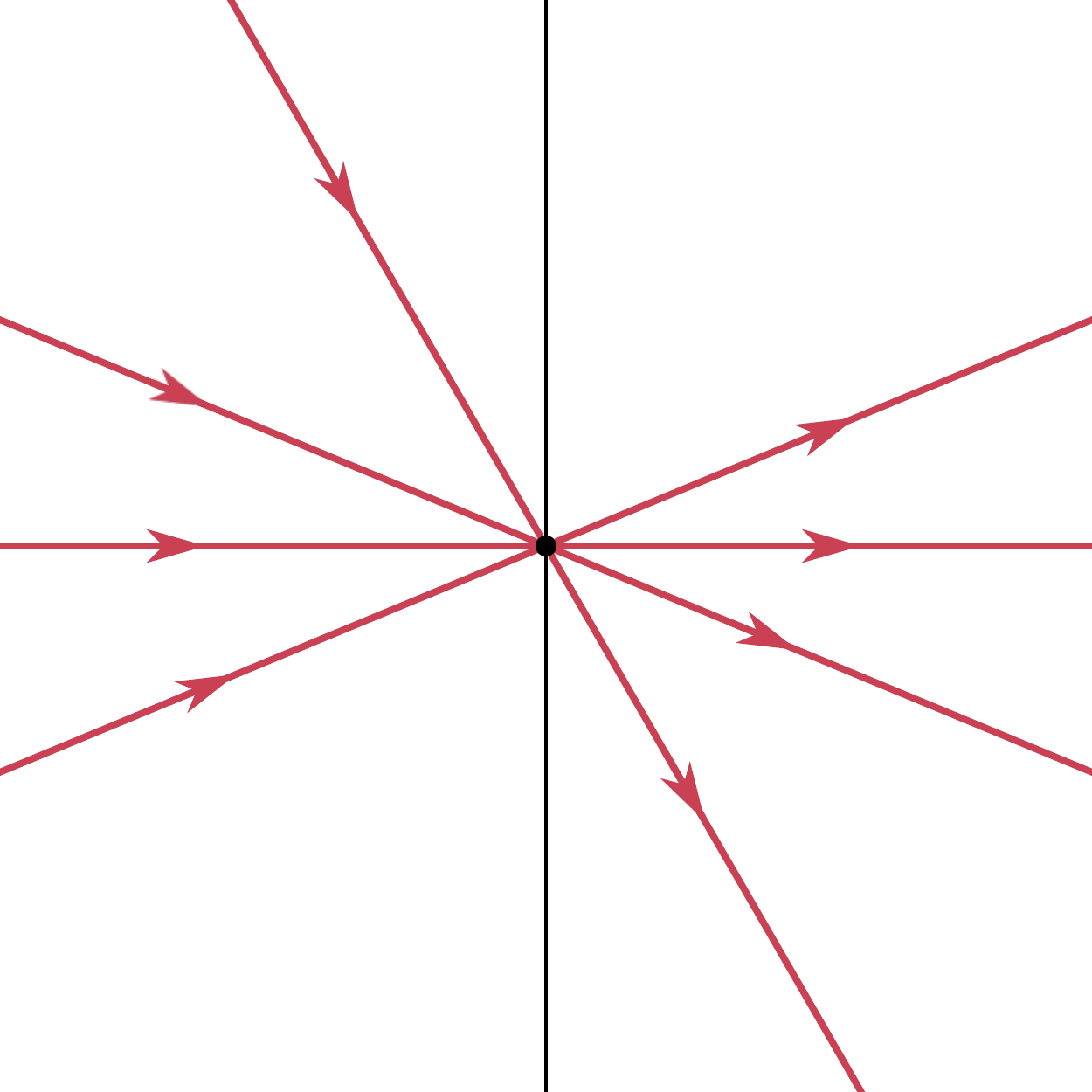}
        \put (90,53) {$L_1$}
        \put (88,72) {$L_2$}
        \put (30,90) {$L_3$}
        \put (5,72) {$L_4$}
        \put (0,53) {$L_5$}
        \put (5,25) {$L_6$}
        \put (70,18) {$L_7$}
        \put (88,25) {$L_8$}
    \end{overpic}
    \caption{The jump contours of the $\Psi$ Riemann-Hilbert problem in the $\xi$-plane.}
    \label{fig:Local-contours}
\end{figure}

\newpage
\appendix
\section{Regularization factor in the main theorem.}\label{appendix-regularization}
Here, we define the regularization factor $\log Z_{reg}$ of Theorem \ref{main-theorem}.

\begin{defn}
    The function $\log Z_{reg}$ in the variables $\eta,\mu,\nu,n$ is given by
        \begin{equation}
            \log Z_{reg} = \sum_{k=2}^{13} M_k(c_{\eta}\eta,c_{\mu}\mu,c_{\nu}\nu)n^{-k/7},
        \end{equation}
    where
        \begin{align*}
            M_2(\eta,\mu,\nu) &= \frac{276}{125}\eta,\qquad M_3(\eta,\mu,\nu) = 0,\qquad M_4(\eta,\mu,\nu) = \frac{2288}{625}\eta^2,\qquad M_5(\eta,\mu,\nu) = 0,\\
            M_6(\eta,\mu,\nu) &= \frac{115624}{9375}\eta^3 + \frac{\nu}{375}(828\chi-925),\qquad M_7(\eta,\mu,\nu) = 0,\\
            M_8(\eta,\mu,\nu) &= -\frac{4\eta}{1875}\left(14360\eta^3 + (2947\chi-2825)\nu\right),\qquad M_9(\eta,\mu,\nu) = -2\eta^2,\\ 
            M_{10}(\eta,\mu,\nu) &= \frac{8057728}{78125}\eta^5-\frac{6}{25}\mu^2+\frac{16}{9375}(18652\chi-18625)\nu\eta^2,\\
            M_{11}(\eta,\mu,\nu) &= -4\eta^2,\\
            M_{12}(\eta,\mu,\nu) &= \frac{2 \left(2947 \chi^{2}-5650 \chi +1925\right) \nu^{2}}{1875}+\frac{1209088 \eta^{3} \left(\chi -\frac{17815}{18892}\right) \nu}{9375}+\frac{375329536 \eta^{6}}{1171875}-\frac{144 \eta  \,\mu^{2}}{125},\\
            M_{13}(\eta,\mu,\nu) &= \frac{32}{3}\eta^3+2(1-\chi)\nu.
        \end{align*}
\end{defn}

As was observed in \cite{BleherDeano} in the critical case in the one matrix model, this regularization factor arises from the fact that, although the partition function is not analytic near the critical point, it does admit finitely many continuous derivatives there. Since analytic terms in the `global variables' (here, $\tau,t,H$) transfer to analytic terms in the `local variables' (here, $\eta,\mu,\nu$) under the multi-scaling limit, it follows that that the first few Taylor coefficients of the partition function will appear as a polynomial contribution in the local scaling variables. Thus, the `interesting' part of the limit appears when one subtracts off the regular part of the partition function, i.e. the first few analytic terms. 

\section{Expansion of the uniformizing coordinate near the branch points.} \label{Appendix-A}
Here, we list the relevant expansions of the uniformizing coordinate on each sheet of the spectral curve (which agree, up to redefinition of the parameter $A$, with the expansion of the uniformizing coordinate on the \textit{modified} spectral curve). 

\begin{itemize}
    \item \textit{Expansion at $z= +\alpha$}. Let $\zeta:=z-\alpha$, $A = A(1,1,1)=\sqrt{6/5}$; as $\zeta\to 0$, we have the expansions:
        \begin{align}
            u_1(z) &= 1+ \left(\frac{3}{4A}\right)^{1/3}\zeta^{1/3} + \frac{3}{4}\left(\frac{3}{4A}\right)^{2/3}\zeta^{2/3} + \frac{21}{64A}\zeta + \frac{37}{192}\left(\frac{3}{4A}\right)^{4/3}\zeta^{4/3} + \OO( \zeta^{5/3})\\
            u_2(z) &=
            \begin{cases}
                1+ \left(\frac{3}{4A}\right)^{1/3}\omega^2\zeta^{1/3} + \frac{3}{4}\left(\frac{3}{4A}\right)^{2/3}\omega\zeta^{2/3} + \frac{21}{64A}\zeta + \frac{37}{192}\left(\frac{3}{4A}\right)^{4/3}\omega^2\zeta^{4/3} + \OO( \zeta^{5/3}),
                & \text{Im }\zeta >0,\\
                 1+ \left(\frac{3}{4A}\right)^{1/3}\omega\zeta^{1/3} + \frac{3}{4}\left(\frac{3}{4A}\right)^{2/3}\omega^2\zeta^{2/3} + \frac{21}{64A}\zeta + \frac{37}{192}\left(\frac{3}{4A}\right)^{4/3}\omega\zeta^{4/3} + \OO( \zeta^{5/3}), & \text{Im }\zeta <0
            \end{cases},\\
            u_3(z) &=-\frac{1}{3} + \frac{1}{64A}\zeta - \frac{27}{16384A^2}\zeta^2 + \frac{891}{4194304A^3}\zeta^3 + \OO(\zeta^4),\\
            u_4(z) &=
            \begin{cases}
                1+ \left(\frac{3}{4A}\right)^{1/3}\omega\zeta^{1/3} + \frac{3}{4}\left(\frac{3}{4A}\right)^{2/3}\omega^2\zeta^{2/3} + \frac{21}{64A}\zeta + \frac{37}{192}\left(\frac{3}{4A}\right)^{4/3}\omega\zeta^{4/3} + \OO( \zeta^{5/3}), & \text{Im }\zeta >0,\\
                1+ \left(\frac{3}{4A}\right)^{1/3}\omega^2\zeta^{1/3} + \frac{3}{4}\left(\frac{3}{4A}\right)^{2/3}\omega\zeta^{2/3} + \frac{21}{64A}\zeta + \frac{37}{192}\left(\frac{3}{4A}\right)^{4/3}\omega^2\zeta^{4/3} + \OO( \zeta^{5/3}), & \text{Im }\zeta <0.
            \end{cases}
        \end{align}
    \item \textit{Expansion at $z= -\alpha$}. Let $\zeta:=z+\alpha$, $A=\sqrt{6/5}$; as $\zeta\to 0$, we have the expansions:
        \begin{align}
            u_1(z) &= 
            \begin{cases}
            -1+ \left(\frac{3}{4A}\right)^{1/3}\omega\zeta^{1/3} - \frac{3}{4}\left(\frac{3}{4A}\right)^{2/3}\omega^2\zeta^{2/3} + \frac{21}{64A}\zeta - \frac{37}{192}\left(\frac{3}{4A}\right)^{4/3}\omega\zeta^{4/3} + \OO( \zeta^{5/3}), & \text{Im } \zeta >0,\\
            -1+ \left(\frac{3}{4A}\right)^{1/3}\omega^2\zeta^{1/3} - \frac{3}{4}\left(\frac{3}{4A}\right)^{2/3}\omega\zeta^{2/3} + \frac{21}{64A}\zeta - \frac{37}{192}\left(\frac{3}{4A}\right)^{4/3}\omega^2\zeta^{4/3} + \OO( \zeta^{5/3}), & \text{Im } \zeta <0,
            \end{cases}\\
            u_2(z) &= 
            \begin{cases}
            -1+ \left(\frac{3}{4A}\right)^{1/3}\omega^2\zeta^{1/3} - \frac{3}{4}\left(\frac{3}{4A}\right)^{2/3}\omega\zeta^{2/3} + \frac{21}{64A}\zeta - \frac{37}{192}\left(\frac{3}{4A}\right)^{4/3}\omega^2\zeta^{4/3} + \OO( \zeta^{5/3}),
             & \text{Im } \zeta >0,\\
             -1+ \left(\frac{3}{4A}\right)^{1/3}\omega\zeta^{1/3} - \frac{3}{4}\left(\frac{3}{4A}\right)^{2/3}\omega^2\zeta^{2/3} + \frac{21}{64A}\zeta - \frac{37}{192}\left(\frac{3}{4A}\right)^{4/3}\omega\zeta^{4/3} + \OO( \zeta^{5/3}), & \text{Im } \zeta <0,
            \end{cases}\\
            u_3(z) &= -1+ \left(\frac{3}{4A}\right)^{1/3}\zeta^{1/3} - \frac{3}{4}\left(\frac{3}{4A}\right)^{2/3}\zeta^{2/3} + \frac{21}{64A}\zeta - \frac{37}{192}\left(\frac{3}{4A}\right)^{4/3}\zeta^{4/3} + \OO( \zeta^{5/3}),\\
            u_4(z) &= \frac{1}{3} + \frac{1}{64A}\zeta + \frac{27}{16384A^2}\zeta^2 + \frac{891}{4194304A^3}\zeta^3 + \OO(\zeta^4).
        \end{align}
\end{itemize}

\section{Model Riemann Hilbert Problem for the $(3,4)$ String Equation.}\label{STRING-APPENDIX}
As was demonstrated in \cite{DHL2}, the string equation \eqref{string-equation} is equivalent to the isomonodromy deformation equations of a linear differential equationwith rational coefficients. More concretely, define contours

\begin{figure}
    \begin{center}
    \begin{overpic}[scale=.5]{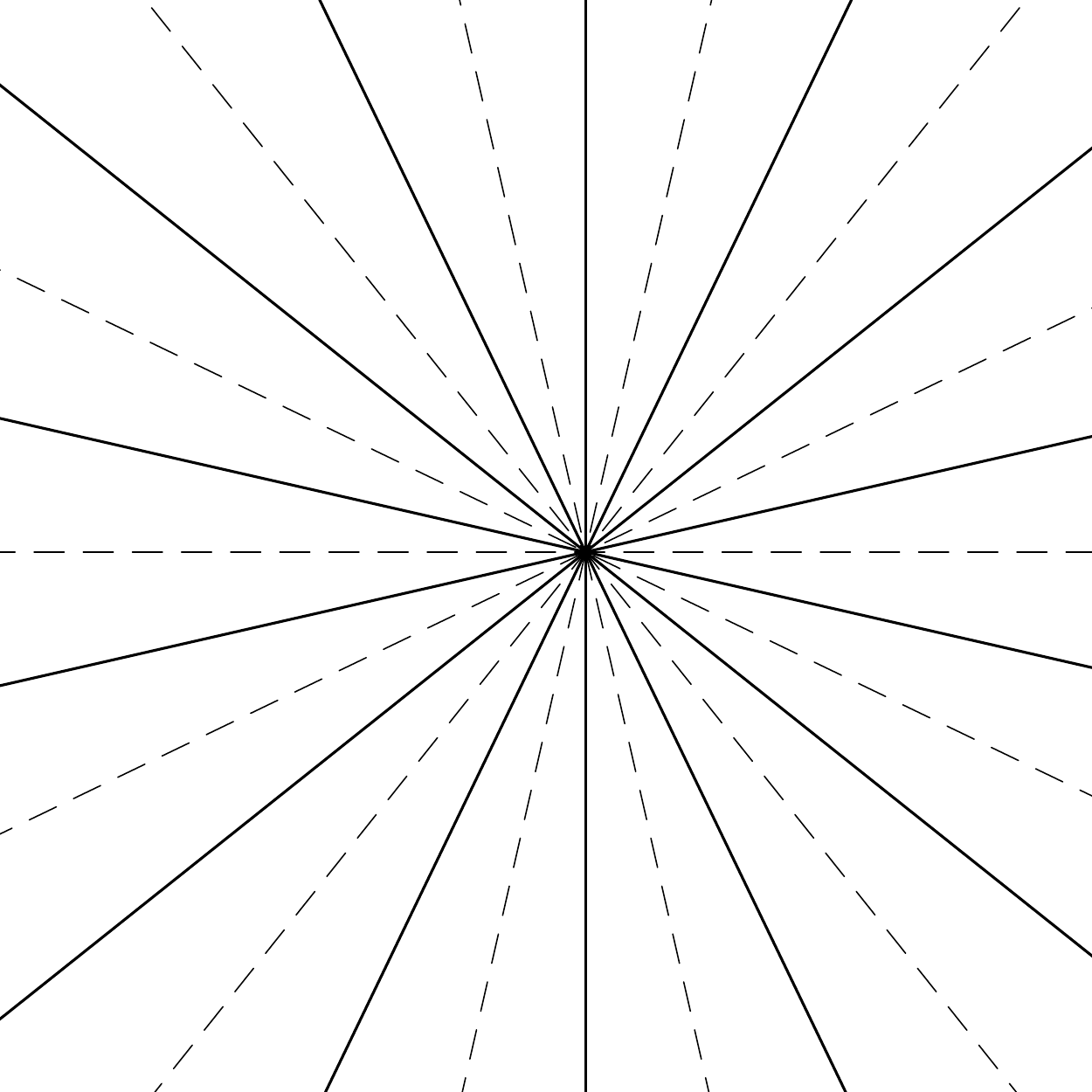}
                \put (100,50) {$\text{Re } \lambda$}
          	\put (103,60) {$\gamma_1$}
      		  \put (103,37) {$\gamma_{-1}$}
     		  \put (98,82) {$\gamma_2$}
      		  \put (98,15) {$\gamma_{-2}$}
                \put (78,95) {$\gamma_3$}
      		  \put (78,3) {$\gamma_{-3}$}
                \put (55,99) {$\gamma_4$}
      		  \put (55,1) {$\gamma_{-4}$}
                \put (34,98) {$\gamma_5$}
      		  \put (34,1) {$\gamma_{-5}$}
                \put (5,92) {$\gamma_6$}
      		  \put (5,6) {$\gamma_{-6}$}
                \put (-5,61) {$\gamma_7$}
      		  \put (-5,36) {$\gamma_{-7}$}
                \put (-4,48) {$\rho$}
                \put (80,57){\small \rotatebox{15}{\textcolor{BrickRed}{$\mathbb{I} + s_1E_{21}$}} } 
                \put (80,44){\small \rotatebox{-11}{\textcolor{BrickRed}{$\mathbb{I} + s_{-1}E_{31}$}} } 
                \put (76,70){\small \rotatebox{42}{\textcolor{BrickRed}{$\mathbb{I} + s_2E_{23}$}}  } 
                \put (79,30){\small \rotatebox{-36}{\textcolor{BrickRed}{$\mathbb{I} + s_{-2}E_{32}$}}  } 
                \put (66,83){\small \rotatebox{65}{\textcolor{BrickRed}{$\mathbb{I} + s_{3}E_{13}$}}  } 
                \put (69,19){\small \rotatebox{-61}{\textcolor{BrickRed}{$\mathbb{I} + s_{-3}E_{12}$}}  } 
                \put (50,85){\small \rotatebox{90}{\textcolor{BrickRed}{$\mathbb{I} + s_{4}E_{12}$}}  } 
                \put (54,20){\small \rotatebox{-90}{\textcolor{BrickRed}{$\mathbb{I} + s_{-4}E_{13}$}}  } 
                \put (35,90){\small \rotatebox{-63}{\textcolor{BrickRed}{$\mathbb{I} + s_{5}E_{32}$}}  } 
                \put (30,8){\small \rotatebox{65}{\textcolor{BrickRed}{$\mathbb{I} + s_{-5}E_{23}$}}  } 
                \put (18,80){\small \rotatebox{-38}{\textcolor{BrickRed}{$\mathbb{I} + s_{6}E_{31}$}}  } 
                \put (12,19){\small \rotatebox{42}{\textcolor{BrickRed}{$\mathbb{I} + s_{-6}E_{21}$}}  } 
                \put (6,61){\small \rotatebox{-12}{\textcolor{BrickRed}{$\mathbb{I} + s_{7}E_{21}$}}  } 
                \put (6,40){\small \rotatebox{12}{\textcolor{BrickRed}{$\mathbb{I} + s_{-7}E_{31}$}}  } 
                \put (6,50){\small \rotatebox{0}{\textcolor{BrickRed}{$\mathcal{S}$}}  } 
    		\end{overpic}
    \end{center}
    \caption{The Stokes lines $\gamma_j$ for the Riemann-Hilbert problem for $\Psi(\xi;\eta,\mu,\nu)$. Each of the Stokes sectors is bisected by an anti-Stokes line, depicted by a dashed line. All contours are oriented \textit{outwards} from the origin. The anti-Stokes line $(-\infty,0]$ is labelled by $\rho$. The Stokes matrix $S_k$ is the matrix associated to the parameter $s_k$; these parameters are not all independent, and must satisfy the equation
    $S_{-7}\cdots S_{-1}S_{1}\cdots S_{7} = \mathcal{S}^T$.}
    \label{fig:43-equation-jumps}
\end{figure}
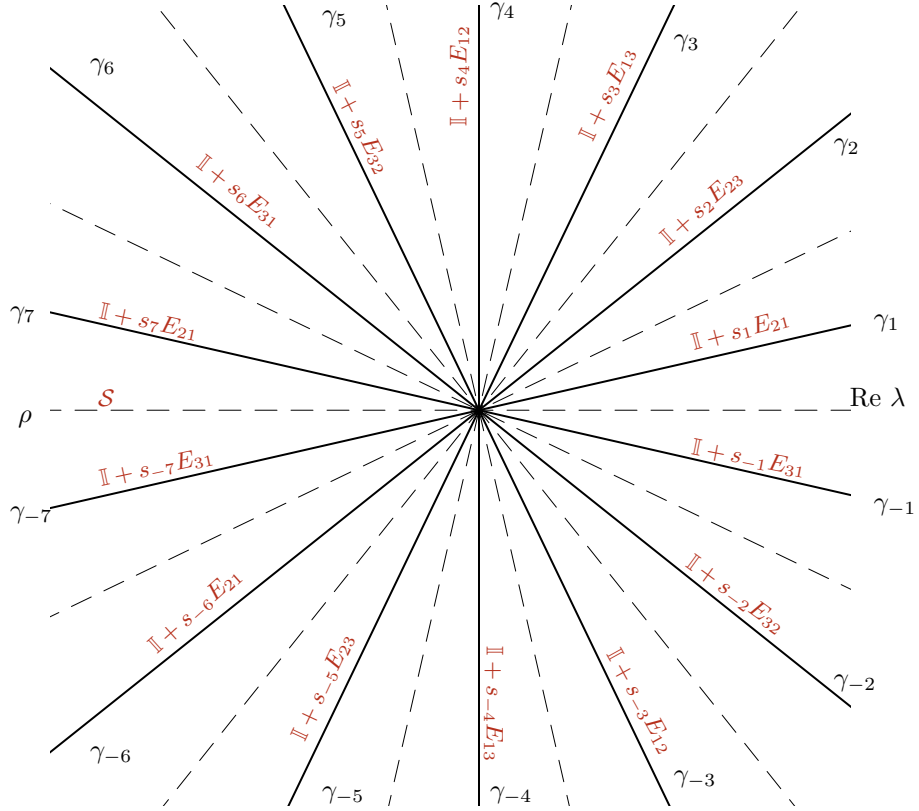

\begin{align*}
        \gamma_{\pm k} := \left\{\lambda \big| \arg \lambda = \pm \frac{\pi}{14} \pm \frac{\pi}{7}(k-1) \right\}, \qquad k = 1,...,7,
    \end{align*}
and $\rho := (-\infty,0)$. The string equation can be seen to arise from the following Riemann-Hilbert problem for a $3\times 3$ 
sectionally analytic function $\Psi(\xi;\eta,\mu,\nu)$:
    \begin{equation}
            \begin{cases}
                \Psi_{+}(\xi;\eta,\mu,\nu) =  \Psi_{-}(\xi;\eta,\mu,\nu) S_k, & \xi \in \gamma_k,\qquad k = \pm 1, ...,\pm 7,\\
                \Psi_{+}(\xi;\eta,\mu,\nu) =  \Psi_{-}(\xi;\eta,\mu,\nu) \mathcal{S}, & \xi \in \rho,\\
                \Psi(\xi;\eta,\mu,\nu) = g(\xi)\left[\mathbb{I} + \frac{\Psi_1}{\xi^{1/3}} + \frac{\Psi_2}{\xi^{2/3}} + \OO(\xi^{-1})\right]e^{\Theta(\xi;\eta,\mu,\nu)}, & \xi \to \infty,
            \end{cases}
        \end{equation}
    where $g(\xi), \Theta(\xi;\eta,\mu,\nu)$ are given by
\begin{equation}\label{gauge-matrix}
        g(\xi) := 
        \frac{i}{\sqrt{3}}\underbrace{\begin{pmatrix}
            \xi^{1/3} & 0 & 0\\
            0 & 1 & 0\\
            0 & 0 & \xi^{-1/3}
        \end{pmatrix}}_{\xi^{\Delta/3}}
        \underbrace{\begin{pmatrix}
            1 & \omega & \omega^2\\
            1 & 1 & 1\\
            1 & \omega^2 & \omega
        \end{pmatrix}}_{-i\sqrt{3}\mathcal{U}},
    \end{equation}
    \begin{equation}
        \Theta(\xi;\eta,\mu,\nu) := \text{diag }(\vartheta_1(\xi;\eta,\mu,\nu),\vartheta_2(\xi;\eta,\mu,\nu),\vartheta_3(\xi;\eta,\mu,\nu)),
    \end{equation}
with $\vartheta_j(\xi;\eta,\mu,\nu) = \frac{3}{7}\omega^{j-1}\xi^{7/3} + \omega^{1-j}\eta\xi^{5/3} + \omega^{1-j}\mu\xi^{2/3} + \omega^{j-1}\nu\xi^{1/3}$, the jump matrices $S_k$ are given in Figure \eqref{fig:43-equation-jumps}, and the matrix $\mathcal{S}$ is
    \begin{equation}
        \mathcal{S} := 
        \begin{psmallmatrix}
                0 & 1 & 0\\
                0 & 0 & 1\\
                1 & 0 & 0
            \end{psmallmatrix}.
    \end{equation}
These matrices must satisfy the following constraint equation:
    \begin{equation}\label{Stokes-Equation}
        S_{-7}\cdots S_{-1}S_{1}\cdots S_{7} = \mathcal{S}^T.
    \end{equation}
The matrices $\Psi_j := \Psi_j(\eta,\mu,\nu)$, $j=1,2$, are given by
    \begin{align}
        \Psi_1(\eta,\mu,\nu) &= 
            \begin{psmallmatrix}
                H_1 & 0 & 0\\
                0 & \omega^2 H_1 & 0\\
                0 & 0 & \omega H_1
            \end{psmallmatrix}\\
        \Psi_2(\eta,\mu,\nu) &=
        \begin{psmallmatrix}
            \frac{1}{2}(H_1)^2+\frac{1}{2}H_2 & -\frac{i\omega^2\sqrt{3}}{12} U & \frac{i\omega\sqrt{3}}{12} U \\
            \frac{i\omega^2\sqrt{3}}{12} U& \omega \left(\frac{1}{2}(H_1)^2+\frac{1}{2}H_2\right) & -\frac{i\sqrt{3}}{12} U\\
            -\frac{i\omega\sqrt{3}}{12}  U &  \frac{i\sqrt{3}}{12} U& \omega^2 \left(\frac{1}{2}(H_1)^2+\frac{1}{2}H_2\right)
        \end{psmallmatrix},
    \end{align}
where $H_1$, $H_2$ are the Hamiltonians given above, and $U,V$ are solutions to the string equation. The above jump conditions and 
asymptotics, along with the determination of $\Psi_j$, $j=1,2$, determine the solution to the above Riemann-Hilbert problem uniquely. We remark that the
coefficients $\Psi_j(\eta,\mu,\nu)$ carry the symmetry
    \begin{equation}
        \Psi_j(\eta,\mu,\nu) = \omega^{-j}\mathcal{S}^T \Psi_j(\eta,\mu,\nu) \mathcal{S}.
    \end{equation}
The matrix-valued function $\Psi(\xi;\eta,\mu,\nu)$ is essentially what we will use in the construction of the local parametrix. However, we must first 
``dress up'' this Riemann-Hilbert problem so that it indeed has the required properties for steepest descent analysis. We set
    \begin{equation}
        \hat{\Psi}(\xi;\eta,\mu,\nu) := 
            \begin{cases}
                \Psi(\xi;\eta,\mu,\nu)
                    \begin{psmallmatrix}
                        1 & 0 & 0\\
                        0 & 0 & 1\\
                        0 & 1 & 0
                    \end{psmallmatrix}, & \text{Im } \xi > 0,\\
                \Psi(\xi;\eta,\mu,\nu)\begin{psmallmatrix}
                        1 & 0 & 0\\
                        0 &-1 & 0\\
                        0 & 0 & 1
                    \end{psmallmatrix}, & \text{Im } \xi < 0.
            \end{cases}
    \end{equation}
We also choose the particular solution with Stokes parameters
    \begin{align}
        s_1 &= 0,\qquad s_2 = -1,\qquad s_3 = 0,\qquad s_4 = 0,\qquad s_5 = 1,\qquad s_6 = -1, \qquad s_7 = 0,\nonumber\\
        s_{-1} &= 0,\qquad s_{-2} = 1,\qquad s_{-3} = -1,\qquad s_{-4} = 0,\qquad s_{-5} = 0,\qquad s_{-6} = 1,\qquad s_{-7} = 0.
    \end{align}
One can readily check that the above set of parameters is indeed a solution to the constraint equation \eqref{Stokes-Equation}.

The matrix $\hat{\Psi}(\xi;\eta,\mu,\nu)$, with this choice of Stokes data, is then the solution to the following Riemann-Hilbert problem:
    \begin{align}
        \hat{\Psi}_+(\xi;\eta,\mu,\nu) &= \hat{\Psi}_-(\xi;\eta,\mu,\nu)\times
            \begin{cases}
                    \begin{psmallmatrix}
                       1 & 0 & 0\\
                       0 & 0 & -1\\
                       0 & 1 & 0\\
                   \end{psmallmatrix}, & \text{Im }\xi > 0,\\
                    \begin{psmallmatrix}
                       1 & 0 & 0\\
                       0 & 1 & 0\\
                       0 & -1 & 1\\
                   \end{psmallmatrix}, & \xi \in \gamma_{2} \cup \gamma_{-2},\\
                    \begin{psmallmatrix}
                       1 & 0 & 0\\
                       0 & 1 & -1\\
                       0 & 0 & 1\\
                   \end{psmallmatrix}, & \xi \in \gamma_{5},\\
                    \begin{psmallmatrix}
                       1 & 0 & 0\\
                       1 & 1 & 0\\
                       0 & 0 & 1\\
                   \end{psmallmatrix}, &\xi \in \gamma_{6}  \cup \gamma_{-6},\\
                    \begin{psmallmatrix}
                       0 & 1 & 0\\
                       -1 & 0 & 0\\
                       0 & 0 & 1\\
                   \end{psmallmatrix}, & \text{Im }\xi < 0,\\
                   \begin{psmallmatrix}
                       1 & 1 & 0\\
                       0 & 1 & 0\\
                       0 & 0 & 1\\
                   \end{psmallmatrix}, & \xi \in \gamma_{-3},
                \end{cases}\\
            \hat{\Psi}(\xi;\eta,\mu,\nu) &= \hat{g}(\xi) \left[\mathbb{I}+ \sum_{k=1}^{\infty} \frac{\hat{\Psi}_k}{\xi^{k/3}}\right]e^{\hat{\Theta}(\xi;\eta,\mu,\nu)},\qquad\qquad \xi\to \infty .\label{Psi-hat-asymptotics}
    \end{align} 
The model RHP $\hat{\Psi}$ described above is what we will use in the construction of local parametrices (cf. Equation \eqref{D-matrix}).

The explicit form of the Hamiltonians $H_1,H_2,H_5$ as functions of $U,V$ are
    \begin{align*}
        H_1 &:= -\frac{1}{12}U'U''' +\frac{1}{24}(U'')^2 + \frac{3}{8}U(U')^2 +\frac{1}{2}(V')^2 - \frac{1}{8}U^4 - \frac{3}{2}UV^2 -\frac{5}{6}\eta\left(\frac{1}{4}(U')^2 - \frac{1}{2}U^3 - 3 V^2\right) + \frac{\nu}{2}U^2,\\
        H_2 &:= \frac{1}{6}U'''V' -\frac{1}{2}VUU''+\frac{1}{4}V(U')^2 - UU'V'+U^3V+V^3+\frac{5}{6}\eta\left(U''-3U^2\right)V + \frac{1}{3}\mu(U''-3U^2) + 2\nu V,\\
        H_5 &:= \frac{1}{144}U'U''U'''+\frac{1}{12}V'U'''-\frac{1}{16}U^2U'U'''+\frac{1}{144}U(U''')^2+\frac{1}{16}U(U')^2U''-\frac{1}{4}UVU'V' + \frac{3}{8}V^4\\
        &-\frac{1}{128}(U')^4-\frac{1}{16}U^6+\frac{1}{432}(U'')^3+\frac{1}{12}U^4U''+\frac{3}{32}U^3(U')^2-\frac{1}{8}U^3V^2-\frac{1}{32}U^2(U'')^2+\frac{1}{8}U^2(V')^2\\
        &-\frac{1}{12}U''(V')^2-\frac{1}{16}V^2(U')^2 + 
        \frac{5}{6}\eta\bigg(\frac{3}{2}U^2V^2+\frac{1}{8}UU''-\frac{1}{2}U(V')^2-\frac{1}{12}(U')^2U''-\frac{1}{2}V^2U''-\frac{7}{12}U^3U''\\
        &-\frac{3}{4}(U')^2U^2+\frac{1}{3}UU'U'''+\frac{1}{2}VU'V'+\frac{5}{8}U^5-\frac{1}{36}(U''')^2\bigg) + \frac{1}{2}\mu\left(U^2V+\frac{1}{3}U''V-U'V'\right)\\
        &+ \frac{\nu}{2}\left(V^2-\frac{1}{4}U'^2-\frac{1}{2}U^3+\frac{1}{3}UU''\right) - \frac{5}{3}\eta\mu UV + \frac{5}{3}\eta\nu\left(U^2-\frac{1}{3}U''\right) + \frac{25}{36}\eta^2\bigg(U^2U''+3UV^2-\frac{3}{2}U^4\\
        &-\frac{1}{6}(U'')^2-2(V')^2-8\mu V\bigg) - \frac{125}{18}\eta^3V -\frac{10}{9}\eta\mu^2-\frac{2}{3}\nu^2.
    \end{align*}
The associated set of Darboux coordinates are
        \begin{align}
            Q_1 :&= U - \frac{4}{3}\eta, \qquad\qquad\qquad\qquad\qquad\qquad\, Q_2 := V, \qquad\qquad  Q_3 := U', \label{Darboux-Q}\\
            P_1 :&= \frac{1}{4}\left( 3UU' - \frac{1}{3}U''' - \frac{7}{3}\eta U'\right), \qquad\quad\,\, P_2 := V', \qquad\quad\,\,\,\,\, P_3 := \frac{1}{12}U''-\frac{1}{6}\eta U + \frac{7}{18}\eta^2.\label{Darboux-P}
        \end{align}
These Hamiltonians pairwise commute, both with respect to the Poisson bracket induced by the canonical coordinates above, and in the sense that $\frac{\partial H_k}{\partial t_j} = \frac{\partial H_j}{\partial t_k}$, $k,j = 1,2,5$. One can thus define a $\boldsymbol{\tau}$-function via the formula
    \begin{equation}\label{tau-okamoto}
        {\bf d}\log \tau_{Okamoto} = H_1 dt_1 + H_2 dt_2 + H_5 dt_5.
    \end{equation}
As was shown in \cite{DHL2}, this $\boldsymbol{\tau}$-function is proportional to the (modified) isomonodromic $\boldsymbol{\tau}$-function, as is essentially defined in the
classic works of Jimbo, Miwa, and Ueno.

\section{Critical partition function for the quartic $1$-matrix model}

    An analogous statement to our main theorem regarding the partition function for the cubic $1$-matrix model was made in \cite{BleherDeano}. For completeness, we state here (with proof) the equivalent Proposition for the quartic $1$-matrix model. This result relies on the work established in \cite{DK0}, as well as the formula for the derivative of the partition function derived in \cite{BGM}. Consider the orthogonal polynomials
        \begin{equation}
            \int_{\Gamma} p_n(z,t;N)p_m(z,t;N)e^{-NW(z;t)}dz = h_n(t;N)\delta_{nm},
        \end{equation}
    where $W(z;t):=\frac{1}{2}z^2 + \frac{t}{4}z^4$,and $\Gamma$ is the contour starting from $e^{3\pi i/4}\cdot\infty$ and ending at $e^{-\pi i/4}\cdot \infty$. These polynomials satisfy the recurrence relation
        \begin{equation}
            zp_n(z,t;N)=p_{n+1}(z,t;N) + R_n(t;N)p_{n-1}(z,t;N).
        \end{equation}
    We first observe the following theorem, which is a slight extension of the result of \cite{DK0} (with the parameters $\alpha = \beta = 1$ there).
    \begin{theorem}
        Let $t$ be such that
        \begin{equation}
            n^{4/5}(t+1/12) = -c_1x,\qquad\qquad c_1 = 2^{-9/5}3^{-6/5},
        \end{equation}
        with $x\in \mathbb{R}$ not a pole of the tritronqu\'{e}e Painlev\'{e} I transcendent. Then,
        \begin{equation}
            R_n(t;n)\sim 2 - 2c_2 u(x)n^{-2/5}-\frac{c_2^2}{7}\left(\frac{17}{4}u(x)^2-\frac{43}{24}x+3u'(x)\mathcal{H}(x)\right)n^{-4/5} + \mathcal{O}(n^{-6/5}),
        \end{equation}
        where $c_2=2^{3/5}3^{2/5}$, $u(x)$ solves Painlev\'{e} I, and $\mathcal{H}(x)$ is the Painlev\'{e} I Hamiltonian.
    \end{theorem}
    
    We define the \textit{partition function} of this model by the formula
        \begin{equation}
            \log \mathcal{Z}_{n}(t;N) := \sum_{k=0}^{n-1} h_{k}(t;N),
        \end{equation}
    which is equal to the (analytic continuation of) the partition function for the corresponding quartic $1$-matrix model, up to an additive $t$-independent constant. Since the result we will state is in terms of the derivative in $t$ of this quantity, we can ignore this discrepancy, by a slight abuse of notations.
    Now, since we have the identity
        \begin{equation}
            \frac{\partial}{\partial t} \log \mathcal{Z}_{n}(t;N) dt = \bigg[\frac{N^2}{4}R_{n}(t;N)\left(\frac{n}{Nt} + R_{n+1}(t;N)R_{n-1}(t;N)\right)-\frac{n^2}{4t}\bigg]dt,
        \end{equation}
    we obtain the following corollary:
        \begin{theorem}
            Under the same scaling assumptions limit the previous Theorem,
                \begin{equation}
                    \frac{\partial}{\partial t} \log \mathcal{Z}_{n}(t;N) dt = \left[-\frac{1}{72}c_2^2n^{6/5}-\frac{1}{144}c_2^4xn^{2/5} -2\mathcal{H}(x) + \mathcal{O}(n^{-2/5})\right]dx.
                \end{equation}
            Consequently, if we define $t_c:=-\frac{1}{12}$, and
                \begin{align}
                    Z_{reg}(t) &:=\exp\bigg(-\left[F_0(t_c) + F_0'(t_c)(t-t_c) + \frac{1}{2}F_0''(t_c)(t-t_c)^2\right]\bigg)\nonumber\\
                    &= \exp\left(\frac{1}{24} + (t+1/12) - 18(t+1/12)^2\right),
                \end{align}
            where $F_0(t) = \lim_{n\to \infty} \frac{1}{n^2}\log Z_n(t;N)$, which is well-defined and differentiable for $t_c<t<0$, then
                \begin{equation}
                    \lim_{n\to \infty}{\bf d} \log \frac{Z_n(t;N)}{Z_{reg}(t)} = -2\mathcal{H}(x)dx.
                \end{equation} 
        \end{theorem}

\printbibliography

\end{document}